\documentclass[12pt]{article}

\usepackage{amsmath,amsfonts,amssymb,latexsym,graphicx,hyphenat}

\usepackage{caption}
\usepackage{subcaption}

\newtheorem{theorem}{Theorem}[section]
\newtheorem{corollary}{Corollary}[section]
\newenvironment{proof}[1][Proof]{\textsc{#1.} }{\ \rule{0.5em}{0.5em}}
\newtheorem{definition}{Definition}[section]
\newtheorem{lemma}{Lemma}[section]
\newtheorem{proposition}{Proposition}[section]
\newtheorem{remark}{Remark}[section]
\numberwithin{equation}{section}
\def\be{\begin{equation}}
\def\ee{\end{equation}}
\def\bq{\begin{eqnarray}}
\def\eq{\end{eqnarray}}
\def\beq{\begin{eqnarray}}
\def\eeq{\end{eqnarray}}

\begin{document}
\title{\textsc{Bifurcation diagrams for spacetime singularities and black holes}}
\author{{\Large{\textsc{Spiros Cotsakis}}$^{1,2,3,}$}\thanks{\texttt{skot@aegean.gr}}\\
$^{1}$Department of Applied Mathematics and Theoretical Physics \\ University  of Cambridge \\
Wilberforce Road, Cambridge CB3 0WA, United Kingdom \\ \\
$^{2}$Clare Hall, University of Cambridge, \\
Herschel Road, Cambridge CB3 9AL, United Kingdom\\  \\
$^{2}$Institute of Gravitation and Cosmology\\  RUDN University\\
ul. Miklukho-Maklaya 6, Moscow 117198, Russia}
\date{November  2023}
\maketitle
\newpage
\begin{abstract}
\noindent We reexamine the  focusing effect crucial to the theorems that predict the emergence of spacetime singularities and  various results in the general theory of black holes in general relativity. Our investigation incorporates the fully nonlinear and dispersive nature of the underlying equations. We introduce and thoroughly explore the concept of versal unfolding (topological normal form) within the framework of the Newman-Penrose-Raychaudhuri system, the convergence-vorticity equations (notably the first and third Sachs optical equations), and the Oppenheimer-Snyder equation governing exactly spherical collapse. The findings lead to a novel dynamical depiction of spacetime singularities and black holes, exposing their continuous transformations into new topological configurations guided by the bifurcation diagrams associated with these problems.
\end{abstract}
\newpage
\tableofcontents
\newpage
\section{Introduction}
A fundamental attribute of strong gravitational fields is the Hawking-Penrose prediction of spacetime singularities in the  gravitational collapse to a black hole and in cosmology (cf. standard papers \cite{pen65}-\cite{ha73}, and books \cite{he}-\cite{on}). The Hawking-Penrose analysis generalized the first mathematical model of a black hole by Oppenheimer and Snyder \cite{os},  was based on the focusing effect due to the Ricci curvature, and can be best described using the language of the causal structure of spacetime. These works contain results that predict the existence of singularities at the centre of black holes and in general cosmological models in the form of causal geodesic incompleteness, and offer a first evidence as to how spacetime may behave inside black holes, or near  cosmological singularities, e.g., the area theorem, trapped surfaces inside an event horizon, the caustics formed by the intersection of geodesics on approach to the singularity due to the focusing effect, etc.

The Einstein equations are not used in the Hawking-Penrose works except only indirectly through the energy conditions, and there only in order to obtain the focusing effect. (This effect was first noted in Refs. \cite{ll}-\cite{ray2}, but its central significance for general relativity was only clearly realized with the appearance of the Hawking-Penrose theorems on singularities and black holes.) Instead,  the main equations used by Hawking-Penrose for this purpose are the so-called Raychaudhuri equation that describes the rate of change of the expansion (or convergence) of the geodesic congruence,  the volume (or area) equation that describes the rate of change of volume (or  area) associated with the geodesic congruence, and additional equations that describe  changes in the shear and in the vorticity of the congruence. The shear equation is combined with the Raychaudhuri equation and together describe the rates of change of the convergence and the shear of the congruence in a set of equations called the Newman-Penrose-Raychaudhuri system (cf. e.g., \cite{ha73}), and the vorticity equation is also combined with the Raychaudhuri equation in a form which appears as a subsystem of the Sachs optical equations (cf. e.g., \cite{strau}), below we shall call this the `convergence-vorticity' system. (In fact, the convergence-vorticity system is not really used or needed in the derivation of the focusing effect.)

The deployment of the focusing effect in the proofs of the singularity theorems and other related results is very well-known, as is its use in the various theorems, in conjunction with other assumptions, in particular, the generic assumption and the energy condition. The combined use of the focusing effect with these physical or plausible assumptions leads to the singularity theorems and other basic results, the proofs of  which involve the methods  of causal structure in general relativity \cite{pen65}-\cite{ha73},  \cite{he}-\cite{on}.

However, the successful exploitation of the focusing effect and its use in combination with the energy and generic conditions in the proofs of the singularity theorems when working with the nonlinear equations such as those that describe the large-scale structure of spacetime, lead us to ask two more general questions about the basic approach to such equations:
\begin{enumerate}
  \item Given a nonlinear system of equations, how do we study the way an equation in the system interacts with or influences another?
  \item What is the relation between the structural (in-)stability of the nonlinear system \emph{itself} and the genericity or global  stability of its solutions?
\end{enumerate}
It is obvious that both  questions apply to the nonlinear systems used when studying spacetime singularities, and so both questions  become relevant in the present context.

For a dynamical system of the form $\dot{X}=F(X)$, $F$ being some smooth function of $X$,  the solution $X$ is generally speaking influenced by two factors,  an initial condition (datum) $X_0$, and the nonlinear `forcing' term $F(X)$.
The main issue is to understand the `feedback loop',  in which the solution $X$ influences the forcing term $F(X)$ which in turn influences the solution.

There are cases in which instead of looking at the full nonlinear system and nonlinear feedback effects, one is able to isolate and capture distinctive features in the behaviour of the problem by reducing the problem to a scalar equation. This usually becomes possible through the use of physical assumptions or special structures present in the original system, and using those one may end up with a \emph{linear} feedback effect that may provide a viable approach. In this way, we may be studying the full nonlinear feedback effect by acquiring control of only the linear part of it.

In fact, this is a viable method when dealing with essentially  nonlinear systems for which it is difficult to separate what the effects of the linear and nonlinear feedback on the solution really are. This is a standard way of approach,  particularly in the class of \emph{dispersive} problems, that is those described by equations that share some sort of degeneracy or instability, cf. \cite{tao}.

Let us now move on to a brief discussion of the second question. Gravitating systems describing instabilities such as those studied in this work, are all described by structurally unstable systems of equations. This raises the question of what exactly one means by the word `generic' for a structurally unstable system because of the following reason. In the space of all vector fields the non-generic ones can be thought to lie on a hypersurface of some finite codimension\footnote{To study problems with degeneracies of infinite codimension is also possible, but in this work all systems have finite, and in general a small codimension.}, with the generic systems occupying the complement - the non-generic systems lie on the boundaries of the generic domains (cf. \cite{ar83} for a detailed discussion). A small perturbation of a non-generic system will then  take it off that hypersurface  to the domain of the generic ones. This is perhaps the main reason why under normal circumstances one's attention is driven away of non-generic systems and focuses almost exclusively to the generic ones.

However, consider the transversal intersection (i.e., at nonzero angle) of a curve (i.e., a 1-parameter family) of systems with the non-generic boundary surface. Under a small perturbation, this family will again intersect that surface at some nearby point, and so although a single non-generic system can be made generic by perturbation, it is not possible to achieve this with all members of a family. In general, it is not possible to remove degeneracies of codimension not exceeding $k$ in $k$-dimensional families, but all degeneracies of higher codimension are removable in such families. This argument shows that the natural object to study is not the original vector field but the one that has the right codimension, so that its degeneracies do not disappear upon perturbation. Objects with the `right' codimension can be constructed starting from some degenerate one, using the subtle rules of bifurcation and singularity theory (cf. e.g., \cite{ar83}, \cite{ar94}, \cite{ar72}).

In this work, we take up this problem  for the systems involved in the original analysis of Hawking-Penrose that led to the singularity theorems and black holes. In a sense, in this work we provide an answer to the problem posed in the book \cite{he}, p. 363\footnote{In the coming decades since this sentence was written, catastrophe theory was eventually taken to imply a general term describing possible applications of bifurcation theory and singularity theory (by the latter we mean the singularity theory of functions, cf. e.g., \cite{ar83}, \cite{ar94}). In fact, we shall not refer to `catastrophe theory', but use instead `bifurcation theory' as a general term that encompasses all three.}:
\begin{quote}
\emph{  ... It may also be that there is some connection between the singularities studied in General Relativity and those studied in other branches of physics (cf. for instance, Thom's theory of elementary catastrophes (1969)) ...}
\end{quote}
To be more precise, we shall provide a complete analysis based on bifurcation theory of the following three systems:
\begin{enumerate}
  \item The Newman-Penrose-Raychaudhuri system
  \item The convergence-vorticity system
  \item The Oppenheimer-Snyder system.
\end{enumerate}
It is a remarkable fact that as seen from the present perspective, the original analyses  by S. W. Hawking and R. Penrose   constitute the first ever bifurcation calculation and analysis in general relativity. In particular, their treatment of the focusing effect (through their employment of the energy and generic conditions and  subsequent applications to the study of singularities and black holes) exactly corresponds to an analysis of the \emph{versal unfolding associated with a codimension-1  reduction of the full Newman-Penrose-Raychaudhuri system}.  From this point of view, the results discovered in the original papers \cite{pen65}-\cite{ha73} (and subsequently described in various  sources such as \cite{he}-\cite{on}) provide the appropriate basis for the analysis  performed in this work.

The plan of this paper is as follows. In the next Section, we offer a guide for the reader about  the most important results of the subsequent sections. Section 3 is a summary of some of the basic ideas of bifurcation and singularity theory, which form the basis of our subsequent developments. In Section 4, we present a review of the focusing effect,  introduce the idea of a bifurcation theory approach for  spacetime singularities and black holes, and examine how the Hawking-Penrose pioneering  analysis is closely related to bifurcation theory and the feedback loop problem. In Section 5-7, the bifurcation treatment of the three main systems mentioned above is fully developed. In Section 8, we present some first applications of our results to the problem of singularities in general relativity, only with the purpose of providing a few examples of the possible breadth and probable importance that a bifurcation theory  approach has to offer to the problem of the nature of spacetime singularities, black holes, and related issues. Some extra discussion is also given in the last Section of this paper.

\section{Summary of the main results of this paper}
In this Section, we provide a brief summary of some of the results in subsequent Sections.

In the next Section, we develop some bifurcation theory ideas with a view to their subsequent applications in later Sections. The main purpose is to acquaint the reader with the symbolic sequence (\ref{bifn-method}) which describes a basic message of bifurcation theory. Namely, starting with a system which has degeneracies (as in  the `original system' in (\ref{bifn-method})), the way to study these through bifurcation theory  is to first obtain the normal form of the original system. This is usually a different (or topologically inequivalent) system than that we started with.

The principal reason to find the normal form of the original system, and not work directly with the latter, is because the structure of the nonlinear terms that affect the solutions of a \emph{degenerate}  nonlinear system is determined by its linear part, and such crucial nonlinear terms may not be fully present in the original form in which the system is given. The normal form procedure is described in some detail in subsection \ref{normal}, whereas in the first part of Section \ref{bif-th}, we introduce basic ideas of the Poincar\'e program for bifurcation theory: structural (in-)stability, stability of solutions and perturbations of unstable systems,  the idea of genericity, types of degeneracies present in such systems and, finally, the bifurcation diagram. The most important idea in this Section is that of a versal unfolding,  treated  in Section \ref{versal}, and the closely related notions of stratification and moduli. Both of these are crucial for the construction of the bifurcation diagram.

In Section \ref{r-w}, we introduce the three main systems mentioned above,  and in Section \ref{stand} we briefly review the standard argument for the focusing effect and how it leads to the global theorems about the structure of singularities and black holes, before we embark on the bifurcation theory approach to this problem in Section \ref{bifn-1}. In this latter Section, we show that the focusing effect corresponds to the linear part of the feedback loop for the NPR-system, and also show how the original Hawking-Penrose treatment of it closely resembles the modern approach employed in this work. In addition, we discuss how the original analysis of Hawking-Penrose  clearly points to the need for consideration of  nonlinear feedback effects, and we provide a description of what such an analysis would entail.

In Sections \ref{bif}-\ref{bif3}, we provide a detailed bifurcation analysis of each one of the three systems mentioned earlier. This analysis is performed in a number of different steps, but the main results are presented in a concise form in  four bifurcation diagrams given in the following figures: Fig. \ref{bifn} for the NPR-system, Fig. \ref{bifnOmega} for the convergence-vorticity system, and Figs. \ref{oppie1}, \ref{oppie2} for the Oppenheimer-Snyder system.

For these diagrams we make the following remarks with a purpose of making their understanding somewhat smoother. Firstly, there are certain structures common to all four, namely, the existence of a central, `parameter diagram', which is stratified in `subregions', and secondly, the placement of corresponding phase portraits in each one of them. We imagine that as the parameter point moves in any of the parameter planes in the four bifurcation diagrams, the corresponding phase portraits smoothly deform to one another, producing the famous `metamorphoses' (or `perestroikas' in other terminology) of bifurcation theory, here, however,  in a gravitational context. Some of these phenomena are briefly discussed in Section 8 of this work, and some extra comments are also given in the last Section.

For the reader who has some acquaintance with the basic terminology of bifurcation theory and with standard results from the theory of global spacetime structure, one way to obtain a quicker summary of this work  is this: after a review of the three basic systems in Section \ref{r-w}, read through Section \ref{bifn-1}, and then have a look at the four bifurcation diagrams in the Figures \ref{bifn}, \ref{bifnOmega}, \ref{oppie1}, and \ref{oppie2}. An introduction to the main metamorphoses of singularities and black holes is then given in Section 8. The  work for all the proofs of the main statements and constructions in this paper is presented (with some brevity!) in Sections 5-7.

\section{Bifurcation theory: degeneracy, instability, and versality}\label{bif-th}
In Section \ref{gen}, we discuss general aspects  of bifurcation theory such as   the idea of instability as it emerges in the study of structurally unstable systems, genericity and degeneracy, and an overview of Poincar\'e's program to study these issues. In Section \ref{normal}, we discuss the normal form theorem, which leads to a first familiarity with certain novel fundamental dynamical aspects of the three main systems studied later in this work. In the Section \ref{versal}, a further discussion is given of more advanced material from singularity and bifurcations. This material is about the ideas of codimension and stability of bifurcating families, unfoldings, and versality in general. Last, we include a discussion of the bifurcation diagram, the cornerstone of any analysis of degenerate problems.

\subsection{General remarks on bifurcations}\label{gen}
\subsubsection{Intuitive discussion}
It has been said that bifurcation theory describes the behaviour of solutions of a dynamical system as the parameters of the system change. This is of course true, and that  is perhaps a standard definition of the subject. In bifurcation theory problems, one  always ends up studying a dynamical system which depends on one or more parameters, and observes how the behaviour and/or number of solutions change as the parameters of the system pass through some `bifurcation set' (cf.  standard references of this subject, e.g.,  \cite{ar83}-\cite{thom}).

However, this definition  may give  the misleading impression that bifurcation theory enters the scene only when some parameter is present in the problem. On the contrary,  bifurcation theory is the only mathematical field solely devoted to the study of \emph{instabilities}. From the growth of a population to the saddle-node bifurcation, from the simple harmonic oscillator to the Hopf bifurcation,  from the pitchfork bifurcation to the Lorenz system, or in the stable versal  families of diverse degenerate unstable systems, one gradually becomes acquainted with the unfamiliar but fundamental fact that \emph{to correctly account for unstable phenomena one has to extend, or `unfold', the original system describing them just so much as to reach a stable parametric family, without at the same time removing the defining degeneracies of the original system}.

To properly  perform this extension and  fully study the `unfolded dynamics' of the resulting parametric families, represents the glorious mathematical developments of bifurcation, catastrophe, and singularity theory over a period of more than a century.

\subsubsection{Stable and unstable systems}
Before we proceed further, we briefly discuss the difference between stable and unstable systems.

It is a central lesson of bifurcation theory that, given an unstable or special solution, it is  inadequate to perturb \emph{only itself} in order to see if it stabilizes. Ideally, and perhaps more importantly, one needs to perturb \emph{the system itself} to a point where a stable family of \emph{systems} containing the original one is reached.

In the space of all dynamical systems, we have structurally stable and structurally unstable systems. A structurally stable system is one whose behaviour can be deduced from that of its linearization, and as such it has, for example, only hyperbolic fixed points. If a system is not structurally stable, it is called a structurally unstable, dispersive, or bifurcating system (we shall avoid the finer differences that exist in the meanings of these three terms and consider them as synonymous).

There was indeed a time during the sixties and the seventies when many people were led to believe that only structurally stable systems  are important, or more common and  abundant, and called such systems `generic' meaning typical or  retaining their form and properties under perturbations\footnote{It is an interesting historical fact that the first book on the use and importance of bifurcation theory in science (in that case biology) by R. Thom \cite{thom} in 1972, had the title `Structural Stability and Morphogenesis', even though it studies the different ways that structure may emerge from changes of different forms that may arise in unstable systems. In that book, the foundations for a bifurcation theory approach to all of science were discussed  in both scientific and philosophical terms, and the fundamental idea of structural stability of \emph{families} was laid down for the first time.}.
This led to the general tendency to distinguish or `prefer' the structurally stable systems  from the non-generic or physically implausible ones that represented special cases, and so devoid of any physical significance. This had the unfortunate consequence in some cases to totally neglect the latter as being unimportant.

\subsubsection{Genericity and degeneracy}
The development of bifurcation theory (and also its sister field `singularity theory') in the last half-century or so has shown that an approach based only on individual `generic' or structurally stable systems is rather naive, if not totally wrong. It is important to clarify first of all whether or not the given  system at hand is structurally unstable and if yes, its exact type of `degeneracy', because otherwise there is a real danger to treat such a system as stable one when it is not.
In fact, an individual  structurally stable nonlinear system is in a sense uninteresting because its behaviour is essentially linear, and so nonlinearities do not offer anything new.

Secondly, it has become apparent that various kinds of degeneracies,  such as zero eigenvalues, are the rule rather than an exception  in nonlinear systems, and therefore cannot really be avoided for reasons of convenience or `simplicity'. In turn, structurally unstable systems  appear everywhere\footnote{This constitutes  a kind of paradox (an `unstable trap' so to speak) associated with structurally unstable systems: since they usually appear as solitary, individual curiosities, they can be easily mixed up with uninteresting  systems of no physical importance.} and, although they can \emph{individually} be perturbed to stable ones, this cannot be done at all for unstable \emph{families} of systems: If for some value of the parameter present in a family, one perturbs the resulting unstable system to a structurally stable one, then the degeneracy and non-genericity are avoided \emph{for that parameter value} but appear again for another. It is thus impossible to perturb an unstable \emph{family} of systems into a stable one for \emph{all} values of the parameters present in the system simultaneously.

For these reasons, we shall only focus on structurally unstable, `non-generic' systems. As we discussed above, such systems become unavoidable when considered in the context of parametrized families.

\subsubsection{Poincar\'e's program}
The approach of bifurcation theory to the study of dynamical systems that describe unstable phenomena consists of three  steps.
\begin{enumerate}
\item \textbf{Normal form theory} (this is the `static' part): Given an unstable system (we shall only deal below with vector fields), put it in a  `simplified' form using normal form theory: By a coordinate transfomation\footnote{These transformations are those of the unknown functions and their derivatives as these enter in the `field equations' of the problem, and have nothing to do with the coordinate transformations usually considered in general relativity. They represent the coordinates in the phase space of the given problem.}, rewrite it in a way that exhibits only the `unremovable' terms at each level in a perturbation expansion. Sometimes this leads to a new form of the system, where many (perhaps all) terms at a given order may be absent (as they can be eliminated). Of course, as we shall see, this merely indicates the need for the consideration of higher-order terms.
\item \textbf{Singularity theory} (this is the `kinematic' part):  Find all possible (topological) extensions, or unfoldings,  of the normal form system that was obtained in the previous step. In some cases, one is able to reach a universal form containing all possible such extensions, the `versal unfolding'. Here one introduces various kinds of parameters, called `modular' and `standard' parameters respectively,  as dictated  by the nature of the problem, and the determinancy of the degenerate vector field (in general, the determinancy of the vector field is not equivalent to that of the unfolding)\footnote{We note that while determinancy is a highly non-trivial process for a vector field with some degeneracy (and we need to include higher-order terms), it is trivial for a structurally stable vector field, as the latter is completely determined by the jacobian of its linear part as per the Grobman-Hartman theorem.}.
\item \textbf{Bifurcation theory} (this is the `dynamical' part)\footnote{Below we shall use the word `bifurcation', perhaps somewhat degenerately, to cover all three steps of the analysis.}:   Study the dynamics of the unfoldings and construct the \emph{bifurcation diagram}. The unfoldings respect the symmetries and other characteristics of the original system, and in the case of a versal unfolding, contain all possible \emph{forms of instability} that the original system may exhibit - they are stable with respect to any perturbation. In a sense, the versal unfolding determines the bifurcation diagram completely. The latter contains  all possible phase portraits and  possible parameter regions,  gives the overall and complete behaviour of any perturbation associated with the original system, and most  importantly it describes all metamorphoses of the phase portraits of the system\footnote{We note that singularity theory may be described as one where only metamorphoses of \emph{equilibria}, but not phase portraits, can be given.}.
\end{enumerate}

This is a far-reaching generalization and refinement of the original approach to physical science.
In essence, changing the parameters `kinematically' in the resulting families is the way to completely describe the possible instabilities of the system without going outside the family - a new form of structural stability, this time referring to \emph{families}.

One thus achieves a major goal,   to arrive at a (or, perhaps better, `the') global picture of all instabilities,  how they are all related to each other via  their \emph{metamorphoses} - smooth changes in the phase portraits. This is essentially \emph{Poincar\'{e}'s program} for bifurcation theory, which aims to discover all possible forms of behaviour of unstable systems in a self-consistent, systematic way.

Up until the present day,  this program is far from being completed, despite the very substantial progress by many mathematicians over a period of more than 100 years.

One central idea in bifurcation theory is the \emph{global bifurcation diagram}. This is  a set of distinct (topologically inequivalent) diagrams each having the following structure: a set of qualitatively different phase portraits corresponding to different regions of the \emph{parameter diagram} of the system. In fact, constructing the bifurcation diagram of a given dynamical system is the key step in understanding all possible dynamical behaviours associated with the system as well as those of all dynamical systems that lie near it (in a suitable sense), and describing all stable perturbations of it.

\subsection{Normal forms}\label{normal}
To give a more precise discussion of the bifurcation diagram, we need to introduce some standard terminology from bifurcation and singularity theory (see \cite{cot23}, Section 3, for an introductory discussion of more foundational material on bifurcation theory not discussed here).

We consider a dynamical system,
\be \label{ds1}
\dot{w}=G(w),\quad w\in\mathbb{R}^n,
\ee
where $G$ is a $\mathbf{C}^r$ function on some open subset of $\mathbb{R}^n$, and suppose that (\ref{ds1})
has a non-hyperbolic fixed point at $w_0$. Although this system may depend on a vector parameter $\epsilon\in\mathbb{R}^p$, and the non-hyperbolic fixed point be at $(w,\epsilon)=(w_0,\epsilon_0)$, we shall in fact forget about parameter-dependence for the moment. In addition, although our discussion holds for $n$-dimensional systems, for concreteness we shall restrict our development to planar systems, i.e., we shall consider only consider the case $n=2$.

For the present purposes, we shall only consider the case where the linearized Jacobian evaluated at $w_0$,  $A=D_w G(w_0)$ (which enters in the linear system $\dot\xi=A\xi$) has a double-zero eigenvalue, and the Jordan normal form of the linear part of (\ref{ds1}) has been found. This means that we can introduce the linear transformation $v=w-w_0$ and transfer $w_0$ to the origin, so that (\ref{ds1}) becomes a system of the form $\dot{v}=H(v), H(v)=G(v+ w_0)$.  We can then split the system  into a linear and a nonlinear part, $\dot{v}=DH(0)v+\bar{H}(v)$ and using the eigenvector matrix $T$ of $DH(0)$, we can simplify the system and write its linear part in Jordan canonical form $J$ under the transformation $v=TX$, so that the full nonlinear system will be written as,
\be\label{ds2}
\dot{X}=JX+F(X),
\ee
where $J=T^{-1}DH(0)T$, and $F(X)=T^{-1}\bar{H}(TX)$. This is a `normal form' of the system, in which only the linear part $DH(0)$ has been simplified as much as possible.
We shall  assume that the Jordan form $J$ has either the `cusp' (or, Bogdanov-Takens) form,
\be\label{bt}
J|_{(0,0)}=\left(
  \begin{array}{cc}
    0 & 1 \\
    0 & 0\\
  \end{array}
\right),
\ee
or else,  $J$ is the zero matrix,
\be\label{zero}
J|_{(0,0)}=\left(
  \begin{array}{cc}
    0 & 0 \\
    0 & 0\\
  \end{array}
\right).
\ee
In the last case, we shall assume that the system (\ref{ds2}) is invariant under the $\mathbb{Z}_2$-symmetry (a particular case of equivariance), that is if, $X=(x,y), F(X)=(f(x,y),g(x,y))$, the system $\dot{x}=f,\dot{y}=g$, is invariant under the transformation,
\be
x\to x,\quad y\to-y.
\ee

We shall show later that the NPR and convergence-vorticity systems are $\mathbb{Z}_2$-equivariant, while the Oppenheimer-Snyder system has a linear part that  is of the Bogdanov-Takens form.

Because of the non-hyperbolicity of the origin, the flow near the origin is not topologically conjugate to that of its linearization, and so the flow will be sensitive to nonlinear perturbations. Therefore for the given dynamical system (\ref{ds1}) written in the form (\ref{ds2}), the fundamental problem arises of \emph{how to fully describe the flow}.

This problem is further perplexed because the system (\ref{ds1}) (or (\ref{ds2})) will in this case be subject to certain \emph{degeneracy conditions} at various \emph{levels} (i.e., orders in a Taylor expansion  of the $X$), and these will lead to further important terms that will appear by necessity in the original system. This problem can be accounted for  through the construction of the so-called \emph{Poincar\'e  normal form} of the original system (\ref{ds2}) as in the following theorem, which simplifies the nonlinear part $F(X)$ at each order.
\begin{theorem}[Normal Form Theorem]
Under a sequence of analytic changes $X=Y+h_k (Y)$ of the coordinate $X$, the system (\ref{ds2}) takes the form,
\be\label{ds3}
\dot{Y}=JY+\sum_{k=2}^N F_k(Y) +O(|T|^{N+1}),
\ee
where the unknowns $h_k (Y)$ satisfy  the equation,
\be\label{hom}
L_J^{(k)}(h_k(Y))=F_k(Y),\quad L_J^{(k)}(h_k(Y))=Dh_k(Y)JY-Jh_k(Y),
\ee
at each order $k$.
\end{theorem}
Equation (\ref{ds3}) is called the \emph{normal form} of (\ref{ds2}) at order $N$. Equation (\ref{hom}) is known as the \emph{homological equation} associated with the linear vector field $JY$. If the operator $L_J^{(k)}$ is invertible, then $h_k(Y)$ can be chosen so that $h_k(Y)=(L_{J}^{(k)})^{-1}F_k(Y)$, and so all terms $F_k(Y)$ in (\ref{ds3}) can be eliminated leaving only the linear system $\dot{Y}=JY+O(|T|^{N+1})$. Of course this rarely happens, and there will be extra \emph{resonant terms} remaining in the normal form (\ref{ds3}) of the system (\ref{ds2}). The terms that can be eliminated at each step are called \emph{nonresonant}.

At each order $k$, one views the terms $h_k(Y), F_k(Y)$ as belonging to the linear space of vector-valued homogenous polynomials of order $k$, denoted here by $H_k$. For instance, for $k=2$ and in $\mathbb{R}^2$, this space is spanned by the products of the monomials $x^2, xy, y^2$ times the basis vectors of $\mathbb{R}^2$, and $H_2$ can be represented by the direct sum,
\be
H_2=L_J^{(2)}(H_2)\oplus G_2,
\ee
with the last term being a complementary space to $L_J^{(2)}(H_2)$ that contains all those terms $F_2^r(Y)$ (`$r$ stands for `resonant') that cannot be in  the range of $L_J^{(2)}(H_2)$, and hence cannot be removed. All other terms can be eliminated, except such resonant terms of the form $F_2^r(Y)$. So at each order, the eigenvectors of $L_J^{(k)}$ will form a basis for $H_k$, while the eigenvectors of $L_J^{(k)}$ having non-zero eigenvalues will form a basis of the image $L_J^{(k)}(H_k)$. The components of $F_k(Y)$ in $L_J^{(k)}(H_k)$  can be expessed in terms of such eigenvectors and so  can be eliminated. Hence, the terms that remain in the transformed vector field Eq. (\ref{ds3}) will be of the form $F_k(Y)$ that cannot be written as linear combinations of the eigenvectors of $L_J^{(k)}$ having non-zero eigenvalues.

The important thing that the normal form gives us is that the structure of the nonlinear \emph{remaining} terms will be entirely determined by the Jordan matrix $J$, and also that simplifying (or eliminating) the terms at a given order $k$ will not alter the lower-order terms. However, higher-order terms will be modified at each step of the method of normal forms. Eventually a simplified vector field will be the result instead of (\ref{ds2}) at some given order.

\subsection{Versal unfolding}\label{versal}
Returning to the general bifurcation problem, the normal form (\ref{ds3}) obtained this way will still be unstable with respect to different perturbations, that is with respect to nearby systems (vector fields), and so its flow (or that of the original system) will not be fully determined this way. It is here that bifurcation theory makes its entry, in that it uses the normal form to construct a new system, the \emph{(uni-)versal unfolding}, that is based on the normal form but \emph{contains the right number of parameters needed to take into full account the degeneracies} of the normal form system (and so also of the original one). The necesary number of parameters needed to take into full account the nature of the degeneracy of the normal form is called the \emph{codimension of the bifurcation}. We shall be dealing in this work only with codimension-2 problems, that is those which can be fully unfolded using two independent parameters.

Once one knows the versal unfolding of a particular system, then any perturbation of the system will be realized in the versal unfolding for some particular choice of the parameters. Therefore studying the dynamics of the versal unfolding instead of that of the original system (or its normal form) implies that we have a complete knowledge of the behaviour of all possible perturbations of it.
Hence, bifurcation theory suggests that in order to study and fully understand the behaviour of a degenerate system, one proceeds in the direction:
\be\label{bifn-method}
\textrm{Original system}\to\textrm{normal form}\to\textrm{versal unfolding}\to\textrm{dynamics},
\ee
and one eventually studies the dynamics of the versal unfolding rather than that of the original system (or its normal form)\footnote{We note the important remark not often stressed enough, that the dynamics (e.g., phase portraits) of a given  system and that of its normal form are generally \emph{inequivalent}. One aspect of bifurcation theory, in particular the versal unfolding construction, that is particularly important in this respect is that it is not really relevant at the end whether any of the two dynamical situations (original or normal form system) is the correct one. This is so because on the one hand,  the phase dynamics of the normal form corresponds to the `stratum' at the origin in the versal unfolding, while on the other hand, certain features of the original system dynamics (assuming that the original system is not already in normal form), appear as `scattered'  in various strata in the final bifurcation diagram.}.

Since the unfolded system by construction contains parameters, instead of just ending up  with a single phase portrait for this purpose, one is required to study the global bifurcation diagram of the versal unfolding which contains:
\begin{enumerate}
\item the modular coefficients
\item the parameter diagram
\item the various phase portraits.
\end{enumerate}
Let us briefly explain these terms. Suppose we have constructed the versal unfolding starting from a normal form system (corresponding to the original equation). In this example, this will be a system of the form,
\be \label{vaS}
\dot{X}=F(X,\mu,s),
\ee
where $X$ is polynomial in $x,y$, having a non-hyperbolic equilibrium at $X=0$, $s$ will denote the set of values of coefficients appearing in front of certain terms of the polynomial $X$,  and $\mu=(\mu_1,\mu_2)$ will be two parameters in the versal unfolding (corresponding to a codimension-2 bifurcation). We shall assume for simplicity that the modular coefficient $s$ only appears taking two distinct integer values (the `moduli')  in just one of the terms of the vector polynomial $X$. In this case, the two values of the $s=\pm 1$ will lead to two different versions of the versal unfolding, one corresponding to $s=+1$, and a second corresponding to $s=-1$. Thus for each moduli, we obtain a \emph{version} of the versal unfolding that can be analysed separately.

We now show that the parameter plane $(\mu_1,\mu_2)$ can be \emph{stratified}. For a fix moduli value (and so given a particular version of versal unfolding), we  take a parameter value $\mu=\mu_0$ and consider all points in the plane $(\mu_1,\mu_2)$ for which the system (\ref{vaS}) has phase portraits topologically equivalent to that which corresponds to $\mu_0$. This point set is called a stratum in the parameter plane, and all such strata make up the parametric portrait of (\ref{vaS}). The parameter plane is thus partitioned into different strata.

This means that for a fixed moduli value the parameter plane provides a stratification of the parameter space induced by topological equivalence. For each stratum in a given stratification, we have a  phase portrait, and the total number of phase portraits thus constructed together with the parameter plane give us the bifurcation diagram. The global bifurcation diagram is the set of all so constructed bifurcation diagrams, and provides a complete picture of the dynamics of the versal unfolding.

We note that a versal unfolding as a \emph{family} of systems parametrized by the parameter $\mu$ is a \emph{structurally stable} family (cf. \cite{ar83}), and as such it contains all physically relevant perturbations of the original system.

\section{A bifurcation theory approach to spacetime singularities}\label{bifn+sings}
In Section \ref{r-w}, we write down the precise forms of the  dynamical systems involved in the three main problems studied in this paper. In Section \ref{stand} we give a short summary of the focusing effect for causal geodesic congruences, and indicate how this is used in the proofs of the singularity theorems and related results. In Section \ref{bifn-1}, we relate the dynamics of the focusing state with more general dynamical issues such as  the linear feedback loop and adversarial behaviour, and highlight how the original calculations leading to the focusing effect constitute  a form of versal analysis for the Raychaudhuri equation. Lastly, in Section \ref{lin}, we set the stage for the consideration of further effects which are of an essentially nonlinear character associated with the problem of taking into account all stable perturbations of these problems.

\subsection{The three basic systems}\label{r-w}
\noindent We consider a timelike or null congruence of geodesics in spacetime, and denote by $\theta$ the trace of the extrinsic curvature, also called the \emph{expansion} of the congruence, $\rho=-\theta$ is the \emph{convergence} of the congruence, and $v$ is the 3-volume form (or the 2-area element in the case of a null congruence) of a positive-definite metric  on the spacelike 3-surface (correspondingly, 2-surface). An overdot denotes derivatives with respect to proper time (or an affine parameter if a null congruence is considered),  $2\sigma^2$ stands for $\sigma_{ab}\sigma^{ab}$, $\sigma_{ab}$ being the shear tensor, while $2\omega^2=\omega_{ab}\omega^{ab}$, and $\omega_{ab}$ is the rotation (vorticity) tensor.

We set $\mathcal{R}=R_{ab}\,t^a t^b$, where $R_{ab}$ is the Ricci curvature and $t^a$  a timelike vector field tangent to the congruence. To describe especially  the null case, it is standard to introduce a null tetrad $l,m,n,\bar{m}$, that is  $l^a$ is tangent to the null congruence, $n^a$ is null such that $l^a n_a=1$, and $m^a,\bar{m}^a$ are also null vectors orthogonal to $l^a$, and satisfy $m^a\bar{m}_a=-1$, with $m^a$ being  a complex combination of two spacelike vectors orthogonal to $l^a,n^a$. We then set  $\mathcal{W}=C_{abcd}\,l^a m^b l^c m^d$, with $C_{abcd}$ the Weyl curvature (we still use the letter $\mathcal{R}$ to denote $R_{ab}\,l^a l^b$).

Standard (but nonuniform!) conventions apply, and these together with other properties  can be found in the general references \cite{pen65}-\cite{on}, whose notation and proofs we  generally assume in this work.

\subsubsection{The Newman-Penrose-Raychaudhuri dynamical system}

The first problem we shall study requires a \textbf{hypersurface-orthogonal congruence}, where $\sigma^2\neq 0, \omega^2=0$. In this `zero-rotation' case, we shall be concerned with the global structure of solutions of the \textbf{Newman-Penrose-Raychaudhuri} (in short `NPR') system:
\begin{equation}\label{1}
\begin{split}
\dot{\rho}&=\rho^2+\sigma^2 +\mathcal{R},\\
\dot{\sigma}&=n\rho\,\sigma +\mathcal{W},
\end{split}
\end{equation}
where $n=2,3$, according to whether the congruence is null or timelike. The terms $\mathcal{R}$ (resp. $\mathcal{W}$) in (\ref{1}) represent matter (resp. gravitational radiation) crossing the congruence transversally. We assume that $\rho,\sigma$ are real.

We shall also use the \emph{definition},
\be\label{2}
\theta=(\log v)^\centerdot,
\ee
where $v$ is proportional to the volume (area) element of the  hypersurface orthogonal to the timelike (null) congruence.

Another common form of (\ref{1}) is  obtained by changing to $\rho\rightarrow -\theta$ (or $t\rightarrow -t$), to obtain its \emph{past} version, \emph{dynamically equivalent} to (\ref{1}). We do not discuss it  further here, because exactly the same conclusions will apply (we note that this equivalence is a technical term as in Section \ref{bif-th}).

We omit the derivation of (\ref{1}),   as it is discussed in great detail in the standard references given above. Indeed, (\ref{1}) may also be viewed as a subsystem of the  \emph{Sachs optical equations}, given by the vector field $(\rho^2+\sigma^2-\omega^2 +\mathcal{R},n\rho\,\sigma +\mathcal{W},n\rho\,\omega)$. A derivation of the Sachs equations can be found  e.g., in \cite{he} (cf.  Eqns. (4.22), (4.26), (4.27) for the timelike, and (4.34-6) for the null cases, respectively), in \cite{wald}, Sect. 9.2,  in \cite{strau}, pp. 500-1, or in \cite{on}, chap. 5.

\subsubsection{The convergence-vorticity dynamical system}
The second problem relates to the full dynamical  description of a pure \textbf{vorticity congruence}, that is one for which the shear is zero, $\sigma^2=0$, but $\omega^2$ never vanishes.
The term $\mathcal{R}$ again represents matter  crossing the congruence transversally. This vorticity, shearfree (or, `type-D') case is described by the  following \textbf{convergence-vorticity} (in short `CV') system:
\begin{equation}\label{v}
\begin{split}
\dot{\rho}&=\rho^2-\omega^2 +\mathcal{R},\\
\dot{\omega}&=n\rho\,\omega ,
\end{split}
\end{equation}
where $n=2,3$, according to whether the congruence is null or timelike.

We shall also use the definition (\ref{2}), and refer to the dynamical system (\ref{v}) as the \emph{future} version of the convergence-vorticity system. Changing to $\rho\rightarrow -\theta$ (or $t\rightarrow -t$), we obtain its \emph{past} version, \emph{dynamically equivalent} to (\ref{v}). Exactly the same conclusions will apply to the past version.

We omit the derivation of (\ref{v}),  as it is discussed in great detail in the standard references given above. We shall see later that the NPR- and CV-systems  (\ref{1}), (\ref{v}) respectively, are very closely related dynamically.

\subsubsection{The Oppenheimer-Snyder example}
The third problem to be analysed in this work is the original Oppenheimer-Snyder equation, namely,
\be\label{os1}
\ddot{x}+\frac{3}{4}\dot{x}^2=0,
\ee
(cf.  Eqn. (20) in \cite{os}), which describes the gravitational collapse of a dustlike sphere.  The geometric setup is very standard and goes as follows.
We introduce the Schwarzschild metric in comoving coordinates, $ds^2=d\tau^2-e^\lambda dR^2-e^x d\sigma^2$, where $\tau, R$ are the time and radial coordinates respectively, $e^x=r^2$, with $r$ the `radius', and $d\sigma^2$ is the metric of the unit 2-sphere (we use $x$ in the place of the Oppenheimer-Snyder function $\omega$ to avoid confusion with the vorticity function introduced above).

In \cite{os}, it is shown that in this case the Einstein equations reduce to the equation (\ref{os1}), (cf. Eqns. (13)-(20) in \cite{os}, see also \cite{ll}, Section 100, Problem 5 on p. 304, Section 103). This leads to the following \emph{solution} of the field equations (cf. \cite{os}, Eq. (21)): $e^x=(F\tau +G)^{4/3}$, where $F,G$ are arbitrary functions of $R$, so that, $r=(F\tau +G)^{2/3}$ (cf. \cite{os}, Eq. 27).
Using this solution, the standard result of \cite{os}, namely, their Eq. (37) is obtained, describing the optical disconnection with the exterior spacetime and the formation of a singularity at the centre of the black hole in a  finite time (see also \cite{ll}, Section 103).

\subsection{The standard argument for the focusing state}\label{stand}
The pioneering arguments generalizing the Oppenheimer-Snyder example and leading to the focusing effect and spacetime singularities in general relativity were obtained using the system (\ref{1}). As it is well-known, these arguments were deployed in the standard works and led to the existence theorems for spacetime singularities in general relativity.  A brief summary will be given here in several steps (all definitions, proofs, and constructions in this Section can be found in the standard references \cite{pen65}-\cite{strau}, and so we do not cite them below). Our review of these results, however, aims to relate them with certain central ideas of bifurcation theory,  and although very analogous, the methods of proof for the timelike and null cases presented here in subsections \ref{t-foc}, \ref{n-foc} will in this respect be useful to us later.

The focusing effect plays a central role in the proofs of the singularity theorems for gravitational collapse and  cosmology, and also in the proofs of the area law and other fundamental properties of black holes. In general relativity, this effect emerges when the convergence $\rho$ of a congruence of causal geodesics becomes infinite. Because of the definition in Eq. (\ref{2}), this happens at zero volume (or area):
\begin{definition}
We say that a congruence of causal geodesics through a point $p$ has a \textbf{focal point} at $q$ (or, there is pair of \textbf{conjugate points} $(p,q)$ along a causal geodesic), if $\rho\to\infty$ along solutions of (\ref{1}), or,  because of (\ref{2}), when $v\to 0$. In this case, we say that we have \textbf{focusing} along the geodesic congruence  (also called `positive convergence').
\end{definition}
In terms of the expansion $\theta=-\rho$ of the congruence, when focusing occurs we have $\theta\to-\infty$. According to standard arguments, the inevitability of a \emph{focusing state}, that is when:
\be \label{fs}
\textrm{\textbf{Focusing State}:}\quad\rho\to\infty\quad\Leftrightarrow\quad\theta\to-\infty\quad\Leftrightarrow\quad v\to 0,
\ee
arises provided we choose initial conditions such that $\rho=\rho_0>0$, or equivalently, $\theta=\theta_0<0$. This state is synonymous to  the existence of a \textbf{spacetime singularity}, in the sense of geodesic incompleteness\footnote{A \emph{non-singular} spacetime is defined to be one that is  geodesically complete.},  formed at the `end' of gravitational collapse either in a cosmological situation or at the center of black holes. Except for the singularity theorems, the focusing state is also used in a very essential way to prove the area law for black holes, the statement that event horizons contain  trapped surfaces, and many other fundamental properties of black holes.

We now proceed to review the standard argument that leads to conditions for the occurrence of a focusing state. We treat  timelike geodesic congruences before the null case.

\subsubsection{Timelike focussing}\label{t-foc}

\emph{Step-T1: Use of the \textbf{generic condition}, $R_{abcd}t^b t^c\neq 0$.}

\noindent For the timelike case, ones sets $n=3$ in Eq. (\ref{1}b). A violation of the generic condition occurs precisely when $\mathcal{R}=0$, $\mathcal{W}=0$ in (\ref{1}). The usual arguments (cf. e.g., \cite{hp70}, p. 540 after Eq. 3.11) imply that only in very special, non-generic, unrealistic spacetimes and models, this situation may arise. In all such non-generic cases, the NPR system, Eq. (\ref{1}), becomes,
\begin{equation}\label{3}
\begin{split}
\dot{\rho}&=\rho^2+\sigma^2,\\
\dot{\sigma}&=n\rho\sigma.
\end{split}
\end{equation}
Therefore all the non-generic, special cases described by the system (\ref{3}) may be avoided by simply assuming that $\mathcal{R}\neq 0$, $\mathcal{W}\neq 0$ in the system (\ref{1}), or equivalently `re-inserting the perturbations' back into the (\ref{3}).

\noindent \emph{Step-T2: Use of the \textbf{energy condition}, $\mathcal{R}\geq 0$.}

A very lucky circumstance occurs here in the sense that  the strict inequality $\mathcal{R}> 0$, 1) complies with the non-generic-cases-avoiding condition (that is the generic condition of Step-1), and 2)  appears as the positive energy-density condition for matter crossing the geodesic congruence transversally. Thus  the energy condition assumption is absolutely necessary and plays  crucial role in the arguments leading the focusing state, and in addition it complies with a very plausible physical situation.

\noindent \emph{Step-T3: Partial \textbf{decoupling} of the Landau-Komar-Raychaudhuri equation.}

The main technical role of the energy condition is to alter the Eq. (\ref{1}a) into a weak inequality. In the first instance, one observes the the first equation in the system, Eq.  (\ref{1}a) (usually called the Landau-Komar-Raychaudhuri equation, \cite{ray1,ko,ray2}, \cite{ll}, p. 289), \emph{decouples} from the volume-convergence equation (2), namely the equation $\dot{v}=\rho v$, because it does not contain the variable $v$. In fact, it also decouples from the Eq.  (\ref{1}b), because using the energy condition and the positivity of the shear term, we have that $\sigma^2 +\mathcal{R}\geq 0$ (with the equality holding iff both terms on the left vanish), and so the Landau-Komar-Raychaudhuri equation   (\ref{1}a) becomes the weak inequality,
\be \label{ray1}
\dot{\rho}\geq\rho^2,
\ee
with the equality holding iff $\sigma=\mathcal{R}=0$.

Hence, the thought strikes one that the Landau-Komar-Raychaudhuri equation can be treated \emph{separately} both from the definition (\ref{2}) (thought of as a volume/area equation), and also from the second (the shear) equation (\ref{1}b), as something equivalent to the inequality (\ref{ray1}). In terms of the expansion of the geodesic congruence, we find  equivalently,
\be \label{ray2}
\dot{\theta}\leq -\theta^2.
\ee
It follows that the  weak inequality (\ref{ray1}) (equivalently (\ref{ray2}) for the expansion) fully describes  the Landau-Komar-Raychaudhuri equation (\ref{1}a), and so an infinite growth for $\rho$ (obtained by a simple integration of (\ref{ray1}) (or, resp., (\ref{ray2})) is unavoidable in all cases:  an initial condition $\rho_0>0$ (or $\theta_0<0$, respectively), implies that $\rho$ becomes infinite in proper time equal to $1/\rho_0$.

\noindent \emph{Step-T4: Use of the \textbf{volume equation} (\ref{2}).}

Since we now know the behaviour of $\rho(t)$ from the above argument,  the Eq. (\ref{2}) in the form $\dot{v}=\rho(t) v$, becomes a \emph{linear} (variable coefficient) equation in $v$, and it may be shown  that the volume function $v(t)$ vanishes as $\rho$ diverges to infinity. The standard argument for this is to show that  the positive function $l(t)=v^{1/3}(t)$ satisfies $\ddot{l}\leq 0$, and so it is concave and  vanishes when $\rho$ diverges (namely, in time at most $1/\rho_0$ a focal or conjugate point is created).

Therefore a focusing state  for a timelike congruence is the result.

\subsubsection{Null focusing}\label{n-foc}
The method to show that a focusing state results for null geodesic congruences is essentially analogous to that in the timelike case, with small differences in the two treatments, as we now discuss. One uses a null tetrad, in particular, we use the null vector field $l^a$ with obvious modifications in the definitions of the quantities $\rho, \sigma, \mathcal{R}, \mathcal{W}$, and of course the area (instead of volume) element. Under these changes, one uses the system (\ref{1}) for the treatment of the null case.

\noindent \emph{Step-N1: Use of the generic condition, $R_{abcd}l^b l^c\neq 0$.}

 This works exactly like in the timelike case, Step-T1 above, but with $n=2$ in the non-generic system (\ref{3}).

\noindent \emph{Step-N2: Use of the energy condition, $\mathcal{R}\geq 0$.}

This is constructed as a limiting case as $t^a\to l^a$, with exactly the same conclusions as in Step-T2.

\noindent \emph{Step-N3: Partial decoupling of the Raychaudhuri equation.}

Again, one obtains the equation (\ref{ray1}) (or, (\ref{ray2})) but through a different procedure from the physical point of view. We consider a pulse of light (i.e., a congruence of null geodesics near some given one) that initially is  a parallel circular beam defined by the state: $\rho=0, \sigma=0$, in a region where $\mathcal{R}=\mathcal{W}=0$. In this situation, focusing is generated in the two main cases, namely, that of an anastigmatic lens with $\sigma=0$, and that of an astigmatic lens, $\sigma\neq0$ (where we have $\mathcal{W}\neq 0$).

In the former case, the situation is described by Eq. (\ref{ray1}), and so focusing follows as before, and the  Eq. (\ref{ray1}) is decoupled from the shear equation (\ref{1}b). In the latter case of an astigmatic lens, the Raychaudhuri equation is the first equation in the system (\ref{3}), and so focusing still occurs (due to the positivity of the shear term, one still gets Eq. (\ref{ray1})). The remaining cases are as follows: If we add a nonzero $\mathcal{R}$ satisfying the energy condition of Step-N2, to an anastigmatic lens, then $\dot{\rho}=\rho^2+\mathcal{R}$, which is non-negative provided this is a strict inequality, and focusing follows. If we further add a nonzero $\mathcal{W}$, then we shall have shear present, that is an astigmatic lens, and we end up with the case considered previously, where again focusing follows.

\noindent \emph{Step-N4: Use of the area equation (\ref{2}).}

This step proceeds exactly as before in Step-T4, but this time with $l(t)=v^{1/2}(t)$. We note that the area equation is again used in the calculation for the second derivative.
Therefore in the null case, focusing is the result of the (essential similar) application of the four steps above.

To end our discussion on the standard focusing mechanism, we note an alternative derivation leading to  a focusing state. This is  given in Ref. \cite{stew}, Sect. 2.7, and  p. 203 (and refs. therein), and is completely equivalent to the above. This derivation is of interest for the present work because it uses the non-generic system (\ref{3}): Defining the functions,
\be \label{w}
w_{\pm}=\rho\pm\sigma,
\ee
and taking the algebraic sum of the two equations in (\ref{3}), reduces the two-dimensional system (\ref{3}) to a one-dimensional one for $w_{\pm}$, namely, $\dot{w}_\pm=w_{\pm}^2$, which can then be treated using the above methods but now applied to the function $w_\pm$. Then $w_{\pm}$ are found to diverge at a finite affine parameter value, hence so do $\rho$ and $\sigma$, assuming (as it is done in \cite{stew}) that $\rho\sim\sigma$ (Note: these lines will appear below as the `Stewart separatrices').

\subsubsection{Applications to spacetime singularities and black holes}\label{sings-app}
As is well-known, there are a number of fundamental results in general relativity that use the existence of a focusing state (\ref{fs}) in an essential way in their proofs. We note that all such proofs are constructed by contradiction. We refer below to a small, selected number of  theorems, in which the existence of a focusing state (\ref{fs}) is used as a true statement, so that this - or an implication of it - be compared with some other statement or hypothesis of the result to be proved,  to  obtain the desire contradiction.

The following results have proofs that depend in an essential way on the use of the statement about the divergence of $\rho$ (or, $\theta$), and so on the existence of a focusing state (or a conjugate point), and also on the fact that the focusing state contradicts some other proposition or a hypothesis of the theorem.
\begin{enumerate}
\item \textbf{The Penrose 1965 singularity theorem},  \cite{pen65}.

\item \textbf{The two Hawking singularity theorems of 1967}, \cite{ha67}.

\item \textbf{The Hawking-Penrose 1970 singularity theorem}, \cite{hp70}.

\item \textbf{The area law for black holes}, \cite{ha71}, p. 1345 (cf. discussion before Eq. 2),  \cite{wald}, p. 312.
\item The black hole property that: \textbf{`a trapped surface is contained in the event horizon'}, \cite{he}, prop. 9.2.1, and many other black hole properties, cf. e.g., \cite{he}, Sections 9.2, 9.3.
\end{enumerate}
A typical argument met in the theorems above concerning how the focusing state is used, goes as follows (here it is about the Penrose theorem which is the prototypical result of all), and shows how important the existence of a focusing state is in the proofs of all these fundamental theorems.

One chooses the future trapped surface $\mathcal{T}$ which corresponds to the \emph{maximum} negative value of the expansion, say $\theta_0$, for both sets of null geodesics orthogonal to $\mathcal{T}$. This means that $\mathcal{T}$ corresponds to the `earliest' such surface.

One then considers $A$, the set of all spacetime points lying on a null geodesic  that starts at $\mathcal{T}$ and proceeds with an affine parameter all the way up to $2/|\theta_0|$. One  shows that $A$ is compact as a continuous image of a compact set under suitable continuous maps.

Then we have the crucial step: the existence of a focusing state as in the prescription  (\ref{fs}) is used to show that any point of the boundary of the  future of the surface $\mathcal{T}$, $\partial I^+(\mathcal{T})$, also belongs to $A$: it cannot proceed further than $2/|\theta_0|$ along a null geodesic, because $\theta\to -\infty$ there.

In the last step, one shows that the compactness of $\partial I^+(\mathcal{T})$ (it is a closed subset of $A$) contradicts another hypothesis of the theorem (namely, the global hyperbolicity of the spacetime).

\subsection{The versality of the focusing effect}\label{bifn-1}
In this Subsection, we discuss two silent but very important points in the standard treatment of the system (\ref{1}), common to both the timelike and the null cases.

The first is that the standard argument shows us the way of how to correctly distinguish between the linear, or \emph{adversarial},  behaviour leading to the focusing state and any other: we simply have to take into account  \emph{essentially nonlinear} feedback effects associated with the system (\ref{1}). Such effects may also lead to something less than adverse behaviour.

The second point is associated with the use of the word `generic' in the standard treatment of the problem.
In this connection, we discuss how the standard approach to spacetime singularities through the focusing effect represents the first ever bifurcation theory approach to spacetime singularities  in general relativity.

We end this Subsection with a number of further questions associated with a bifurcation theory approach to the problems of singularities and black holes.

\subsubsection{Focussing is a linear, adversarial feedback effect}
We show that the focusing effect is the result of combining the decoupling of the Landau-Komar-Raychaudhuri equation with the \emph{linearity} of the volume/area evolution equation.


As explained in the previous Section, the energy condition (together with the non-negativity of the shear) in conjunction with the generic condition which dictates that the terms $\mathcal{R},  \mathcal{W}$ \emph{must} be present to ensure that one avoids non-generic effects lead to the partial decoupling of the Landau-Komar-Raychaudhuri equation from the volume/area \emph{and} the shear equations (Eqns. (\ref{1}), (\ref{2}), respectively), and this is absolutely necessary in order  to obtain the desired focusing behaviour of the expansion scalar $\theta$ (or, the convergence $\rho$).
With this decoupling, the Landau-Komar-Raychaudhuri equation (\ref{1}a) becomes the inequality (\ref{ray1}) (or, equivalently,  (\ref{ray2})) (for the shear case, this follows since $\rho^2+\sigma^2\geq \rho^2$, so that one always ends up with the Landau-Komar-Raychaudhuri inequality without a shear term), which can be directly integrated. Therefore  the volume equation (\ref{2})  can now be treated \emph{independently} from the Eq. (\ref{1}a).

As we also noted in the Step-T,N4 of Sections \ref{t-foc}, \ref{n-foc}, the standard treatment of volume/area equation proceeds by the direct calculation of the second derivative of the quantity $x$ that denotes either the volume,  or the area $a$ (or the `luminosity parameter' $L$ of \cite{hp70}, p. 542), and satisfying the linear equation $\dot{x}=\theta(t) x$. The result is that $x$ vanishes in a finite time, and a focusing state is the result.
Similarly, for  the linear equation for the shear, $\dot{\sigma}=n\rho(t)\sigma$, knowing the behaviour of $\rho$, we can directly integrate\footnote{We note that the term $\mathcal{W}$ can be also ignored from it  because it is usually regarded as inducing extra shear thus enhancing the convergence power and so the focusing effect ($\mathcal{W}$ acts as a purely astigmatic lens), cf. \cite{pe68}, pp. 167-9, \cite{ha73}, pp. 44-45.}.


We shall  now provide a different treatment of the linear equation (\ref{2}), that is, $\dot{v}=\theta(t) v$. We make use of a differential form of the Gronwall's inequality, as this is developed  in  \cite{tao}, pp. 12-3. We take a more general stance, and consider instead a differential \emph{inequality} for the variable $x$ playing the role of either the volume $v$, or the area $a$, or the shear $\sigma$, and satisfying an inequality of the form,
\be\label{dif}
  \dot{x}\leq\theta(t) x.
\ee
We note that this inequality is sharp in the `worst case scenario', that is  when $\dot{x}=\theta(t) x$ holds for all $t$.
Looked at it this way, the latter is the case of \emph{adversarial feedback} \cite{tao}, when the forcing term $\theta(t) x$ always acts to increase $\theta(t)$ the maximum possible amount, precisely  as in the focusing effect.
To see this,  from Gronwall's theorem, assuming that $\theta$ is  continuous on some interval $[t_0,t_1]$,   it follows directly from the linear relation (\ref{dif}) that,
\be
x(t)\leq x(t_0)\exp\left( \int_{t_0}^{t} \theta (s) ds \right),
\ee
from which the focusing state ($x\to 0$ as $\theta\to -\infty$) directly follows.

Hence, the linear feedback of the term $\theta x$ to the solution of the equation $\dot{x}=\theta x$, causes exponential decay of the solution $x$ as $\theta\to -\infty$. On the other hand, if we had a \emph{nonlinear} forcing term, say $F(X), X=(\rho,x),$ (instead of $(\rho^2,\rho x)$) influencing the solution $X$, then the feedback loop,  $X\to F(X)\to X$,  of the solution $X$ influencing the forcing term nonlinearly which in turn influences the solution $X$,  would lead to an overall difference in the behaviour of the system (\ref{1}).

In conclusion, the standard treatment of the system (\ref{1}) helps to clearly distinguish between  the adversarial   behaviour of the solutions associated with a \emph{linear} feedback loop, namely the focusing behaviour, and any other. This is accomplished by an analysis of the linear and adversarial feedback effects of the equations, and by associating the focusing effect to the linearity of the volume/area equation (as a result of the decoupled treatment of the Landau-Komar-Raychaudhuri equation).

Therefore any truly nonlinear feedback effect associated with the system (\ref{1}) must, in addition to the focusing state, lead to some distinctly different behaviour controlling the feedback loop, as compared to the focusing effect. Of course, the problem with this is that normally one does not have any general  procedure to realize \emph{that} distinction.

We shall show presently that such a method may be subtly obtained by an application of bifurcation theory to this problem. The crucial advance that bifurcation theory brings about in this case is that taking seriously the versal unfolding idea and applying it to the problem (\ref{1}), one ends up with a concrete proposal of the exactly admissible forms of the \emph{essentially nonlinear} `perturbation terms' $\mathcal{R},  \mathcal{W}$, in such a way that the resulting `unfolding' is \emph{versal}: that is, it contains all stable perturbations of the system.

\subsubsection{The focusing effect as a primitive bifurcation calculation}
Let us now move on to the consideration of the meaning of the generic condition as used in the standard approach of the system (\ref{1}). For the standard treatment and meaning of the generic condition, we refer to the references, in particular, e.g.,  \cite{hp70}, p. 540. According to this, the generic condition is used to avoid very special and therefore physically unrealistic geometric situations when employing the focusing effect.

When the generic condition fails, we have $\mathcal{R}=\mathcal{W}=0$, and (\ref{3}) - instead of (\ref{1}) - is the result, namely, we obtain the dynamical system (\ref{3}).
In the standard approach, and because this system is associated with the unphysical situation discussed in the previous paragraph,  any analysis of it is  avoided by using the generic condition and `re-inserting the perturbation terms'  $\mathcal{R},\mathcal{W}\neq 0$ to obtain the original system (\ref{1})\footnote{An exception is the method noted earlier, cf. Eq. (\ref{w}), however, that analysis is completely equivalent to the standard treatment.}.

However, this approach is indeed very similar to that met in bifurcation theory. Namely, one starts with a system (cf. (\ref{3})) which has some degeneracy (in this case, this means that it represents unphysical solutions). As we have emphasized earlier (cf. the sequence \ref{bifn-method}), the way to deal with such  systems is to augment or `unfold' the original system to consider all possible perturbations of it in a consistent way.

Seen in this way, the pioneering Hawking-Penrose treatment leading to the focusing effect and the singularity theorems  represents the very first bifurcation calculation to the study of singularities: With a vanishing shear, the system (\ref{3}) reads,
\be\label{equal} \dot{\rho}=\rho^2,\ee which represents the `normal form' of the system (\ref{3}) for vanishing shear. Since the Landau-Komar-Raychaudhuri equation decouples, in the  general situation we expect to have $\dot{\rho}\geq\rho^2$, and this can be regarded as  a `versal unfolding' of (\ref{equal}), because `it contains all possible perturbations' of it (that is,  the terms $\sigma^2, \mathcal{R}$). Then the treatment of the versal unfolding through the use of the volume/area equation, yields the final answer, that is the focusing effect epitomized as the existence of spacetime singularities (through the globalization obtained by global causal structure techniques).

It is interesting that from this point of view, the original Hawking-Penrose theorems are very complete because they led to the \emph{only} possible answer, namely, the focusing effect. The essential uniqueness of this answer is of course related to the `versal unfolding' mechanism in conjunction with the adversarial behaviour.

\subsection{Nonlinear feedback effects, versality,  and spacetime singularities}\label{lin}
The present work may be regarded as a nonlinear version of the Hawking-Penrose approach seen in this way. As we are interested in genuine nonlinear feedback effects associated with the system (\ref{1}), the non-generic system (\ref{3}) appears as a simpler one compared to (\ref{1}) just because there are no `unknown' perturbation terms like $\mathcal{R},\mathcal{W}$ in it. Regarding the possible role that the  system (\ref{3}) may play for the main problem, that is (\ref{1}), we ask:
\begin{enumerate}
\item What are the dynamical properties of (\ref{3})?
\item What is the relation between properties of the system (\ref{3}) and those of the `perturbed system' (\ref{1})?
\item What is the nature of the perturbation terms $\mathcal{R},\mathcal{W}$?
\item Is the system (\ref{3}) structurally stable under small perturbations to `nearby' systems?
\item In what sense is the behaviour of the system (\ref{3}) `non-generic' as compared to that of nearby systems?
\item How does one perturb  a degenerate system such as (\ref{3})?
\item How does any `degeneracy' of the subsystem (\ref{3}) affect the behaviour of the solutions of the original system (\ref{1})?
\item Can one account for the essentially nonlinear nature of the system (\ref{1}), without at the same time understanding the nonlinear nature of its `degenerate' subsystem given by (\ref{3}) ?
\item What is the nature of the set of all possible stable perturbations  of (\ref{3}), and what is their precise relation to the original system (\ref{1})?
\item What is the influence of the vorticity in the evolution?
\end{enumerate}
These questions are related to essentially nonlinear behaviours present of the system (\ref{1}) (and of course (\ref{3})), and because they are also important for the other two systems, namely  (\ref{v}), (\ref{os1}), they are a main topic of the present paper. In fact, the pioneering papers that established the existence of singularities and black holes in general relativity are very useful, even essential, in this respect because the focusing and related behaviours must be present in the final answer, and as such may play the role of a basis to orient ourselves in the possible patterns and forms that may emerge from, and be associated with,  the answers to the questions above.

The mathematics needed for the full analysis of such  problems as (\ref{3}), (\ref{1})   belongs to bifurcation theory, where the nature of important ideas such as degeneracy, non-hyperbolicity, topological normal forms, versal unfoldings,  symmetries, and local and global bifurcations, can be adequately clarified. In turn, these ideas and others will necessarily play a central role in an attempt to unravel the dynamical nature of spacetime singularities and black holes.

The system (\ref{3}) is not  as trivial or uninteresting system as it looks. A purpose of this paper is to provide a full analysis of both this system and of the original system (\ref{1}). A full understanding of the latter depends on that of the former system, to such an  extent that it is not possible to say anything reliable about the system   (\ref{1}) without first understanding fully the system (\ref{3}). Once this is done, we shall return  to apply the results to the problem of spacetime singularities and the structure of black holes.

Similar remarks apply to the treatment of the systems (\ref{v}) describing vorticity-induced effects, and (\ref{os1}) which is the Oppenheimer-Snyder for spherically symmetric dustlike gravitational collapse. We note that the system (\ref{v}) can be studied with bifurcation methods very similar to those of the NPR system (\ref{1}) because they share the same structure of the linear terms and also they both are $\mathbb{Z}_2$-equivariant. However, the Oppenheimer-Snyder system  (\ref{os1}) is different because its linear part is different than those of the NPR and CV-systems, even though here we too have a double-zero eigenvalue.

However, the vorticity system (\ref{v}) leads to very different effects compared to the NPR system, the main difference between them being that whereas the NPR solutions are generally characterized as being unstable or of a `runaway' character, the CV-system has a unique, stable limit cycle attracting all solutions to a self-sustained state of oscillations, for certain ranges of the parameters.

On the other hand, in the the Oppenheimer-Snyder problem, the versal unfolding has two moduli and so leads to two very different bifurcation diagrams, one sharing some of the effects found in the NPR-system, and the other being closer to the vorticity-induced effects of the CV-system, namely, the appearance of closed orbits and global bifurcations.

\section{The bifurcation diagram of the NPR system}\label{bif}
In this Section, we study  the versal unfolding dynamics associated with the Newman-Penrose-Raychaudhuri system  (\ref{1}).

\subsection{The normal form and versal unfolding}
We start with the degenerate system  (\ref{3}), and the basic observation that it possesses the $\mathbb{Z}_2$-symmetry, namely it is invariant under the transformation,
\be\label{z2}
\rho\to\rho,\quad \sigma\to-\sigma.
\ee
The system (\ref{3}) is already in normal form. The various normal forms and versal unfoldings for systems with a $\mathbb{Z}_2$-symmetry have been completely classified, cf. e.g., \cite{golu3}, Sections XIX.1-3, \cite{kuz}, Section 8.5.2, \cite{wig}, Sections 20.7, 33.2, \cite{gh83}, Sect. 7.4, \cite{ar94}, Section 4.4, and refs. therein (finding the versal unfolding in such systems has been  one of the most illustrious problems in bifurcation theory). Consequently, the versal unfolding of (\ref{3}) is given by,
\begin{equation}\label{vaA}
\begin{split}
\dot{\rho}&=\mu_1+\rho^2+\sigma^2,\\
\dot{\sigma}&=\mu_2 \sigma+n\rho\sigma,
\end{split}
\end{equation}
where $\mu_1,\mu_2$ are the two unfolding parameters, and $n=2,3$. Our efforts here will be directed to obtain the complete bifurcation diagram of (\ref{vaA}), that is the parameter diagram together with the corresponding phase diagrams for each one of the strata partitioning the parameter diagram.

With regard to the system (\ref{vaA}), the following remarks are in order.
Although the parameter diagram for the system (\ref{vaA}) will be the end-result of the analysis  in this Section, it is instructive and helpful to give it here and refer to it as we develop the details, cf. Fig. \ref{fig-rem}. In this Figure, we observe the different strata determining the subsequent phase space dynamics. We see that there are seven important regions in the full parameter diagram in Fig. \ref{fig-rem}, namely\footnote{We shall introduce a particular naming system for the various strata of the parameter diagrams in this and following Sections of this paper, that  reflect the great variety of the possibilities that arise due to the codimension-2 bifurcations. This naming system comes from the poem `Theogeny' of Hesiod, and imparts names and corresponding letters to the various strata according to the \emph{First Gods} appearing in that poem.},
\begin{enumerate}
\item The origin, the \textbf{Gaia}-$\gamma$ stratum.
\item The right half-plane, the \textbf{Chaos}-$\chi$ stratum.
\item The $\mu_2$-axis, the \textbf{Eros}-$\varepsilon$ stratum, in two components $\varepsilon_+,\varepsilon_-$.
\item The parabola in the left half-plane, the \textbf{Tartara}-$\tau$ stratum, in two components $\tau_+,\tau_-$.
\item The upper region between $\varepsilon_+$ and  $\tau_+$, the \textbf{Uranus}-$o$ stratum.
\item The lower region  between $\varepsilon_-$ and  $\tau_-$, \textbf{Pontus}-$\pi$ stratum.
\item The region inside the parabola, the \textbf{Ourea}-$\beta$ stratum.
\end{enumerate}
These regions will have an important role to play in the dynamics of the system (\ref{vaA}), and will appear as the analysis unfolds.
\begin{figure}
\centering
\includegraphics[width=0.3\textwidth]{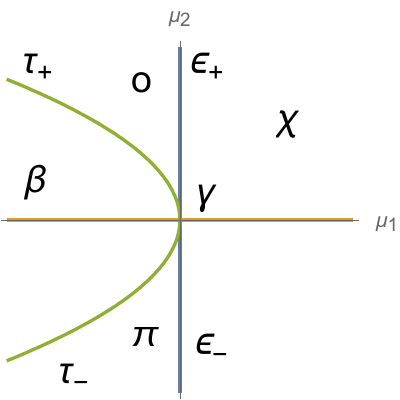}
\caption{The complete picture of the  strata partitioning the parameter diagram for the NPR-system. We have  seven strata, namely, $\chi,\gamma,\varepsilon, o,\pi,\tau,\beta$.}\label{fig-rem}
\end{figure}

\subsection{Dynamics at zero parameter}\label{0-npr}
Let us first consider the dynamics of the versal unfolding (\ref{vaA}) in the case where $(\mu_1,\mu_2)=(0,0)$, that is the degenerate system (\ref{3}), which we also reproduce here,
\begin{equation}\label{3A}
\begin{split}
\dot{\rho}&=\rho^2+\sigma^2,\\
\dot{\sigma}&=n\rho\sigma.
\end{split}
\end{equation}
This is the case that corresponds to the origin in the parameter diagram, the Gaia-$\gamma$ point, and it comprises the degenerate system (\ref{3A}) that is already in normal form.
To obtain the corresponding phase portrait in this case, we work as follows. First, by exploiting the fact that in our problem (\ref{3A}) we have $n=2,3$, we can now introduce the key fact that the lines,
\be \label{invl}
\sigma=m\rho,
\ee
represent  invariant lines in the phase plane $(\rho,\sigma)$ of the problem. This follows because using (\ref{3A}), (\ref{invl}), the condition of tangency of these lines to the flow, gives,
\be\label{inv-line}
m^2=n-1,
\ee
which is always satisfied, so that these lines are separatrices of the flow. (We note that the $\rho$-axis is always invariant in this problem as one may easily check.) For $n=2$ (resp. 3), we have $m=\pm 1$ (resp. $\pm\sqrt{2}$), and we  call them the \emph{Stewart separatrices}, since the former lines were first found in \cite{stew}, p. 203, using different methods (which, however, reduce the codimension of the singularities).

To find the direction of the flow along the Stewart lines, we calculate the product of the vector field (\ref{3A}) times the radial vector, namely,  $l=(\rho^2+\sigma^2,n\rho\sigma)\cdot(\rho,\sigma)$, on the Stewart lines, to get,
\be
l=n^2\rho^3,
\ee
so that $l>0$ (resp. $<0$) when $\rho>0$ (resp. $<0$), and the flow is directed outward (inward) along the Stewart separatrices.

To obtain the full phase portrait for (\ref{3A}), consider the function,
\be
i_n(\rho,\sigma)=\frac{n}{2\sigma^{2/n}}\left(\rho^2+\frac{\sigma^2}{1-n}\right),
\ee
and it is straightforward to see that the derivative $\dot{i}_n=0$, for any $n>1$ along the flow (\ref{3A}), making $i_n$ a first integral of (\ref{3A}). Therefore the level curves of $i_n$ provide all the phase curves of the phase portrait of (\ref{3A}). For instance, for the null case ($n=2$),  the family,
\be
\rho^2-\sigma^2=a\sigma, \quad a\neq 0,
\ee
gives all the orbits below (above) the Stewart separatrices (which are also shown), with $a\gtrless 0$, as in Fig. \ref{fig-gaia2}. The timelike case (with $n=3$) is very similar\footnote{We note here a subtle difference in the phase portrait given in Fig. \ref{fig-gaia2} in that the Stewart separatrices only exist because in our problem because $n>1$. Had $n$ be in the range $0<n<1$, the phase portrait in that case would have been different (actually it would be more similar to that of Fig. \ref{fig-chaos2} below). In this sense, the $n$ which takes the values 2, 3 (for the null and timelike case, respectively) is in fact a second, and already determined, modular coefficient of the problem.}.
 \begin{figure}\label{fig-gaia}
      \begin{subfigure}[b]{0.3\textwidth}
         \includegraphics[width=\textwidth]{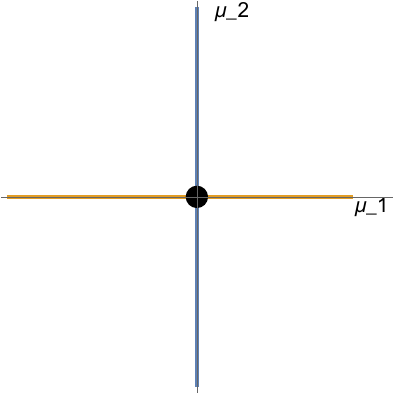}
         \caption{Parameter diagram, Gaia stratum.}
         \label{fig-gaia1}
     \end{subfigure}
     \hfill
     \begin{subfigure}[b]{0.3\textwidth}
         \includegraphics[width=\textwidth]{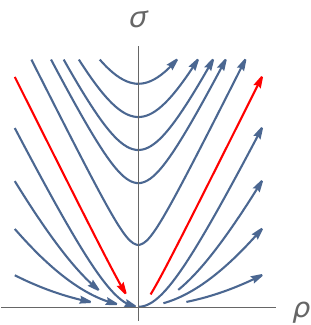}
         \caption{Phase portrait for Gaia.}
         \label{fig-gaia2}
     \end{subfigure}
       \caption{Parameter diagram showing the Gaia-$\gamma$ stratum at the origin in parameter space, and the corresponding phase portrait (the Stewart separatrices are shown in red).}
 \end{figure}

\subsection{Stability of the fixed branches}\label{stab1}
We now return to the consideration of the full system (\ref{vaA}).
To begin the stability analysis of (\ref{vaA}), we note the decisive fact that this system has the $\mathbb{Z}_2$-symmetry (\ref{z2}),
and without loss of generality we may assume that $\sigma>0$.

The system (\ref{vaA}) has three fixed branches\footnote{Since in the situation we shall be dealing  with, the fixed `points' are parameter-dependent, we shall usually call any equilibrium, parameter-dependent, solution family, a \emph{fixed branch}.}:
\begin{enumerate}
  \item $\mathcal{E}_{1,2}=(\mp\sqrt{-\mu_1},0)$. These are real, provided
  \be \mu_1<0.\ee
  \item The third fixed branch is,
\be\label{e3}
\mathcal{E}_{3}=\left(-\frac{\mu_2}{n},\sqrt{-\left(\frac{\mu_2^2}{n^2}+\mu_1 \right)}\right),
\ee
which is real if the bracket inside the square root is negative, $\mu_1<-\mu_2^2/n^2$.
\end{enumerate}
A particular aspect of the ensuing  bifurcation analysis is that although the fixed branches belong in the phase space of the problem, that is on the $(\rho,\sigma)$-plane in this case, because of their parameter dependence they may also be considered as `attached' to the parameter diagram of Fig. \ref{fig-rem}. This observation is useful in understanding many of the subsequent dynamical issues.
For example, the fixed branches $\mathcal{E}_{1,2}$ belong to the $\sigma=0$ axis of the phase space, however, they can also be considered as attached to the negative $\mu_1$ axis of the parameter diagram in Fig. \ref{fig-rem}, while the fixed branch $\mathcal{E}_{3}$ is also attached  to the $\beta$-stratum Ourea of Fig. \ref{fig-rem}  (apart from lying anywhere in the phase plane, except at the origin).

We note that none of  the three fixed  branches exist for the $\chi$-stratum (half-space $\mu_1>0$, cf. Fig. \ref{fig-rem}) - no fixed points there,  and so the phase portrait can be easily drawn in this case, see Fig. \ref{fig-chaos1}, \ref{fig-chaos2}.
 \begin{figure}\label{fig-chaos}
      \begin{subfigure}[b]{0.3\textwidth}
         \includegraphics[width=\textwidth]{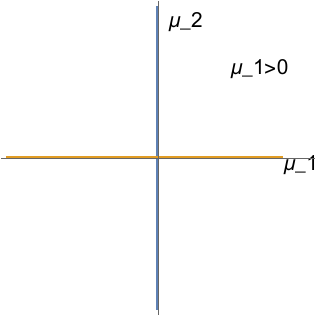}
         \caption{Parameter diagram, Chaos stratum.}
         \label{fig-chaos1}
     \end{subfigure}
     \hfill
     \begin{subfigure}[b]{0.3\textwidth}
         \includegraphics[width=\textwidth]{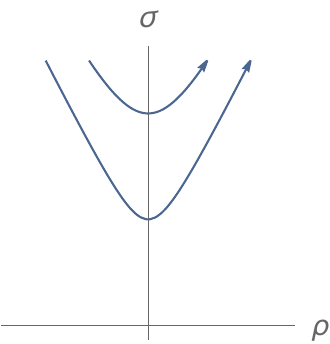}
         \caption{Phase portrait for Chaos-$\chi$.}
         \label{fig-chaos2}
     \end{subfigure}
       \caption{Parameter diagram showing the Chaos-$\chi$ stratum in the right half-space of the parameter diagram, and the corresponding phase portrait for $\chi$-stratum.}
 \end{figure}

Before we examine the stability of the fixed branches, some preliminary work is needed.
\subsubsection{Two lemmas}
The parameter plane is stratified according to hyperbolic or bifurcating behaviour associated to changes in the $\mu$ parameter, and a simple way to connect the different strata of parameter diagram to the corresponding phase space dynamics is to directly find the stability of the three fixed branches $\mathcal{E}_i, i=1,2,3,$. This can be done by employing the two lemmas below, instead of developing stability for each one of the particular subsets of the parameter space $(\mu_1,\mu_2)$.

The linearized Jacobian of (\ref{vaA}) is given by,
\be
J=\left(
  \begin{array}{cc}
    2\rho & 2\sigma \\
    n\sigma & \mu_2 +n\rho\\
  \end{array}
\right),
\ee
and it is helpful to use the standard formulae for the eigenvalues, that is,
\be
\lambda^2-(\textrm{Tr} J)\lambda +\textrm{det} J=0,
\ee
where,
\be \label{eig}
\lambda_{\pm}=\frac{1}{2}(\textrm{Tr}\, J \pm\sqrt{\Delta}),\quad \Delta=(\textrm{Tr}\, J)^2-4\,\textrm{det} J.
\ee
For any of the three fixed branches $\mathcal{E}_i,i=1,2,3,$ the following two lemmas about branches follow easily, and  their proofs will be omitted. The first lemma describes simple bifurcational behaviour, i.e., simple situations where the dynamics near a fixed branch will be radically different from that of the linearization\footnote{The word `slightly' in lemma \ref{1stlemma} means that only codimension-1 bifurcations will exist in this case.}.
\begin{lemma}[Slightly dispersive behaviour.]\label{1stlemma} For a fixed branch $\mathcal{E}$ we have,
\begin{enumerate}
\item If $\textrm{det}\, J=0$, then one of the eigenvalues of the linearized Jacobian is zero.
\item If $\textrm{det}\, J>0$, and  $\textrm{Tr}\, J=0,$ then $\lambda _\pm=\pm i\sqrt{|\textrm{det}\, J|}$, and $\mathcal{E}$ is a \textbf{centre}.
\end{enumerate}
\end{lemma}
The second lemma describes the range of hyperbolic behaviours resulting from any situation with a non-zero linearized Jacobian.
\begin{lemma}[Hyperbolic behaviour.]\label{2ndlemma} For a fixed branch  $\mathcal{E}$ we have:
\begin{enumerate}
\item If $\textrm{det}\, J<0$, then $\mathcal{E}$ is a \textbf{saddle}. When in addition, $\textrm{Tr}\, J>0,$ or $\textrm{Tr}\, J<0,$ then $\lambda _+>0,\lambda_-<0$, whereas when $\textrm{Tr}\, J=0,$ then  $\lambda _\pm=\pm\sqrt{|\textrm{det}\, J|}$.
\item If $\textrm{det}\, J>0$, and  $\textrm{Tr}\, J>0,$ then $\lambda _+>0,\lambda_->0$, and $\mathcal{E}$ is a \textbf{source}.
\item  If $\textrm{det}\, J>0$, and  $\textrm{Tr}\, J<0,$ then $\lambda _+,\lambda_-<0$, and $\mathcal{E}$ is a \textbf{sink}.
\end{enumerate}

\end{lemma}
We are now in a position to proceed with the nature of the fixed branches.

\subsubsection{Stability of the branch $\mathcal{E}_1$}
To apply the two lemmas for the fixed branch $\mathcal{E}_1=(-\sqrt{-\mu_1,0})$, we first calculate,
\be\label{tr-det1}
\textrm{Tr}\, J|_{\mathcal{E}_1}=\mu_2-(2+n)\sqrt{-\mu_1},    \quad\textrm{det}\, J|_{\mathcal{E}_1}=-2\sqrt{-\mu_1}(\mu_2-n\sqrt{-\mu_1}).
\ee
Then we have the following result.
\begin{proposition}[Sign conditions of $\textrm{Tr}\, J|_{\mathcal{E}_1},\textrm{det}\, J|_{\mathcal{E}_1}$.]
For the fixed branch $\mathcal{E}_1=(-\sqrt{-\mu_1,0})$, we have the following sign conditions:
\begin{enumerate}
\item $\textrm{Tr}\, J|_{\mathcal{E}_1}\gtrless 0$, implies that $\mu_2-n\sqrt{-\mu_1}\gtrless 2\sqrt{-\mu_1}$.
\item $\textrm{det}\, J|_{\mathcal{E}_1}\gtrless 0$, implies that $\mu_2-n\sqrt{-\mu_1}\lessgtr 0$.
\end{enumerate}
\end{proposition}
From these results, we arrive at the following theorem about the stability of the fixed branch $\mathcal{E}_1$ (we note that $\mu_1<0$).
\begin{theorem}[Nature of branch  $\mathcal{E}_1$.]\label{thm-e1} We have the following types for the $\mathcal{E}_1$ branch:

\begin{enumerate}
  \item $\textrm{det}\, J|_{\mathcal{E}_1}= 0,\quad\textrm{Tr}\, J|_{\mathcal{E}_1}>0:$ \emph{cannot happen.}
  \item $\textrm{det}\, J|_{\mathcal{E}_1}= 0,\quad\textrm{Tr}\, J|_{\mathcal{E}_1}<0:$ \emph{bifurcation.}
  \item $\textrm{det}\, J|_{\mathcal{E}_1}> 0,\quad\textrm{Tr}\, J|_{\mathcal{E}_1}>0:$ \emph{cannot happen.}
  \item $\textrm{det}\, J|_{\mathcal{E}_1}> 0,\quad\textrm{Tr}\, J|_{\mathcal{E}_1}<0:$ \emph{stable node.}
  \item $\textrm{det}\, J|_{\mathcal{E}_1}< 0,\quad\textrm{Tr}\, J|_{\mathcal{E}_1}>0:$ \emph{saddle.}
  \item $\textrm{det}\, J|_{\mathcal{E}_1}< 0,\quad\textrm{Tr}\, J|_{\mathcal{E}_1}<0:$ \emph{saddle.}
  \item $\textrm{Tr}\, J|_{\mathcal{E}_1}=0,\quad\textrm{det}\, J|_{\mathcal{E}_1}> 0:$ \emph{cannot happen.}
  \item $\textrm{Tr}\, J|_{\mathcal{E}_1}=0,\quad\textrm{det}\, J|_{\mathcal{E}_1}< 0:$ \emph{saddle.}
\end{enumerate}
\end{theorem}
\begin{proof}
Items 1, 3, and 7 follow from (\ref{tr-det1}). For 2, using Lemma 1, we find  from (\ref{eig}) that one of the eigenvalues is zero, and so we expect a saddle-node bifurcation on the $\mu_2$-axis when $\mu_2>0$. For 4, it follows that $\textrm{Tr}\, J|_{\mathcal{E}_1}+\sqrt{\Delta}<0$, which implies that both eigenvalues are negative, and from Lemma 2 we find a sink (i.e., stable node). Items 5, 6 follow directly from Lemma 2, while for item 8, we get a saddle from Lemma 2.
\end{proof}

Therefore from  Theorem \ref{thm-e1}, we find the following types of stability for the fixed branch $\mathcal{E}_1$.
\begin{corollary}\label{cor1}
The nature of the fixed branch $\mathcal{E}_1$ is as follows:
\begin{enumerate}
  \item When $\mu_2-n\sqrt{-\mu_1}>0$,  $\mathcal{E}_1$ is a saddle.
  \item When $\mu_2-n\sqrt{-\mu_1}<0$,  $\mathcal{E}_1$ is a sink.
  \item When $\mu_2-n\sqrt{-\mu_1}=0$,  $\mathcal{E}_1$ bifurcates.
\end{enumerate}
\end{corollary}
This completes the stability of the equilibrium $\mathcal{E}_1$.

\subsubsection{Stability of the branch $\mathcal{E}_2$}
For the fixed branch $\mathcal{E}_2=(\sqrt{-\mu_1,0})$, we find,
\be\label{tr-det2}
\textrm{Tr}\, J|_{\mathcal{E}_2}=\mu_2+(2+n)\sqrt{-\mu_1},    \quad\textrm{det}\, J|_{\mathcal{E}_2}=2\sqrt{-\mu_1}(\mu_2+n\sqrt{-\mu_1}).
\ee
The sign conditions now become:
\begin{proposition}[Sign conditions of $\textrm{Tr}\, J|_{\mathcal{E}_2},\textrm{det}\, J|_{\mathcal{E}_2}$.]
For the fixed branch $\mathcal{E}_2=(\sqrt{-\mu_1,0})$, we have the following sign conditions:
\begin{enumerate}
\item $\textrm{Tr}\, J|_{\mathcal{E}_2}\gtrless 0$, implies that $\mu_2+n\sqrt{-\mu_1}\gtrless -2\sqrt{-\mu_1}$.
\item $\textrm{det}\, J|_{\mathcal{E}_2}\gtrless 0$, implies that $\mu_2+n\sqrt{-\mu_1}\gtrless 0$.
\end{enumerate}
\end{proposition}
Then the nature of the fixed branch $\mathcal{E}_2$ is determined by the following theorem.
\begin{theorem}[Nature of branch  $\mathcal{E}_2$.]\label{thm-e2}We have the following types for the $\mathcal{E}_2$ branch:
\begin{enumerate}
  \item $\textrm{det}\, J|_{\mathcal{E}_2}= 0,\quad\textrm{Tr}\, J|_{\mathcal{E}_2}>0:$ \emph{bifurcation.}
  \item $\textrm{det}\, J|_{\mathcal{E}_2}= 0,\quad\textrm{Tr}\, J|_{\mathcal{E}_2}<0:$ \emph{cannot happen.}
  \item $\textrm{det}\, J|_{\mathcal{E}_2}> 0,\quad\textrm{Tr}\, J|_{\mathcal{E}_2}>0:$ \emph{source.}
  \item $\textrm{det}\, J|_{\mathcal{E}_2}> 0,\quad\textrm{Tr}\, J|_{\mathcal{E}_2}<0:$ \emph{cannot happen.}
  \item $\textrm{det}\, J|_{\mathcal{E}_2}< 0,\quad\textrm{Tr}\, J|_{\mathcal{E}_2}>0:$ \emph{saddle.}
  \item $\textrm{det}\, J|_{\mathcal{E}_2}< 0,\quad\textrm{Tr}\, J|_{\mathcal{E}_2}<0:$ \emph{saddle.}
  \item $\textrm{Tr}\, J|_{\mathcal{E}_2}=0,\quad\textrm{det}\, J|_{\mathcal{E}_2}> 0:$ \emph{cannot happen.}
  \item $\textrm{Tr}\, J|_{\mathcal{E}_2}=0,\quad\textrm{det}\, J|_{\mathcal{E}_2}< 0:$ \emph{saddle.}
\end{enumerate}
\end{theorem}
\begin{proof}
For 1, using Lemma 1, we find  from (\ref{eig}) that one of the eigenvalues is zero, and so we expect a saddle-node bifurcation on the $\mu_2$-axis when $\mu_2<0$, as it follows directly from the vanishing of the determinant. Items, 2, 4, and 7 follow from (\ref{tr-det2}). For 3, it follows that $\textrm{Tr}\, J|_{\mathcal{E}_1}+\sqrt{\Delta}>0$, so that all eigenvalues are positive, and from Lemma 2 we find a source. Items 5, 6 follow directly from Lemma 2, with parameter conditions being $-2\sqrt{-\mu_1}<\mu_2+n\sqrt{-\mu_1}<0$, and $\mu_2+n\sqrt{-\mu_1}<-2\sqrt{-\mu_1}$ respectively,  while for item 8, we get a saddle from Lemma 2, with $\mu_2+n\sqrt{-\mu_1}<0$. All remaining cases cannot happen because of sign incompatibilities.
\end{proof}

Therefore we find:
\begin{corollary}\label{cor2}
The nature of the fixed branch $\mathcal{E}_2$ is as follows:
\begin{enumerate}
  \item When $\mu_2+n\sqrt{-\mu_1}>0$,  $\mathcal{E}_2$ is a source.
  \item When $\mu_2+n\sqrt{-\mu_1}<0$,  $\mathcal{E}_2$ is a saddle.
  \item When $\mu_2+n\sqrt{-\mu_1}=0$,  $\mathcal{E}_2$ bifurcates.
\end{enumerate}
\end{corollary}
The main conclusion from the results in the last two subsections is that \emph{both} parameters are necessarily  used in order to determine the stability of the fixed branches  $\mathcal{E}_{1,2}$. These  are expected to bifurcate in  saddle-node bifurcations happening on the Eros-$\varepsilon$ axis (this is proven in detail below), with $\mathcal{E}_{1}$ giving a saddle and a sink on the (positive) $\varepsilon_+$-semiaxis, and $\mathcal{E}_{2}$ giving a saddle and a source on the (negative) $\varepsilon_-$-semiaxis.

\subsubsection{Stability of the branch $\mathcal{E}_3$}
For the fixed branch $\mathcal{E}_3$ given by Eq. (\ref{e3}), we find,
\be\label{s1-e3}
\textrm{Tr}\, J|_{\mathcal{E}_3}=-\frac{2\mu_2}{n},\quad\textrm{det}\, J|_{\mathcal{E}_3}=2n\left( \frac{\mu_2^2}{n^2}+\mu_1\right).
\ee
It follows that the sign conditions for the trace and determinant are,
\be
\textrm{det}\, J|_{\mathcal{E}_3}<0,\quad\textrm{on the}\,\, \beta-\textrm{stratum},
\ee
and
\be
\textrm{Tr}\, J|_{\mathcal{E}_3}\gtrless 0,\quad\textrm{if}\quad \mu_2\lessgtr 0,
\ee
which implies that the eigenvalues are $\lambda_+>0,\lambda_-<0$, that is the fixed branch $\mathcal{E}_3$  is a saddle (and so totally unstable).

The important conclusion from this result is that $\mathcal{E}_3$ cannot bifurcate further on the $\beta$-stratum.

To complete the stability analysis of the fixed branch $\mathcal{E}_3$, we now examine the remaining (borderline) case, namely, when $\textrm{det}\, J|_{\mathcal{E}_3}=0.$ From Eq. (\ref{s1-e3}), it follows that this takes us to the Tartarus-$\tau$ curve $\mu_1=-\mu_2^2/n^2$,  and we find the following two conditions (note that $\mu_1<0$),
\beq
\tau_+-branch:\quad\mu_2-n\sqrt{-\mu_1}&=&0,\quad\textrm{and}\quad\mu_2>0,\label{tau+}\label{deg-e3a}\\
\tau_--branch:\quad\mu_2+n\sqrt{-\mu_1}&=&0,\quad\textrm{and}\quad\mu_2<0,\label{tau-}\label{deg-e3b}
\eeq
and so
\be
\mathcal{E}_3\to\mathcal{E}_1,\quad\textrm{on the}\,\,\tau_+-\textrm{branch},
\ee
while
\be
\mathcal{E}_3\to\mathcal{E}_2,\quad\textrm{on the}\,\,\tau_--\textrm{branch}.
\ee
From what we showed in the previous two subsections, we know that the conditions (\ref{tau+}), (\ref{tau-}) are necessary and sufficient for the vanishing of the determinants $\textrm{det}\, J|_{\mathcal{E}_{1,2}}$ respectively.

In other words, we arrive at the following conclusion.
\begin{theorem}
\begin{enumerate}
\item On the $\beta$-stratum and on the $\tau$-curve, the fixed branch  $\mathcal{E}_3$ cannot bifurcate further. In addition, $\mathcal{E}_3$ is a saddle on the $\beta$-stratum, while it becomes the branch $\mathcal{E}_1$ on the $\tau_+$-curve, and the branch $\mathcal{E}_2$ on the $\tau_-$-curve.
\item We expect that there will  further bifurcations of the fixed branches $\mathcal{E}_{1,2}$ on the curves $\tau_+,\tau_-$, both to yield the saddle $\mathcal{E}_3$ and possibly other new branches.
\end{enumerate}
\end{theorem}
A complete proof of this `Theorem' (the second part of it is not yet proved!) will be given in the next Subsection, where we shall prove that these new bifurcations are actually pitchforks. We note that the announced extra bifurcations of the fixed branches $\mathcal{E}_{1,2}$ are not related to the saddle-node bifurcations discussed earlier on the Eros-$\varepsilon$ line, because there we had $\mu_1=0$, whereas here we have  $\mu_1<0$ for both $\mathcal{E}_{1,2}$.

\subsection{Description of the bifurcations}

\subsubsection{General comments}
As we discussed  in the previous Subsection,  we expect that there will two kinds of bifurcations associated with  the fixed branches $\mathcal{E}_{1,2}$, and described in the parameter diagram of Fig. \ref{fig-rem}. In this Subsection, we show that these bifurcations  will be a saddle-node bifurcation occurring `horizontally' when crossing the Eros-$\varepsilon$-axis with dynamics governed by the $\rho$ variable, and a pitchfork bifurcation occurring `vertically' as we cross the Tartara-$\tau$ curve with dynamics governed by the $\sigma$ variable.

In both cases, the dynamics of  the system (\ref{vaA}) is drastically reduced to a one-dimensional centre manifold, and  new equilibrium solutions appear (or disappear!): the branches $\mathcal{E}_{1,2}$ for the saddle-node, and the saddle $\mathcal{E}_{3}$ for the pitchfork. Generally, as we showed previously, the number of equilibria fluctuates  from zero at the chaos-$\chi$ region, to one on the Eros-$\varepsilon$-axis and Gaia-$\gamma$ point, to two on the Uranus and Pontus regions, and to three on the Ourea-$\beta$ region.

The system (re-)visits all these regions as guided by the stability of the unfolding, and consequently, the phase space dynamics is smoothly changing according to the transitions between phase portraits associated with the corresponding strata in parameter space.

\subsubsection{The saddle-node}
We can of course just calculate the centre manifold for the system in this case, however, it is simpler and perhaps more illuminating to proceed as follows. In the present case, the saddle-node is generally speaking a bifurcation involving only the convergence $\rho$, because of the following reason.

The fixed branches $\mathcal{E}_{1,2}$ exist for $\mu_1<0$ only, coalesce (or form) at $\mu_1=0$, and disappear when $\mu_1>0$ (that is on the $\chi$-stratum), and in all cases we always have $\sigma=0$ (the second coordinate of these branches) for these fixed branches to exist. Therefore to find the evolution on the centre manifold, we simply set $\sigma=0$ in (\ref{vaA}) to get the reduced system in the form,
\be\label{rho-eqn}
\dot{\rho}=\mu_1+\rho^2,
\ee
on the centre manifold. This is of course the normal form of a system undergoing a saddle-node bifurcation in the $\rho$-direction at $\mu_1=0$, and so the saddle-node dynamics is convergence-dominated, with evolution equation given by (\ref{rho-eqn}).

In this case therefore, the system has between zero (when $\mu_1>0$) and the two equilibrium solutions $\mathcal{E}_{1,2}$ (when $\mu_1<0$), while their stability is determined by the standard  bifurcation diagrams. On crossing the $\varepsilon_+$-axis (i.e., $\mu_2>0$) from the right (i.e.,  starting from the $\chi$-region and proceeding to the left), the bifurcation occurs on the $\varepsilon_+$-axis, and upon entering the Uranus $o$-stratum we have,
\be\label{+ve}
\mu_2-n\sqrt{-\mu_1}>0,
\ee
as one may easily see. Hence, the emerging fixed branches are, (a) the branch  $\mathcal{E}_{1}$ which  is a saddle by the  Corollary \ref{cor1}, and  (b) the branch $\mathcal{E}_{2}$, which from the Corollary \ref{cor2} is a source, since in this case it follows from (\ref{+ve}) that,  $\mu_2+n\sqrt{-\mu_1}>2n\sqrt{-\mu_1}>0$. Similarly, on crossing the $\varepsilon_-$-axis (i.e., $\mu_2<0$), and entering the Pontus $\pi$-region, we find that,
\be\label{-ve}
\mu_2+n\sqrt{-\mu_1}<0,
\ee
so that $\mathcal{E}_{2}$ is a saddle from Corollary \ref{cor2}, while from Corollary \ref{cor1} we find that,  $\mathcal{E}_{1}$ is a sink there, because $\mu_2-n\sqrt{-\mu_1}<-2n\sqrt{-\mu_1}<0$, as it follows from  (\ref{-ve}). These results obviously also hold in the opposite direction of motion in the parameter diagram.

We conclude that the saddle-node bifurcation  moves the dynamics of the system (\ref{vaA}) along the `fragments',
\be
\chi\to\varepsilon_+\to o,\quad\textrm{or}\quad\chi\to\varepsilon_-\to \pi,
\ee
of the parameter diagram in either direction; that is, if the system finds itself on the region $\chi$, then it is moved to the Uranus-$o$ region if $\mu_2>0$, or to the Pontus-$\pi$ region if $\mu_2<0$, by crossing the corresponding half-line on the Eros-$\varepsilon$ axis,   and vice versa. The result is a continuous transformation of the corresponding phase portraits followed by a creation or annihilation of the fixed branches $\mathcal{E}_{1,2}$ as shown above. Therefore we have proved the following Theorem.
\begin{theorem}\label{s-nNPR}
On crossing the $\varepsilon$-stratum, we have the following saddle-node bifurcations:
\begin{enumerate}
\item Fragment $\chi\to\varepsilon_+\to o$: The system moves from zero equilibria on $\chi$, to one on $\varepsilon_+$, to two the fixed branches $\mathcal{E}_{1,2}$ on $o$, such that the branch  $\mathcal{E}_{1}$   is a saddle, and the branch $\mathcal{E}_{2}$ is a source.
    \item Fragment $\chi\to\varepsilon_-\to \pi$: The system moves from zero equilibria on $\chi$, to one on $\varepsilon_-$, to two the fixed branches $\mathcal{E}_{1,2}$ on $\pi$, such that the branch $\mathcal{E}_{2}$ is a saddle, while $\mathcal{E}_{1}$ is a sink.
    \item The evolution equation on the centre manifold for both of these saddle-node bifurcations is given by Eq. (\ref{rho-eqn}).
\end{enumerate}
Given the above directions of $\varepsilon$-crossings, the two new fixed branches $\mathcal{E}_{1,2}$  are created at the crossings. These results hold true in the opposite crossing directions, namely, for the fragments $o\to\varepsilon_+\to\chi$, and  $\pi\to\varepsilon_-\to\chi$, where the two fixed branches are now annihilated at the crossings.
\end{theorem}

We note that a similar result about the role of the saddle-node bifurcation in cosmology was recently found in the more restricted case of the versal unfolding of the Friedmann equations \cite{cot23}.

\subsubsection{The pitchfork}
In this subsection, we shall be interested in the properties of the flow of Eq. (\ref{vaA}) in a neighborhood of the tartarus $\tau$-stratum (cf. Fig. \ref{fig-rem}), and prove the following theorem for the system (\ref{vaA}) about the existence of further  pitchfork bifurcations of the two fixed branches $\mathcal{E}_{1,2}$ when $\mu_1<0$.
\begin{theorem}\label{pitch-thm}
When $\mu_2\neq 0$ and with $\mu_1$ decreasing, the system (\ref{vaA}) undergoes a pitchfork bifurcation on the tartarus $\tau$-stratum such that:
\begin{enumerate}
\item the pitchfork  dynamics reduced on the centre manifold is shear-dominated, with evolution equation given by,
\be
\dot{\sigma}=\pm\,\epsilon\, n\, \sigma-\frac{n^2}{2\mu_2}\sigma^3,
\ee
where $\epsilon$ is a variational parameter that describes  transversal crossings of the $\tau$-stratum, and the upper (resp. lower) sign corresponds to the fixed branch  $\mathcal{E}_{1}$ (resp. $\mathcal{E}_{2}$).
\item  a supercritical pitchfork bifurcation of the fixed branch  $\mathcal{E}_{1}$ when crossing the $\tau$-curve with $\mu_2>0$ (i.e., the upper part of $\tau$), to the appearing fixed branch saddle $\mathcal{E}_{3}$, and
\item a subcritical pitchfork bifurcation of the fixed branch  $\mathcal{E}_{2}$ when crossing the $\tau$-curve  with $\mu_2<0$ (i.e., the lower part of $\tau$), to the appearing fixed branch saddle $\mathcal{E}_{3}$.
\end{enumerate}

\end{theorem}

\begin{remark}
On the tartarus-$\tau$ curve, $\tau=\left\{ (\mu_1,\mu_2): \mu_1=-\mu_2^2/n^2\right\}$,  we have, $\textrm{det}\, J|_{\mathcal{E}_3}=0$, and so the eigenvalues of the linearized Jacobian at $\mathcal{E}_3$ are, $\lambda_+=\textrm{Tr}\, J|_{\mathcal{E}_3}, \lambda_-=0$ (cf. (\ref{s1-e3})). Therefore a centre manifold analysis is required to compute the dynamics of the bifurcations to $\mathcal{E}_3$ on the centre manifold.
\end{remark}

Since the proof of Theorem \ref{pitch-thm} is somewhat long, we shall break it down to a number of different steps.

\underline{Step-1:} \emph{Jordan form of linear part of Eqn. (\ref{vaA}).}

Assuming $\mu_1<0$, since the $0-$eigendirection is associated with the $\sigma$ variable (rather than $\rho$), we move the fixed `point' $\mathcal{E}_3$  to the origin, by setting $\rho_*=\mp\sqrt{-\mu_1},\sigma_*=0$, and letting,
\be
\xi=\rho-\rho_*=\rho\pm\sqrt{-\mu_1},
\ee
and rewrite  the original system (\ref{vaA}) in the form,
\begin{equation}\label{xi-eqn}
\begin{split}
\dot{\sigma}&=\mu_2\sigma+n\sigma (\xi\mp\sqrt{-\mu_1}),\\
\dot{\xi}&=\mu_1+(\xi\mp\sqrt{-\mu_1})^2+\sigma^2.
\end{split}
\end{equation}
On the tartarus curve $-\mu_1=\mu_2^2/n^2$, we have,
\be
\sqrt{-\mu_1}=\frac{|\mu_2|}{n}.
\ee
Now suppose we cross the tartarus $\tau$-curve transversally at some constant but nonzero $\mu_2$ with decreasing $\mu_1$. This can be described by varying the $\epsilon$ defined by,
\be\label{mu-ep}
\sqrt{-\mu_1}=\frac{|\mu_2|}{n}-\epsilon,
\ee
so as to cross the $\tau$-parabola from right to left and vertically (decreasing $\mu_1$). In this way, we regard $\epsilon$ as a parameter replacing $\mu_2$.  Substituting (\ref{mu-ep}) into (\ref{xi-eqn}) we find,
\begin{equation}\label{xi-si}
\begin{split}
\dot{\sigma}&=n\,\sigma\,\xi\pm n\,\sigma\,\epsilon,\quad :=f_1, \\
\dot{\xi}&=\xi^2+\sigma^2+2\left(\frac{\mu_2}{n}\pm\epsilon\right)\xi,\quad :=f_2.
\end{split}
\end{equation}
Therefore we can write this system in the form, $\dot{X}=J_{|(0,0)}X+(X,\epsilon),\dot{\epsilon}=0,$ where $X=(\sigma,\xi)^\top$, that is with its linear part in normal form. We find,
\be \label{linpartOK}
\begin{split}
\left(
  \begin{array}{c}
    \dot{\sigma} \\
    \dot{\xi} \\
  \end{array}
\right)
&=
\left(
  \begin{array}{cc}
    0 & 0 \\
    0 & 2\mu_2/n \\
  \end{array}
\right)
\left(
  \begin{array}{c}
    \sigma \\
    \xi \\
  \end{array}
\right)
+
\left(
  \begin{array}{c}
    n\,\sigma\,\xi\pm n\,\sigma\,\epsilon\\
\sigma^2+\xi^2\pm 2\epsilon\xi \\
  \end{array}
\right),\\
\dot{\epsilon}&=0.
\end{split}
\ee
This is the form we need in order to apply the centre manifold theorem (cf. any of the standard references on bifurcation theory in the bibliography).

\underline{Step-2:} \emph{Evolution on the centre manifold.}

The centre manifold to (\ref{linpartOK}) will be the set,
\be \label{cm1}
W^c=\left\{(\sigma,\xi,\epsilon)|\, \xi=h(\sigma,\epsilon) \right\},
\ee
with $h(0,0)=0, Dh(0,0)=0$. This  is of the form $Y=h(X)$, where, $Y=\xi,X=(\sigma,\epsilon)^\top$, subject to the tangency condition, $Dh\dot{X}-\dot{Y}=0$. Setting,
\be\label{h}
h(\sigma,\epsilon)=\alpha\sigma^2+\beta\sigma\epsilon+\gamma\epsilon^2,
\ee
the problem is to find the coefficients. Using Eq. (\ref{linpartOK}), the tangency condition becomes,
\be\label{tan}
(2\alpha\sigma+\beta\epsilon,\beta\sigma+2\gamma\epsilon)\left(
  \begin{array}{c}
    n\sigma h\pm n\sigma\epsilon\\
    0 \\
  \end{array}
\right) -\frac{2\mu_2}{n}h-g=0,
\ee
with $g:=\sigma^2+\xi^2\pm 2\epsilon\xi$. Using (\ref{h}), and balancing the various terms in (\ref{tan}), we obtain:
\begin{equation}\label{coef1}
\begin{split}
\sigma^2-\textrm{terms}:&\quad \alpha=-\frac{n}{2\mu_2}, \\
\epsilon\sigma-\textrm{terms}:&\quad \beta=0,\\
\epsilon^2-\textrm{terms}:&\quad \gamma=0.
\end{split}
\end{equation}
Therefore the centre manifold (\ref{cm1}) is the graph of
\be \label{cm11}
\xi=h(\sigma,\epsilon)=-\frac{n}{2\mu_2}\sigma^2+O(3),
\ee
The evolution equation on the centre manifold can  now be read off directly from Eq.  (\ref{linpartOK}), namely,
\be \label{si-eqn}
\dot{\sigma}=\pm\epsilon\,n\,\sigma-\frac{n^2}{2\mu_2}\sigma^3.
\ee
We note that the upper (resp. lower) sign is for the $\mathcal{E}_{1}$  (resp. $\mathcal{E}_{2}$) branch. This equation describes the normal form of a system undergoing a pitchfork bifurcation at $\epsilon=0$. As we noted earlier, we regard $\epsilon$ as a parameter and $\mu_2$ as a nonzero constant.

From  (\ref{linpartOK}), we may also obtain the $\xi$-evolution equation, namely,
\be \label{xi-eqn2}
\dot{\xi}=2\frac{\mu_2}{n}\xi,
\ee
This completes the proof of the first part of the Theorem \ref{pitch-thm}.
For the second and third parts of the Theorem, we proceed as follows.

\underline{Step-3:} \emph{Pitchfork evolution for the fixed branch $\mathcal{E}_{1}$.}

Let us first suppose that $\mu_2>0$, and decreasing $\mu_1$. Then  $\mathcal{E}_{1}$ undergoes a pitchfork bifurcation on the $\tau_+$-curve $\mu_2=n\sqrt{-\mu_1}$ (cf. Corollary \ref{cor1}), and
the evolution equation (\ref{si-eqn}) becomes a \emph{supercritical} equation  for $\epsilon>0$,
\be \label{si-eqn1}
\dot{\sigma}=\epsilon\,n\,\sigma-\frac{n^2}{2\mu_2}\sigma^3.
\ee
This is called supercritical because the two new orbits created at
\be\label{s*1}
\sigma_*=\pm \sqrt{\frac{2\mu_2\epsilon}{n}},\quad\mu_2>0,
\ee
are stable: for denoting by $f$ the right-hand-side of (\ref{si-eqn1}), we find the $\sigma_*$ orbits have  eigenvalues $Df(\sigma_*)=-2n\epsilon$, which is negative for $\epsilon>0$. We note that only the positive sign orbit (\ref{s*1}) is admitted here because we have taken $\sigma>0$. So on the centre manifold, this equilibrium attracts all orbits.

In addition, for $\epsilon>0$, all orbits are repelled from the unstable equilibrium at the origin (the equilibrium at the origin always exists for (\ref{si-eqn1}), and is stable for $\epsilon<0$). The bifurcation diagram is therefore as in Fig. \ref{super-bif}, and the stability diagrams of bifurcations on the centre manifold are shown in Fig. \ref{cm-mu2+}. We see from these diagrams that the fixed branch $\mathcal{E}_{1}$ attracts all nearby orbits in the $\beta$-stratum.
\begin{figure}
\centering
\includegraphics[width=0.3\textwidth]{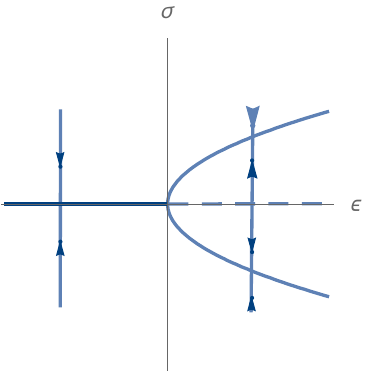}
\caption{The bifurcation diagram for the supercritical case. The origin is unstable and repels all nearby orbits for $\epsilon>0$.}\label{super-bif}
\end{figure}

\begin{figure}
     \centering
     \begin{subfigure}[b]{0.3\textwidth}
         \includegraphics[width=\textwidth]{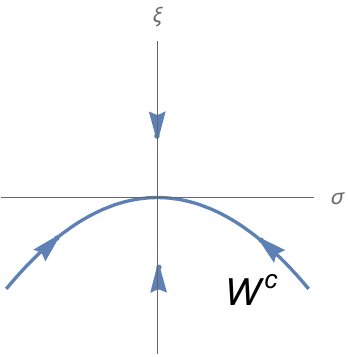}
         \caption{$\epsilon<0$}
         \label{cm1-m2+a}
     \end{subfigure}
     \hfill
     \begin{subfigure}[b]{0.3\textwidth}
         \includegraphics[width=\textwidth]{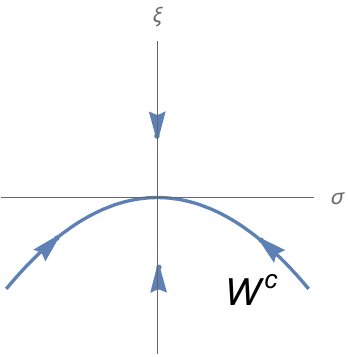}
         \caption{$\epsilon=0$}
         \label{cm1-m2+b}
     \end{subfigure}
     \hfill
     \begin{subfigure}[b]{0.3\textwidth}
         \includegraphics[width=\textwidth]{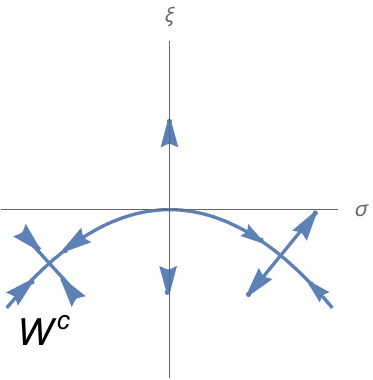}
         \caption{$\epsilon>0$}
         \label{cm1-m2+c}
     \end{subfigure}
             \caption{The centre manifold in the case $\mu_2>0$.}
        \label{cm-mu2+}
\end{figure}

\underline{Step-4:} \emph{Pitchfork evolution for the fixed branch $\mathcal{E}_{2}$.}

The case $\mu_2<0$ is treated by similar methods, but the results are different, in that here for  $\epsilon>0$ all orbits blow up in a finite time as there is no opposing of the cubic term as in the supercritical case. We take $\epsilon<0$, so that $\mu_1$ is decreasing in this case, and the $(-)$ sign in the evolution eqn (\ref{si-eqn}) to examine  $\mathcal{E}_{2}$. The centre manifold graph is given by,
\be \label{cm21}
\xi=-\frac{n}{2|\mu_2|}\sigma^2+O(3),
\ee
and the evolution equation for the shear on $W^c$ reads,
\be \label{si-eqn-sub}
\dot{\sigma}=|\epsilon|\,n\,\sigma+\frac{n^2}{2|\mu_2|}\sigma^3.
\ee
This is therefore a \emph{subcritical} pitchfork bifurcation for $\mathcal{E}_{2}$ taking place on the $\tau_-$-curve $\mu_2=-n\sqrt{-\mu_1}$ (cf. Corollary \ref{cor2}), implying that the origin is a stable equilibrium in this case when $\epsilon<0$. The $\xi$-evolution is given by,
\be \label{xi-eqn3}
\dot{\xi}=-2\frac{|\mu_2|}{n}\xi,
\ee
these being the vertical directions to the centre manifold.

From Eq. (\ref{si-eqn-sub}) with $\epsilon<0$, it follows that the new orbits are created,
\be\label{s*2}
\sigma_*=\pm \sqrt{\frac{2|\mu_2||\epsilon|}{n}},\quad\mu_2<0,
\ee
(we again keep only the positive one as we take $\sigma>0$) and are unstable on the centre manifold, whereas the origin is stable. (We note that when $\epsilon>0$, there are no new orbits created, and also the origin becomes unstable.) These results  lead to the Figs. \ref{sub-bif}, \ref{cm-mu2-}. We see from these diagrams that the fixed branch $\mathcal{E}_{2}$ is a source repelling all nearby orbits.
 \begin{figure}
\centering
\includegraphics[width=0.3\textwidth]{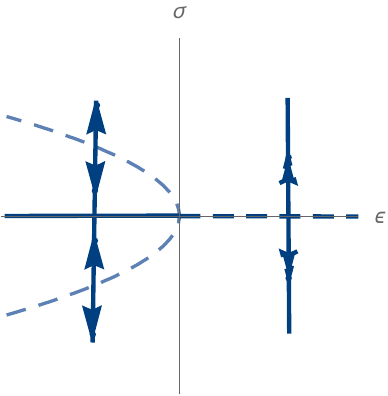}
\caption{The bifurcation diagram for the subcritical case. The origin is stable and attracts all nearby orbits for $\epsilon<0$, but the newly created orbit is stable on the centre manifold.}\label{sub-bif}
\end{figure}
\begin{figure}
     \centering
     \begin{subfigure}[b]{0.3\textwidth}
         \includegraphics[width=\textwidth]{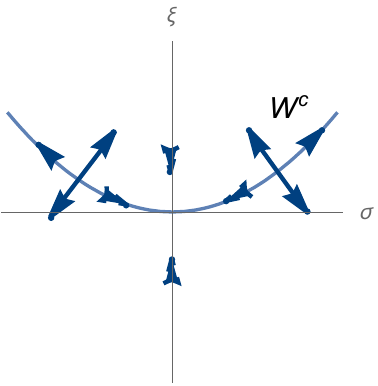}
         \caption{$\epsilon<0$}
         \label{cm1-m2-a}
     \end{subfigure}
     \hfill
     \begin{subfigure}[b]{0.3\textwidth}
         \includegraphics[width=\textwidth]{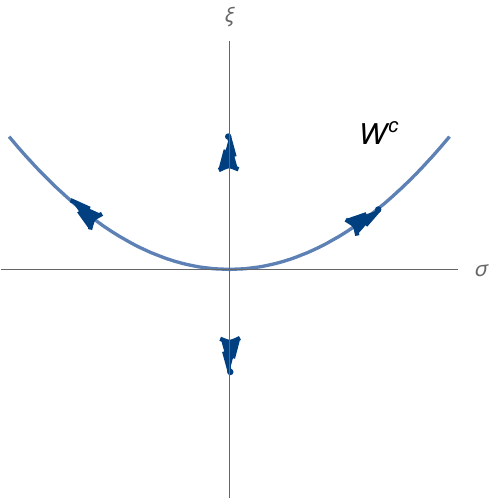}
         \caption{$\epsilon=0$}
         \label{cm1-m2-b}
     \end{subfigure}
     \hfill
     \begin{subfigure}[b]{0.3\textwidth}
         \includegraphics[width=\textwidth]{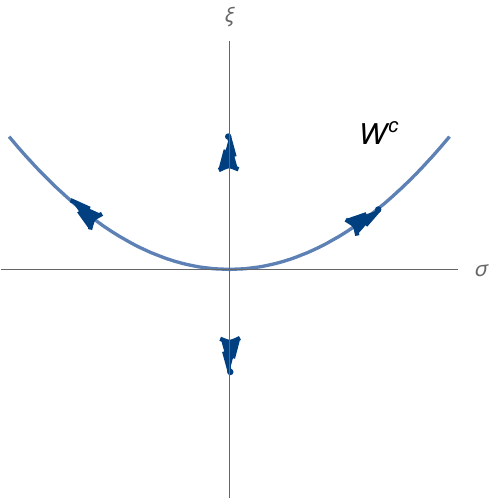}
         \caption{$\epsilon>0$}
         \label{cm1-m2-c}
     \end{subfigure}
             \caption{The centre manifold in the case $\mu_2<0$. The vertical directions show the stability of the origin, whereas the stability of the centre manifold evolution is shown tangentially.}
        \label{cm-mu2-}
\end{figure}

This concludes the proof of the Theorem \ref{pitch-thm}.

\subsection{The bifurcation diagram for the NPR-system}
Assembling all results on the various bifurcations found in the previous subsections, we arrive at the bifurcation diagram \ref{bifn}. This is the parameter diagram together with all phase diagrams showing the dynamics on the different strata superimposed.

 \begin{figure}
\centering
\includegraphics[width=\textwidth]{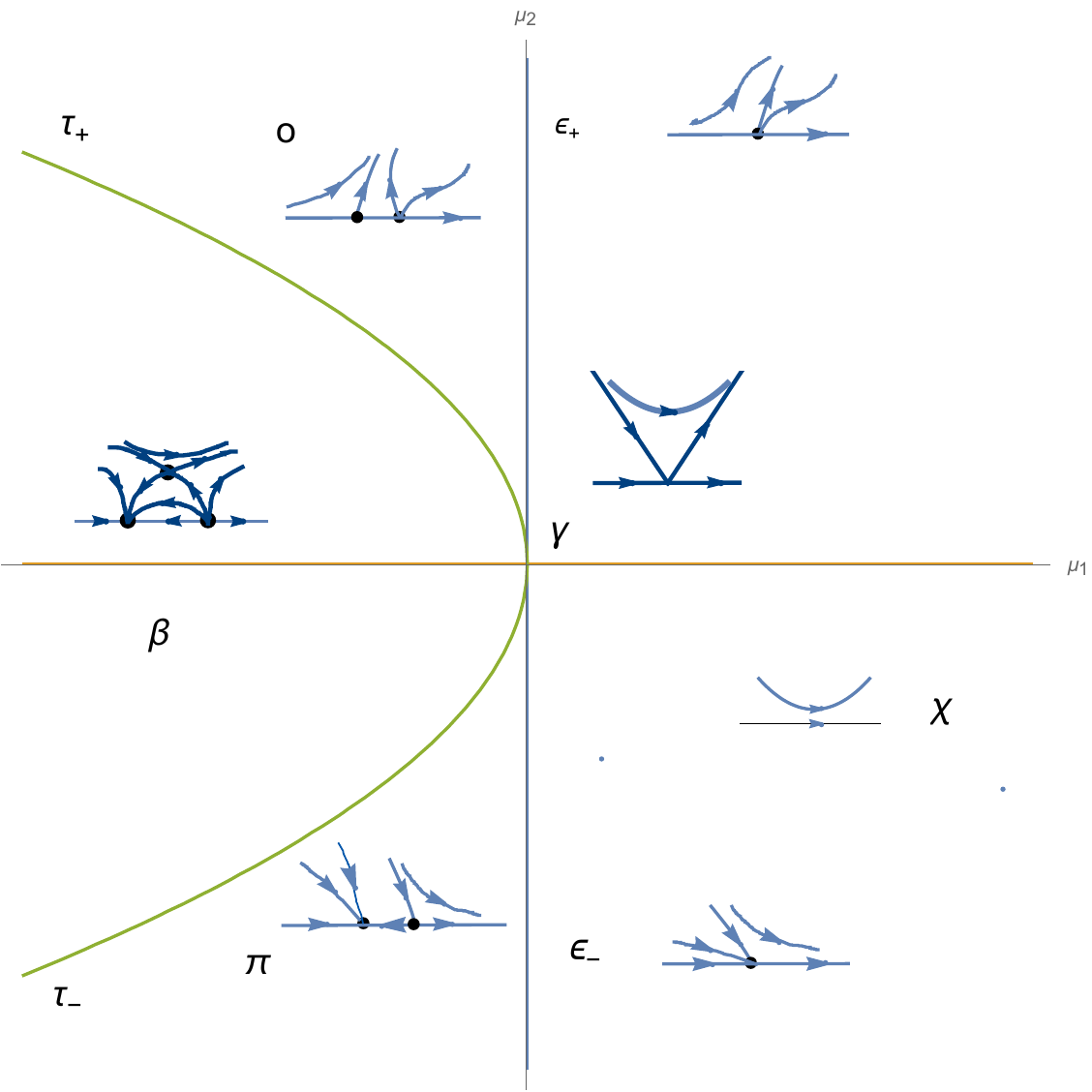}
\caption{The complete bifurcation diagram for the NPR-system. }\label{bifn}
\end{figure}

\section{The bifurcation diagram for the CV system}\label{bif2}
In this Section, we construct the bifurcation diagram of the convergence-vorticity system, that is we take into account the possible effects of rotation when considering the dynamics of the convergence function.

\subsection{Versal unfolding}
We now move on to the examination of the convergence-vorticity problem (\ref{v}) (`CV'-system hereafter), and as with the treatment of the NPR-system, we start with the non-generic situation obtained by setting $\mathcal{R}=0$ in (\ref{v}), and arrive at,
\begin{equation}\label{v0}
\begin{split}
\dot{\rho}&=\rho^2-\omega^2 ,\\
\dot{\omega}&=n\rho\,\omega ,
\end{split}
\end{equation}
where $n=2,3$ as before. As with the system (\ref{3}),  the non-generic convergence-vorticity system (\ref{v0}) is already in normal form, and has a $\mathbb{Z}_2$-symmetry $\rho\to\rho,\omega\to -\omega$.

The versal unfolding for systems with a $\mathbb{Z}_2$-symmetry have been found and is given by (cf. e.g., \cite{golu3}: Sections XIX.1-3, \cite{kuz}: Section 8.5.2, \cite{wig}: Sections 20.7, 33.2, \cite{gh83}: Sect. 7.4, \cite{ar94}: Section 4.4, and refs. therein for complete proofs and further references),
\begin{equation}\label{vaB}
\begin{split}
\dot{\rho}&=\mu_1+\rho^2-\omega^2,\\
\dot{\omega}&=\mu_2 \omega+n\rho\omega.
\end{split}
\end{equation}

For the ensuing analysis of the system (\ref{vaB}),  we shall use the same symbols for the parameter $\mu$ as with the NPR-system of the previous Section. Although the method of analysis is analogous to that of the unfolding (\ref{vaA}), we shall see that the results are dramatically different than those in the previous Section for the NRP-system, the main new phenomenon being the appearance of a stable limit cycle associated with the dynamics of (\ref{vaB})\footnote{We emphasize that limit cycles are \emph{isolated} closed orbits in phase space, and can only appear in nonlinear systems. Linear systems may of course have periodic solutions, but these can never be isolated (such solutions are always surrounded by an infinite number of others). A stable limit cycle attracts all nearby orbits.}.

We give here the full parameter diagram for the convergence-vorticity case. As we show below, in the convergence-vorticity case, there is a more complicated stratification of the parameter space in that there are now four new strata, namely, the positive $\mu_1$-axis, that we call the $E\rho$-line, the $\nu$-line in the fourth-quarter of the parameter space, and the regions $\alpha, \eta$.

\begin{figure}
\centering
\includegraphics[width=0.3\textwidth]{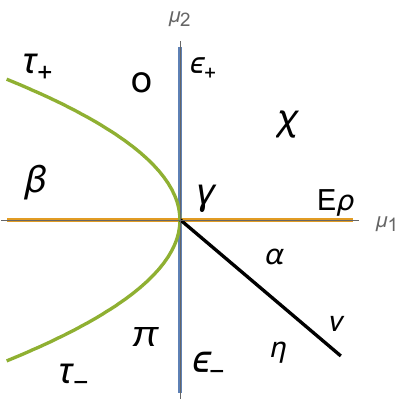}
\caption{The complete picture of the  strata partitioning the parameter diagram in the convergence-vorticity case. We have  eleven strata, namely, $\chi,\gamma,E\rho,\alpha,\nu,\eta,\varepsilon, o,\pi,\tau,\beta$.}\label{fig-rem2}
\end{figure}

\subsection{The zero-parameter case}
To start the examination of the zero-parameter system (\ref{v0}), a similar analysis as that  Section \ref{0-npr} suggests that the invariant line condition (\ref{inv-line}) cannot be fulfilled because we now obtain  $m^2=1-n$, and this  is negative for both the timelike and null congruences.

Therefore there are no invariant  lines (Stewart separatrices) in the vorticity case. However, the phase portrait in this case can be easily drawn by considering the first integral of the system (\ref{v0}), namely, the function
\be
i_n(\rho,\omega)=\frac{n}{2\omega^{2/n}}\left(\frac{\omega^2}{n-1}+\rho^2\right),
\ee
along solutions, with the level curves of $i_n$ providing all the phase curves of  (\ref{v0}). These are given by,
\be
\frac{\omega^2}{n-1}+\rho^2=\frac{2}{n}c\omega^{2/n},\quad c\in\mathbb{R}.
\ee
 For instance, for the null case ($n=2$),  the family,
\be
\rho^2+\omega^2=c\omega,
\ee
provides these curves, and similarly for the timelike case,
\be
\rho^2+\frac{\omega^2}{2}=\frac{2}{3}c\omega.
\ee
The phase diagram is as in Fig. \ref{cv-0}. We note that the direction of the flow in Fig. \ref{cv-0} can still be found by the method used before for the NPR-system.
 \begin{figure}
\centering
\includegraphics[width=0.3\textwidth]{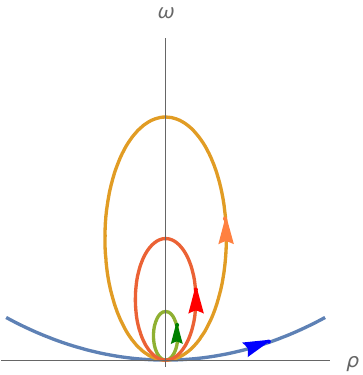}
\caption{The phase portrait for zero parameter case of the convergence-vorticity system. }\label{cv-0}
\end{figure}

\subsection{General comments on the stability and bifurcations of the fixed branches}
The versal unfolding (\ref{vaB}) has three fixed branches, namely,
\begin{enumerate}
  \item $\mathcal{E}_{1,2}=(\mp\sqrt{-\mu_1},0)$. These are real, provided
  \be \mu_1<0.\ee
  \item The third fixed branch is,
\be\label{e3}
\mathcal{E}_{3}=\left(-\frac{\mu_2}{n},\sqrt{\frac{\mu_2^2}{n^2}+\mu_1}\right),
\ee
which is real if,  $\mu_2^2/n^2+\mu_1>0$.
\end{enumerate}
The linearized Jacobian of (\ref{vaB}) is given by,
\be
J=\left(
  \begin{array}{cc}
    2\rho & -2\omega \\
    n\omega & \mu_2 +n\rho\\
  \end{array}
\right).
\ee
Since $\mathcal{E}_{1,2}$ are identical to the corresponding fixed branches of the NPR-system, the analysis of Section \ref{stab1} holds true for the present case. In particular,  the two general lemmas of that Section continue to hold as well as the stability analyses of the branches  $\mathcal{E}_{1,2}$ pass on to the present case, specifically the two main Theorems \ref{thm-e1}, \ref{thm-e2}, and the two Corollaries \ref{cor1}, \ref{cor2}. (We note that this analysis depends only on $n$, not on the sign of the coefficients of the $\sigma^2$ or $\omega^2$ terms.)

However, for the fixed branch $\mathcal{E}_{3}$, although the forms for the trace and determinant are identical to (\ref{s1-e3}),  the sign of the determinant satisfies,
\be\label{det-e3+}
\textrm{det}\, J|_{\mathcal{E}_3}> 0,\quad\textrm{on the}\,\, \beta-\textrm{stratum},
\ee
because of the condition $\mu_2^2/n^2+\mu_1>0$\footnote{We note that the vanishing of this expression on the $\tau$-curve implies the degeneration of $\mathcal{E}_{3}$ into $\mathcal{E}_{1,2}$, exactly as before, cf. (\ref{deg-e3a}), (\ref{deg-e3b}).}. This means that in the vorticity case, the fixed branch $\mathcal{E}_{3}$ cannot be a saddle anymore, but it is a \emph{node}: Using the lemma \ref{2ndlemma}, we find,
\be
\textrm{Tr}\, J|_{\mathcal{E}_3}\gtrless 0,\quad\textrm{if}\quad \mu_2\lessgtr 0,
\ee
that is $\mathcal{E}_{3}$ is a source when $\mu_2<0$, and a sink when $\mu_2>0$. In this case, we find from Eq. (\ref{eig}) that the eigenvalues are real, and using the condition (\ref{det-e3+}), we obtain from (\ref{eig}) that  the parameters must satisfy,
\be \label{real-eig}
\mu_1<(1-2n)\frac{\mu_2^2}{2n^3}.
\ee
This implies that $\mu_1<0$, for both the null and timelike cases.

To proceed further, we note that in all regions of the parameter plane $(\mu_1,\mu_2)$ where the condition (\ref{real-eig}) is violated, that is when,
\be \label{c-eig}
\mu_1>(1-2n)\frac{\mu_2^2}{2n^3},
\ee
the eigenvalues of the linearized Jacobian for  $\mathcal{E}_{3}$ are complex conjugate of the form,
\be
a\pm i b,
\ee
with
\be \label{a}
a=-\frac{\mu_2}{n}.
\ee
This means that in this case (also on the $\varepsilon$-axis, where $\mu_1=0$),
\begin{enumerate}
\item $\mathcal{E}_{3}$ is an unstable node for $a>0,$ (i.e., $\mu_2<0$),
\item $\mathcal{E}_{3}$ is a stable node $a<0$, (i.e., $\mu_2>0$).
\end{enumerate}
Therefore the bifurcations of the fixed branches $\mathcal{E}_{1,2}$ are as for the NPR-system:  we have a saddle-node bifurcation of the Eros $\varepsilon$-line, and a pitchform bifurcation on the $\tau$-curve. These are described exactly as in the previous Section, and we shall not repeat that analysis.  The only difference to the previous case of the NPR-system is that in the present case the pitchfork bifurcations of the branches $\mathcal{E}_{1,2}$ \emph{are to a node instead of a saddle}, a stable node (sink) for $\mu_2>0$, and an unstable node (source) for $\mu_2<0$.

There is, however, one remaining case for $\mathcal{E}_{3}$ - now possible in the vorticity case - that is, when the eigenvalues from Eq. (\ref{eig}) are \emph{purely imaginary} ($a=0$ above): On the Erebus $E\rho$-line, $\mu_1>0,\mu_2=0$, we find,
\be\label{centre}
\lambda_\pm =\pm i\sqrt{2n\mu_1}.
\ee
In this case, we expect a \emph{degenerate} Hopf bifurcation for the system (\ref{vaB}) when $\mu_1>0$, that is on the $E\rho$-line, there exist an infinite number of closed orbits surrounding the centre (cf. Eq. (\ref{centre})). The proof of this statement is very similar to the construction of the phase curves of Fig. \ref{cv-0}, and the result holds  true for both the null and timelike cases.

However, as we shall describe below, it is a fundamental result in this case that the inclusion of higher-order terms in the versal unfolding (\ref{vaB}) stabilizes the situation and leads to a non-degenerate Hopf bifurcation and to a unique stable limit cycle in the $\alpha$-stratum. Furthermore, this limit cycle disappears on the $\nu$-curve  and $\eta$-region (see Section \ref{ulc} below).

We shall now describe in some detail the phase portraits in each of the  strata of Fig. \ref{fig-rem2}.

\subsection{Detailed description of the vorticity-induced bifurcations}\label{vort-detail}
The final bifurcation diagram of the convergence-vorticity system (\ref{vaB}) is given in Fig. \ref{bifnOmega}, and can be deduced by gathering together all the results we have so far in this Section, and attaching the corresponding phase diagrams to each of the strata shown in Fig. \ref{fig-rem2}. This is discussed below for the various cases taking into account all previous results, cf. Fig. \ref{bifnOmega}.

 \begin{figure}
\centering
\includegraphics[width=\textwidth]{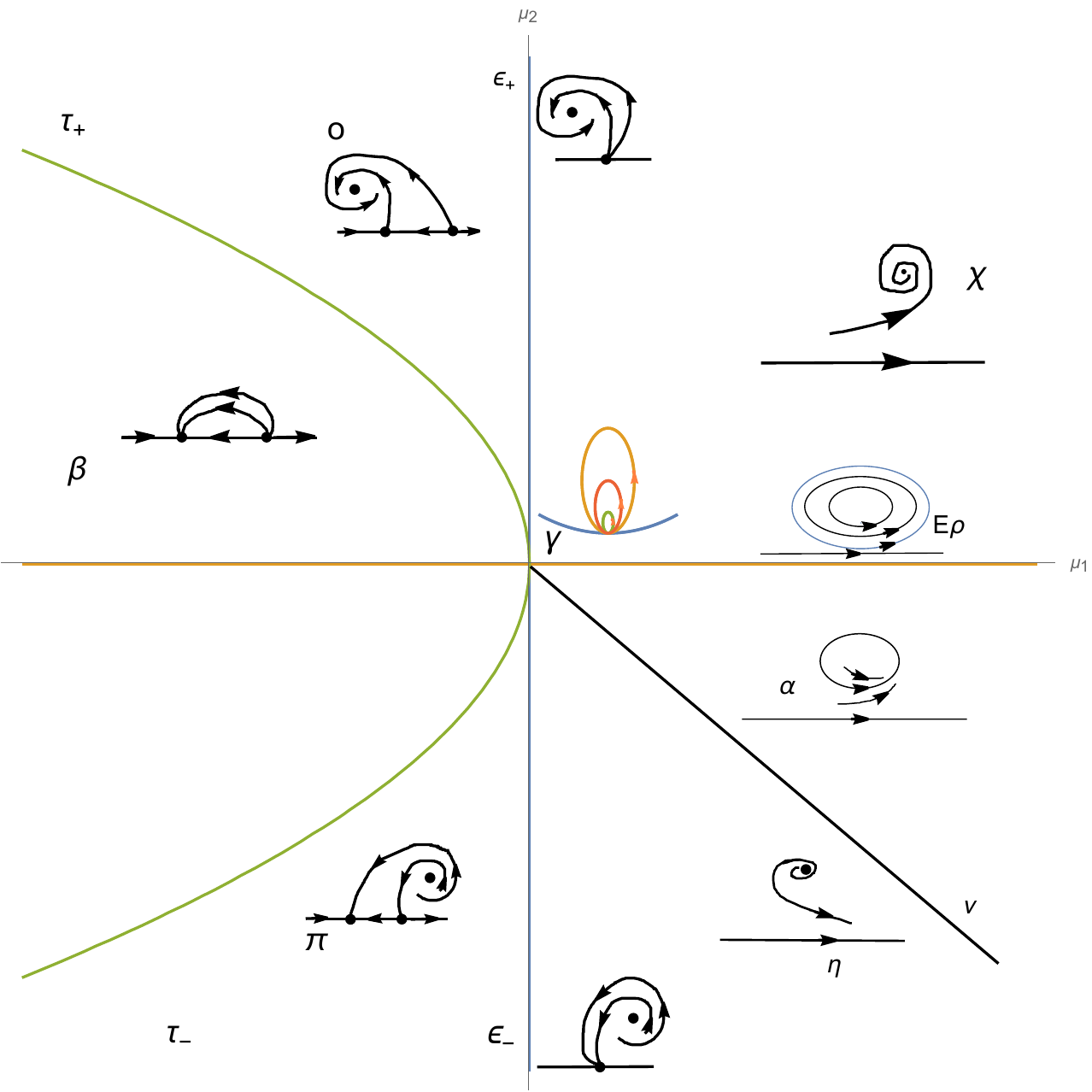}
\caption{The complete bifurcation diagram for the convergence-vorticity system. }\label{bifnOmega}
\end{figure}

The basic difference between the present case and that of the NPR-system is that presently, in the region of the parameter plane with $\mu_1>0$, the fixed branch $\mathcal{E}_{3}$ bifurcates further and affects the global topology of the orbits, despite the fact that the fixed branches $\mathcal{E}_{1,2}$ bifurcate as before in the other regions of the plane. This leads to a totally new and complementary behaviour to the convergence-shearing aspect met in the NPR-system (in the latter we had only the focusing behaviour present in the $\chi$-region, instead of the various strata now appearing). This novel behaviour is apparently linked to the nature of the non-hypersurface-orthogonal congruences associated with the vorticity case.

\subsubsection{Stratum Gaia-$\gamma$}
We have already plotted the phase portrait in the zero-parameter case in Fig. \ref{cv-0}, corresponding to the origin of the parameter plane, the Gaia-$\gamma$ point. One observes the absence of the Stewart lines and similar `focusing' behaviour, now replaced by the closed orbits shown in the phase portrait.

\subsubsection{Stratum Chaos-$\chi$}
In this stratum (i.e., in the first quadrant of the $(\mu_1,\mu_2)$-plane) there is only the fixed branch $\mathcal{E}_{3}$, and no $\mathcal{E}_{1,2}$ branches. Since $\mu_1,\mu_2>0$,  $a<0$ from Eq. (\ref{a}), and since the eigenvalues are complex conjugate, $\mathcal{E}_{3}$ is a stable focus.

\subsubsection{Stratum Eros-$\varepsilon$}
On the Eros $\varepsilon$-line $\{\mu_1=0\}$, both $\mathcal{E}_{1,2}$ become degenerate at the origin, while $\mathcal{E}_{3}=(-\mu_2/n,|\mu_2/n|)$. Then on the $\varepsilon_+$-line, where $\mu_2>0$, $\mathcal{E}_{3}$ is a sink lying in the second quadrant with all orbits approaching it, and on the $\varepsilon_-$-line, where $\mu_2<0$, $\mathcal{E}_{3}$ is a source lying in the second quadrant with all orbits repelled towards the origin. This is described  in  Fig. \ref{bifnOmega} in the two phase portraits drawn near the $\varepsilon$-strata.

\subsubsection{Strata Uranus-$o$, Pontus-$\pi$}
On both of these regions, the condition (\ref{real-eig}) always holds, and  we get a similar situation as with that on the Eros-$\varepsilon$ line, but this time with $\mathcal{E}_{1}=(-\sqrt{-\mu_1},0)$, and $\mathcal{E}_{2}=(-\sqrt{\mu_1},0)$, cf. these strata in Fig. \ref{bifnOmega}. In the $o$-stratum, $\mathcal{E}_{1}$ is a saddle and $\mathcal{E}_{2}$ a source, whereas in the Pontus-$\pi$ stratum $\mathcal{E}_{1}$ is a sink and $\mathcal{E}_{2}$ a saddle as before, cf. Fig. \ref{bifnOmega}.

\subsubsection{Stratum Ourea-$\beta$, tartara-$\tau$ curve}
On the tartara-$\tau$ curve as well as in the Ourea-$\beta$ region,  there is no $\mathcal{E}_{3}$ branch. On the $\tau$-curve,  the two fixed branches $\mathcal{E}_{1,2}$ bifurcate as before, however, instead of getting a saddle we have now  two nodes in the $\beta$-stratum, one unstable for $\mu_2<0$, and one stable when $\mu_2>0$ leading to the corresponding portrait in these two regions, Ourea-$\beta$, $\tau_+,\tau_-$ curves in Fig. \ref{bifnOmega}.

This leaves us with the situation concerning the Hopf bifurcation when $\mu_1>0$, which is examined below.

\subsubsection{Erebus-$E\rho$-line, Aether-$\alpha$ stratum, Nyx-$\nu$-curve, Chemera-$\eta$ stratum}\label{ulc}
On the Erebus-$E\rho$-line there is a degenerate Horf bifurcation of the $\mathcal{E}_{3}$ branch, as a result of which,  there is an accompanied  infinity of closed orbits. This situation is stabilized by the inclusion of higher-order (only cubic suffices) terms on the stratum Aether-$\alpha$, and a \emph{stable limit cycle} appears there. This however, disappears on the Nyx-$\nu$-curve, leaving an unstable focus in the Chemera-$\eta$ stratum.

These results are discussed in \cite{zol84}, and in \cite{kuz}, \cite{wig}, \cite{gh83}, \cite{ar83}, and references therein, where descriptions of the  proofs, and especially that of the uniqueness of the stable limit cycle, may be found. The phase portraits portraits are given in the Fig. \ref{bifnOmega}.

The meaning of the appearance or disappearance of the stable limit cycle during the Hopf bifurcations $\chi\to\eta$ is related to the phenomena of mild and hard loss of stability (cf. \cite{ar86} for a general overview, and \cite{ar83}, \cite{ar94} for details). During evolution in the bifurcation fragment $\chi\to\alpha$, the stable node branch $\mathcal{E}_{3}$ gives birth to the limit cycle (generally of radius $\sqrt{-\mu_2}$, with its stability transferred to the cycle and $\mathcal{E}_{3}$ becoming unstable (mild loss of stability). During evolution along the fragment, $\alpha\to\eta$, the cycle disappears \emph{by becoming invisible} and the state becomes unstable at crossing the $\nu$-line into the $\eta$ stratum (for more on this interesting `cycle blow up' phenomenon, cf. \cite{kuz}, p. 377, \cite{zol84}, \cite{wig}, pp. 783-6).

We know from the Poincar\'e-Andronov theorem that for planar systems (as the ones we study here) the only one-parameter bifurcations in generic families are the ones found here (i.e., merging of equilibria and Hopf). This completes our discussion of the convergence-vorticity problem.

\section{The bifurcation diagram for the OS-system}\label{bif3}
We now turn to the Oppenheimer-Snyder problem (hereafter `OS'-problem) and develop the versal dynamics and bifurcation diagrams. First, however, we shall bring the  OS-system to its normal form and find the versal unfolding. This is done in the next Subsection. The bifurcation diagram is constructed in later subsections.

A main conclusion from the analysis of the problem in the present Section is that the OS dynamics lies somewhat between the NRP-system and the vorticity-induced bifurcations considered in previous Sections, in that there are two main cases, one giving unstable solutions and resembling more to the NPR-problem, and another having stable limit cycle solutions - closer to the vorticity situation met in the previous Section.

\subsection{The topological normal form}
We write the OS equation (\ref{os1}) in a dynamical system form, by setting
\begin{equation}\label{os2}
\begin{split}
\dot{x}&=y,\\
\dot{y}&=-\frac{3}{4}y^2,
\end{split}
\end{equation}
with phase portrait as in Fig. \ref{oppie0}. One observes the standard OS behaviour in this phase portrait, namely that the comoving coordinate $x$ (being $\ln r$),  diverges to $-\infty$ (so corresponding to $r=0$), or to $+\infty$ ($r=\infty$), for negative or positive $y$ respectively.

 \begin{figure}
\centering
\includegraphics[width=0.5\textwidth]{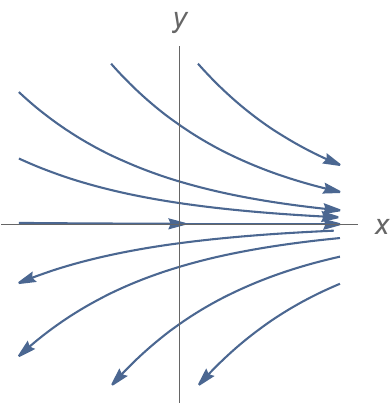}
\caption{The  phase portrait  for the original Oppenheimer-Snyder-system (\ref{os2}).}\label{oppie0}
\end{figure}

However, we note that linearized Jacobian of the system  (\ref{os2})  at the origin is,
\be\label{Jop}
J|_{(0,0)}=\left(
  \begin{array}{cc}
    0 & 1 \\
    0 & 0\\
  \end{array}
\right),
\ee
thus suggesting the nontrivial fact that the system (\ref{os2}) in the form $\dot{X}=J|_{(0,0)} X+F(X)$, where $X=(x,y)^\top , F(X)=(0,-(3/4)y^2)^\top$, may be subject to a Bogdanov-Takens bifurcation at the origin.

\subsubsection{Elimination of the second-order terms}
To study this, the first problem is how to write (\ref{os2}) in normal form.  This step  can be accomplished by first simplifying the second-order terms present in (\ref{os2}).

Since the normal form up to second-order terms is known (it is the Bogdanov-Takens normal form), it is a trivial matter to see that the quadratic Oppenheimer-Snyder term $(0,-(3/4)y^2)^\top$ in (\ref{os2}) is nonresonant\footnote{Note: As we discussed earlier, in standard terminology, the \emph{resonant} terms are the unremovable nonlinear terms which according to the normal form theorem belong to the complement of the set that contains all the terms that can be written as linear combinations of linearly independent elements of the space $L_J^{(2)}(H_2)$ - see below. This is generally true: if at any given order, terms present in the vector field at that order do not appear in the normal form, then they cannot be present and can be completely eliminated.}, and therefore it can be completely eliminated.

To see this explicitly, we first recall some standard terminology from normal form theory (cf.  refs. given in the bibliography). The scope of normal form theory, as we already discussed earlier, is to simplify the appearance of nonlinear terms at a given order. For example, to simplify terms of, say,  second-order, let us denote such terms by $F_2(X)$. Then the simplification may be  accomplished by performing a nonlinear transformation $X=Y+h_2(Y)$, to the original system variables to obtain the homological equation for $h_2(Y)$, namely,
\be
L_J^{(2)}(h_2(Y))=F_2(Y),\quad L_J^{(2)}(h_2(Y))=Dh_2(Y)JY-Jh_2(Y),
\ee
where $h_2(Y),F_2(Y)$ are vector-valued homogeneous polynomials of degree two. In $\mathbb{R}^2$, with standard basis $(1,0)^\top,(0,1)^\top$, the space $H_2$ of all vector-valued homogeneous polynomials of degree two is spanned by the vectors,
\be
H_2=\textrm{span}\left\{\left(
  \begin{array}{c}
    x^2 \\
    0 \\
  \end{array}
\right),\left(
          \begin{array}{c}
            xy \\
            0 \\
          \end{array}
        \right),\left(
          \begin{array}{c}
             y^2 \\
            0 \\
          \end{array}
        \right),\left(
          \begin{array}{c}
            0 \\
            x^2 \\
          \end{array}
        \right),\left(
          \begin{array}{c}
            0 \\
            xy \\
          \end{array}
        \right),\left(
          \begin{array}{c}
            0 \\
            y^2 \\
          \end{array}
        \right)
   \right\},
\ee
obtained by multiplying each of the standard basis vectors by all possible homogenous polynomials of degree two.

Now for $J$ given by  Eq. (\ref{Jop}), it is a simple calculation to show that,
\be
L_J^{(2)}(H_2)=\textrm{span}\left\{
\left(
  \begin{array}{c}
    -2xy \\
    0 \\
  \end{array}
\right),\left(
          \begin{array}{c}
           \pm y^2 \\
            0 \\
          \end{array}
        \right),\left(
          \begin{array}{c}
            x^2 \\
            -2xy \\
          \end{array}
        \right),\left(
          \begin{array}{c}
          xy \\
            -y^2 \\
          \end{array}
        \right),\left(
          \begin{array}{c}
          0 \\
            0 \\
          \end{array}
        \right)
   \right\}.
\ee
Consequently, the Oppenheimer-Snyder quadratic term from (\ref{os2}), can be written as the linear combination,
\be
\left(
\begin{array}{c}
          0 \\
          -\frac{3}{4}y^2 \\
\end{array}
\right)
=
\frac{3}{8}\left[
\left(
  \begin{array}{c}
    -2xy \\
    0 \\
  \end{array}
\right)+2\left(
          \begin{array}{c}
          xy \\
            -y^2 \\
          \end{array}
        \right)   \right],
\ee
of two of the basis elements, and therefore this term can be eliminated leaving no second-order terms in the normal form.

\subsubsection{The third-order terms}
The calculation of $L_J^{(3)}(H_3)$ is of course known in the literature but since no details are often provided, we include below some of the relevant results. A basis for $H_3$ is,

\begin{equation}
H_3=\textrm{span}\left\{\left(
  \begin{array}{c}
    x^3 \\
    0 \\
  \end{array}
\right),\left(
          \begin{array}{c}
            x^2y \\
            0 \\
          \end{array}
        \right),\left(
          \begin{array}{c}
             xy^2 \\
            0 \\
          \end{array}
        \right),\left(
          \begin{array}{c}
             y^3 \\
            0 \\
          \end{array}
        \right), \\
        \left(
          \begin{array}{c}
            0 \\
            x^3 \\
          \end{array}
        \right),\left(
          \begin{array}{c}
            0 \\
            x^2y \\
          \end{array}
        \right),\left(
          \begin{array}{c}
            0 \\
            xy^2 \\
          \end{array}
        \right),\left(
          \begin{array}{c}
            0 \\
            y^3 \\
          \end{array}
        \right)
   \right\}.
\end{equation}
The linear map on $H_3$ is,
\be
L_J^{(3)}=Jh_3(Y)-Dh_3(Y)JY,
\ee
and for each basis element, we calculate the action of the $L_J^{(3)}$-operator to be:
\be
L_J^{(3)}\left(
  \begin{array}{c}
    x^3 \\
    0 \\
  \end{array}
\right)=-3\left(
  \begin{array}{c}
    x^2y \\
    0 \\
  \end{array}
\right),
\ee
\be
L_J^{(3)}\left(
  \begin{array}{c}
    x^2y \\
    0 \\
  \end{array}
\right)=-2\left(
  \begin{array}{c}
    xy^2 \\
    0 \\
  \end{array}
\right),
\ee
\be
L_J^{(3)}\left(
  \begin{array}{c}
    xy^2 \\
    0 \\
  \end{array}
\right)=-\left(
  \begin{array}{c}
    y^3 \\
    0 \\
  \end{array}
\right),
\ee
\be
L_J^{(3)}\left(
  \begin{array}{c}
    y^3 \\
    0 \\
  \end{array}
\right)=\left(
  \begin{array}{c}
    0\\
    0 \\
  \end{array}
\right),
\ee
\be
L_J^{(3)}\left(
  \begin{array}{c}
    0 \\
    x^3 \\
  \end{array}
\right)=\left(
  \begin{array}{c}
    x^3\\
    -3x^2y \\
  \end{array}
\right),
\ee
\be
L_J^{(3)}\left(
  \begin{array}{c}
    0 \\
    x^2y \\
  \end{array}
\right)=\left(
  \begin{array}{c}
    x^2y\\
    -2 xy^2 \\
  \end{array}
\right),
\ee
\be
L_J^{(3)}\left(
  \begin{array}{c}
    0 \\
    xy^2 \\
  \end{array}
\right)=\left(
  \begin{array}{c}
    xy^2\\
    -y^3 \\
  \end{array}
\right),
\ee
\be
L_J^{(3)}\left(
  \begin{array}{c}
    0 \\
    y^3 \\
  \end{array}
\right)=\left(
  \begin{array}{c}
    y^3\\
    0 \\
  \end{array}
\right).
\ee
Therefore we find that
\begin{equation}
H_3=\textrm{span}\left\{\left(
  \begin{array}{c}
    x^3\\
    -3x^2y \\
  \end{array}
\right),\left(
  \begin{array}{c}
    x^2 y\\
    0 \\
  \end{array}
\right), \left(
  \begin{array}{c}
    xy^2 \\
    0 \\
  \end{array}
\right),\left(
  \begin{array}{c}
    y^3\\
    0 \\
  \end{array}
\right),\left(
  \begin{array}{c}
    0 \\
    y^3 \\
  \end{array}
\right),\left(
  \begin{array}{c}
    0 \\
    xy^2 \\
  \end{array}
\right)
  \right\},
\end{equation}
which implies that $\textrm{dim}L_J^{(3)}H_3=6$, and so $\textrm{dim} G_2=\textrm{dim} H_3-\textrm{dim}L_J^{(3)}H_3=8-6=2.$

In other words, to compute $G_2$ we need to find two linearly independent, orthogonal $6$-vectors to each column of the $(8\times 8)$-matrix representation of $L_J^{(3)}$. If $(a\,b\,c\,d\,e\,f\,g\,h)^\top$ are the components of any such (column-)vector, we find,
\be a=3f,\quad \textrm{any}\quad e,\quad \textrm{and}\quad b=c=d=g=h=0,
\ee
and so two such vectors are,  $(1\,0\,0\,0\,0\,-3\,0\,0)^\top$, $(0\,0\,0\,0\,1\,0\,0\,0)^\top$, leading to the two vectors,
\be
\left(
  \begin{array}{c}
    0 \\
    x^3 \\
  \end{array}
\right),\left(
  \begin{array}{c}
    x^3\\
    x^2y \\
  \end{array}
\right).
\ee
In fact, a simpler choice of a basis for $G_2$ is,
\be
\left(
  \begin{array}{c}
    0 \\
    x^3 \\
  \end{array}
\right),\left(
  \begin{array}{c}
    0\\
    x^2y \\
  \end{array}
\right),
\ee
where we have used,
\be
\left(
  \begin{array}{c}
    0\\
    x^2y \\
  \end{array}
\right)=\left(
  \begin{array}{c}
    x^3\\
    x^2y \\
  \end{array}
\right)-\left(
  \begin{array}{c}
    x^3\\
    0 \\
  \end{array}\right),
\ee
where the second vector is an element of $L_J^{(3)}(H_3)$. This is the choice that we shall use below.

\subsubsection{The normal form and the versal unfolding}
Consequently, the normal form of the OS-system near the origin to third-order terms is,
\begin{equation}
\begin{split}
\dot{x}&=y,\\
\dot{y}&=ax^3+bx^2y,
\end{split}
\end{equation}
with $a,b$ constants.

We can further rescale $x,y$ and arrive at the final normal form (cf. \cite{wig}, p. 437-8, \cite{gh83}, p. 365-6),
\begin{equation}\label{os3}
\begin{split}
\dot{x}&=y,\\
\dot{y}&=\pm x^3-x^2y.
\end{split}
\end{equation}
Therefore we have shown the following theorem.
\begin{theorem}
The normal form of the OS-system (\ref{os2}) contains two moduli coefficients and is given by Eq. (\ref{os3}).
\end{theorem}

The versal unfolding can now be found if we recall  that the  versal deformation  of the matrix,
$$\left(
\begin{array}{cc}
0 & 1\\
0 & 0 \\
\end{array}
\right),
$$
is, $$\left(                                                                                               \begin{array}{cc}
0 & 1\\
\mu_1 & \mu_2 \\
\end{array}
\right),
$$
and so using the normal form (\ref{os3}), we arrive at the versal unfolding,
\begin{equation}\label{vaC}
\begin{split}
\dot{x}&=y,\\
\dot{y}&=\mu_1 x+\mu_2 y\pm x^3-x^2y,
\end{split}
\end{equation}
with the two parameters $\mu_1, \mu_2$\footnote{The proof of the $\pm$-sign, thus reducing the number of cases without any loss of generality,  is justified by a rescaling analysis very similar to that perform for the quadratic case of the Bogdanov-Takens singularity, cf. \cite{wig}, p. 437-8, and it is omitted here.}.

The versality of this unfolding is proven  in many places, originally in  \cite{tak1,tak2}. We note the important result that when $\mu_1,\mu_2=0$ we get the degenerate system (\ref{os3}), \emph{not} the original OS-system (\ref{os2}), because in the normal form of the latter system all quadratic terms have disappeared.

We further note that the versal unfolding of the Oppenheimer-Snyder system given by (\ref{vaC}) contains the modular coefficient $s=\pm 1$ in front  of the $x^3$ term, and so there are two separate systems to be considered for the versal dynamics. This is done below.

\subsection{Versal dynamics for a positive modular coefficient}
\subsubsection{Local stability}
We start with the \emph{positive modular parameter} system,
\begin{equation}\label{vaC+}
\begin{split}
\dot{x}&=y,\\
\dot{y}&=\mu_1 x+\mu_2 y+ x^3-x^2y,
\end{split}
\end{equation}
which has the following fixed branches:
\begin{enumerate}
  \item $\mathcal{E}_{1,2}=(\pm\sqrt{-\mu_1},0)$. These are real provided,
  \be \mu_1<0.\ee
  \item The third fixed point is  at the origin,
\be\label{e3}
\mathcal{E}_{3}=(0,0),
\ee
We note that there are no fixed points for $\mu_1>0$.
\end{enumerate}
It is not difficult to deduce the unstable nature of the  branches $\mathcal{E}_{1,2}$ in this case.
The linearized Jacobian is,
\be
J_+=\left(
  \begin{array}{cc}
    0 & 1 \\
   \mu_1+3x^2-2xy & \mu_2 -x^2\\
  \end{array}
\right),
\ee
and so we find,
\be
J_+|_{\mathcal{E}_{1,2}}=\left(
  \begin{array}{cc}
    0 & 1 \\
   -2\mu_1 & \mu_1+\mu_2\\
  \end{array}
\right).
\ee
Therefore,
\be
\textrm{Tr}J_+|_{\mathcal{E}_{1,2}}=\mu_1+\mu_2,\quad \textrm{det}J_+|_{\mathcal{E}_{1,2}}=2\mu_1,
\ee
and so in the region under consideration where $\mu_1<0$, $\textrm{det}J_+|_{E_{1,2}}<0$. Therefore, from Lemma \ref{2ndlemma} we have that both $\mathcal{E}_{1,2}$ are saddles.

On the other hand,
\be
J_+|_{\mathcal{E}_{3}}=\left(
  \begin{array}{cc}
    0 & 1 \\
   \mu_1 & \mu_2\\
  \end{array}
\right),
\ee
and,
\be
\textrm{Tr}J_+|_{\mathcal{E}_{3}}=\mu_2,\quad \textrm{det}J_+|_{\mathcal{E}_{3}}=-\mu_1,
\ee
so that $\textrm{det}J_+|_{\mathcal{E}_{3}}>0,$ when $\mu_1<0$, and using Lemmas \ref{1stlemma}, \ref{2ndlemma},  we have the following result.
\begin{proposition}
In the half-space $\mu_1<0$, the origin $\mathcal{E}_{3}$ is:
\begin{enumerate}
\item a source for $\mu_2>0$,
\item a sink for $\mu_2<0$.
\end{enumerate}
\end{proposition}
When $\mu_1>0$, there are no fixed branches $\mathcal{E}_{1,2}$, while since $\textrm{det}J_+|_{\mathcal{E}_{3}}<0$, the origin is a saddle for any sign of $\mu_2$. We note that in the cases $\mu_1\gtrless 0$, there are no bifurcations. In addition, there are three equilibria in the half-space $\mu_1<0$,  but only $\mathcal{E}_{3}$ when $\mu_1>0$. These phase portraits will be shown after we also complete the study of their bifurcations (see end of the subsection).

The most interesting cases dynamically are when one of the parameters is zero, in which case we have local bifurcations.

\subsubsection{Local bifurcations, Case A: $\mu_1=0$}\label{+a}
In this case, we treat $\mu_2$ is an arbitrary nonzero constant and put the versal unfolding (\ref{vaC+}) in a suitable form for the center manifold theorem, which examines possible bifurcations \emph{near} $\mu_1=0$ (note that the eigenvalues in this case are $0,\mu_2$).

In order to achieve this , we pass to new `coordinates',
\be
\left(
  \begin{array}{c}
    x\\
    y \\
  \end{array}
\right)=\left(\begin{array}{cc}
1 & 1\\
0 & \mu_2 \\
\end{array}
\right)\left(
  \begin{array}{c}
    u\\
    v \\
  \end{array}
\right),\quad\left(
  \begin{array}{c}
    u\\
    v \\
  \end{array}
\right)=\left(\begin{array}{cc}
\mu_2 & -1\\
0 & 1\\
\end{array}
\right)\left(
  \begin{array}{c}
    x\\
    y \\
  \end{array}
\right),
\ee
so that,
\be \label{cm-ini}
\quad\left(
  \begin{array}{c}
    \dot{u}\\
    \dot{v} \\
  \end{array}
\right)=\frac{1}{\mu_2}\left(\begin{array}{cc}
\mu_2 & -1\\
0 & 1\\
\end{array}
\right)\left(
  \begin{array}{c}
    \dot{x}\\
    \dot{y} \\
  \end{array}
\right),
\ee
where the $\dot{x},\dot{y}$ terms are given by (\ref{vaC+}).

The centre manifold in this case is (at least of $O(2)$),
\be
v=v(u,\mu_2),
\ee
and so after some calculation using Eq. (\ref{cm-ini}), the reduced equation on the centre manifold is:
\be
\dot{u}=-\frac{\mu_1}{\mu_2} u-\frac{1}{\mu_2}u^3 +O(5),
\ee
implying a pitchfork bifurcation on $\mu_1=0$. This is supercritical for $\mu_2>0$ (origin repels all orbits), and subcritical for $\mu_2<0$ (origin attracts all orbits). This is like in Figs. \ref{super-bif}, \ref{sub-bif} respectively, with $\mu_1$ in the place of $-\epsilon$.

\subsubsection{Local bifurcations, Case B: $\mu_2=0$}\label{+b}
In the half-space $\mu_1<0$, the origin $\mathcal{E}_{3}$ is a centre when $\mu_2=0$, and so we expect a Hopf bifurcation on $\mu_2=0$. We shall be brief in proving these results.

It follows from  an application of the Hopf bifurcation theory, that since the eigenvalues of the linearized Jacobian satisfy,
\be
\frac{d}{d\mu_2} \biggr\rvert_{\mu_2 =0}\lambda_\pm=\frac{1}{2}\left(1+\frac{\mu_1}{\sqrt{\mu_1^2-8\mu_1}}\right)>0,\quad\mu_1<0,
\ee
we have that for $\mu_2<0$ the origin will be asymptotically stable, and for $\mu_2>0$ it will be unstable. In addition,
since the bifurcating periodic orbit at $\mu_2=0$ is stable, this is a supercritical Hopf bifurcation.

Also there are no periodic orbits when: 1) $\mu_1>0$, 2) on the third quadrant, and 3) by a Hamiltonian analysis, above the curve $\mu_2=-\mu_1/5$ (for this last result, cf. e.g.,  \cite{gh83}, p. 372-3).

This then provides the complete bifurcation diagram for the positive moduli case, cf. Fig. \ref{oppie1}.

 \begin{figure}
\centering
\includegraphics[width=\textwidth]{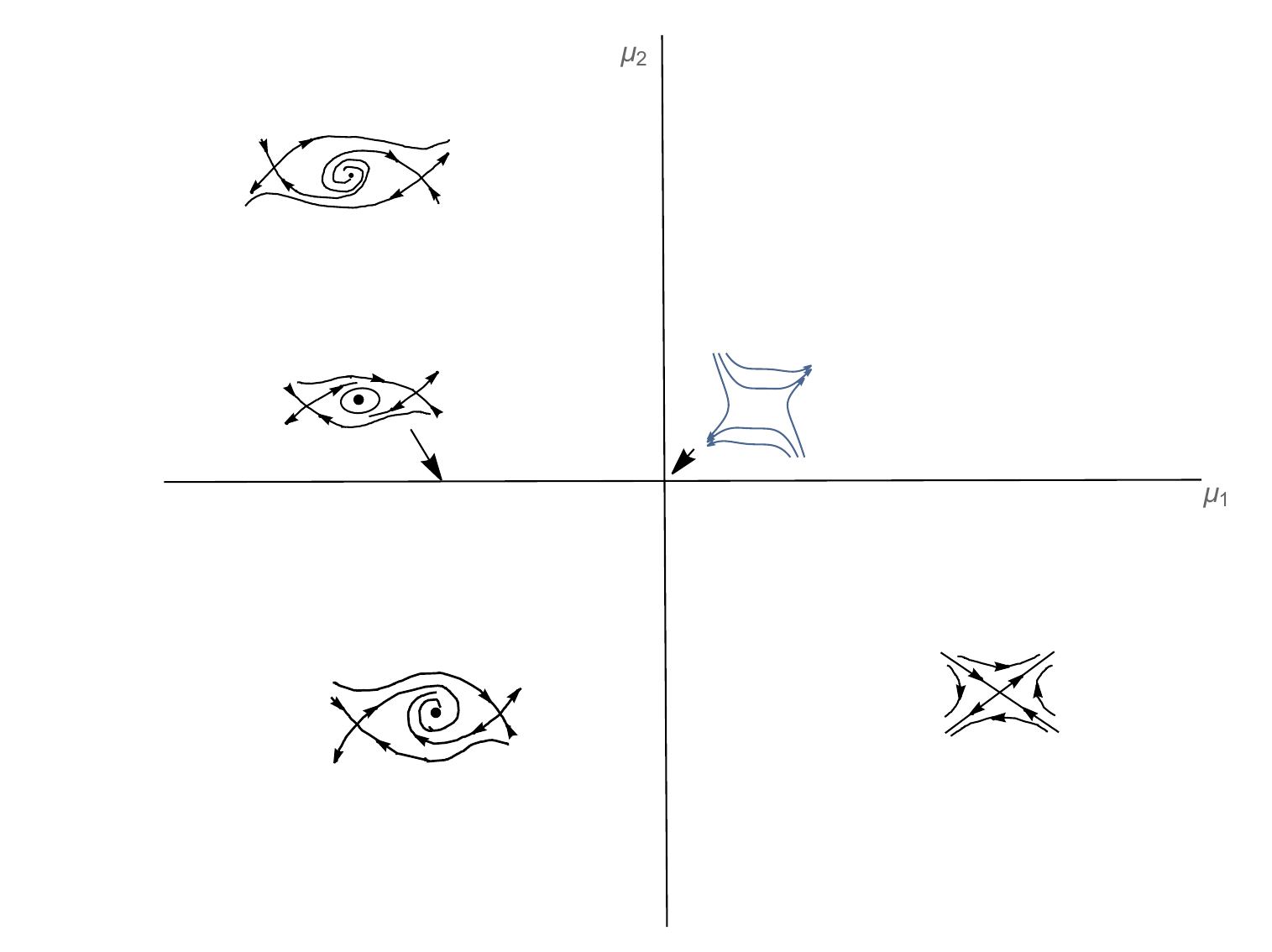}
\caption{The  bifurcation diagram for the Oppenheimer-Snyder-system, positive moduli case. }\label{oppie1}
\end{figure}

A basic characteristic of it is the existence of only unstable solutions in this case.

\subsection{Versal dynamics for a negative modular coefficient}
This case is somewhat more complicated that the positive moduli case (there are more global bifurcations), but the local bifurcations can be treated analogously, and so we refer the reader to the literature for the global problem.

\subsubsection{Local stability}
The \emph{negative modular parameter} system,
\begin{equation}\label{vaC-}
\begin{split}
\dot{x}&=y,\\
\dot{y}&=\mu_1 x+\mu_2 y- x^3-x^2y,
\end{split}
\end{equation}
has the following fixed branches:
\begin{enumerate}
  \item $\mathcal{E}_{1,2}=(\pm\sqrt{\mu_1},0)$. These are real provided,
  \be \mu_1>0.\ee
  \item The third fixed point is  at the origin,
\be\label{e3}
\mathcal{E}_{3}=(0,0),
\ee
and always exists (that is for all $\mu_1$). We note that there are no fixed points of the form $\mathcal{E}_{1,2}$ for $\mu_1<0$.
\end{enumerate}
The linearized Jacobian is
\be
J_-=\left(
  \begin{array}{cc}
    0 & 1 \\
   \mu_1-3x^2 & \mu_2 -x^2\\
  \end{array}
\right),
\ee
and so we find,
\be
J_-|_{\mathcal{E}_{1,2}}=\left(
  \begin{array}{cc}
    0 & 1 \\
   -2\mu_1 & \mu_2-\mu_1\\
  \end{array}
\right),
\ee
and,
\be\label{e1,2}
\textrm{Tr}J_-|_{\mathcal{E}_{1,2}}=\mu_2-\mu_1,\quad \textrm{det}J_-|_{\mathcal{E}_{1,2}}=2\mu_1.
\ee
Also, for $\mathcal{E}_{3}$, we have,
\be
J_-|_{\mathcal{E}_{3}}=\left(
  \begin{array}{cc}
    0 & 1 \\
   \mu_1 & \mu_2\\
  \end{array}
\right),
\ee
and
\be
\textrm{Tr}J_-|_{\mathcal{E}_{3}}=\mu_2,\quad \textrm{det}J_-|_{\mathcal{E}_{3}}=-\mu_1.
\ee
Therefore for $\mu_1>0$, $\mathcal{E}_{3}$ is a saddle. The fixed branches $\mathcal{E}_{1,2}$ have positive determinant from Eq. (\ref{e1,2}), and so they are centres on the line $\mu_2=\mu_1$ of the $(\mu_1,\mu_2)$ parameter plane, with eigenvalues given by
\be \label{ein-oppie}
\lambda_\pm =\pm2i\sqrt{2\mu_1}.
\ee
In addition, they are sources above and sinks below that line.

On the other half-space, i.e., when  $\mu_1<0$, there are no fixed branches $\mathcal{E}_{1,2}$, and since the $\textrm{det}J_-|_{\mathcal{E}_{3}}>0$,  on the lower half-space with $\mu_2<0$,  $\mathcal{E}_{3}$ is a sink, while on $\mu_2>0$,  $\mathcal{E}_{3}$ is a source.

This takes care of the local stability of the three equilibria $\mathcal{E}_{1,2,3}$ of the system (\ref{vaC-}).

\subsubsection{Local bifurcations}\label{-vea-b}
The two centres on the line  $\mu_2=\mu_1$ with the purely imaginary eigenvalues given by Eq. (\ref{ein-oppie})
lead to a Hopf bifurcation on this line when $\mu_1>0$. By the same method as in the previous subsection of positive moduli, we find that in the present case the bifurcating orbit is unstable, and so the Hopf bifurcation is subcritical.

As in the previous subsection, we also find a pitchfork bifurcation on $\mu_1=0$, and we shall not repeat the analysis here.

We note that there are no periodic orbits in the present case when $\mu_2<0$.

The bifurcation diagram for the negative moduli case is therefore given as in \ref{oppie2} (except for the global bifurcations). We note that by a Hamiltonian analysis (similar to that needed also in the positive moduli case), one proves the existence of closed orbits surrounding the three equilibria when $\mu_1>0$, and also the presence of a double saddle connection on the curve $\mu_2=4\mu_1/5$, just below the $\mu_2=\mu_1$ curve (not shown here, cf. \cite{gh83}, p. 373-4, Fig. 7.3.7).

 \begin{figure}
\centering
\includegraphics[width=\textwidth]{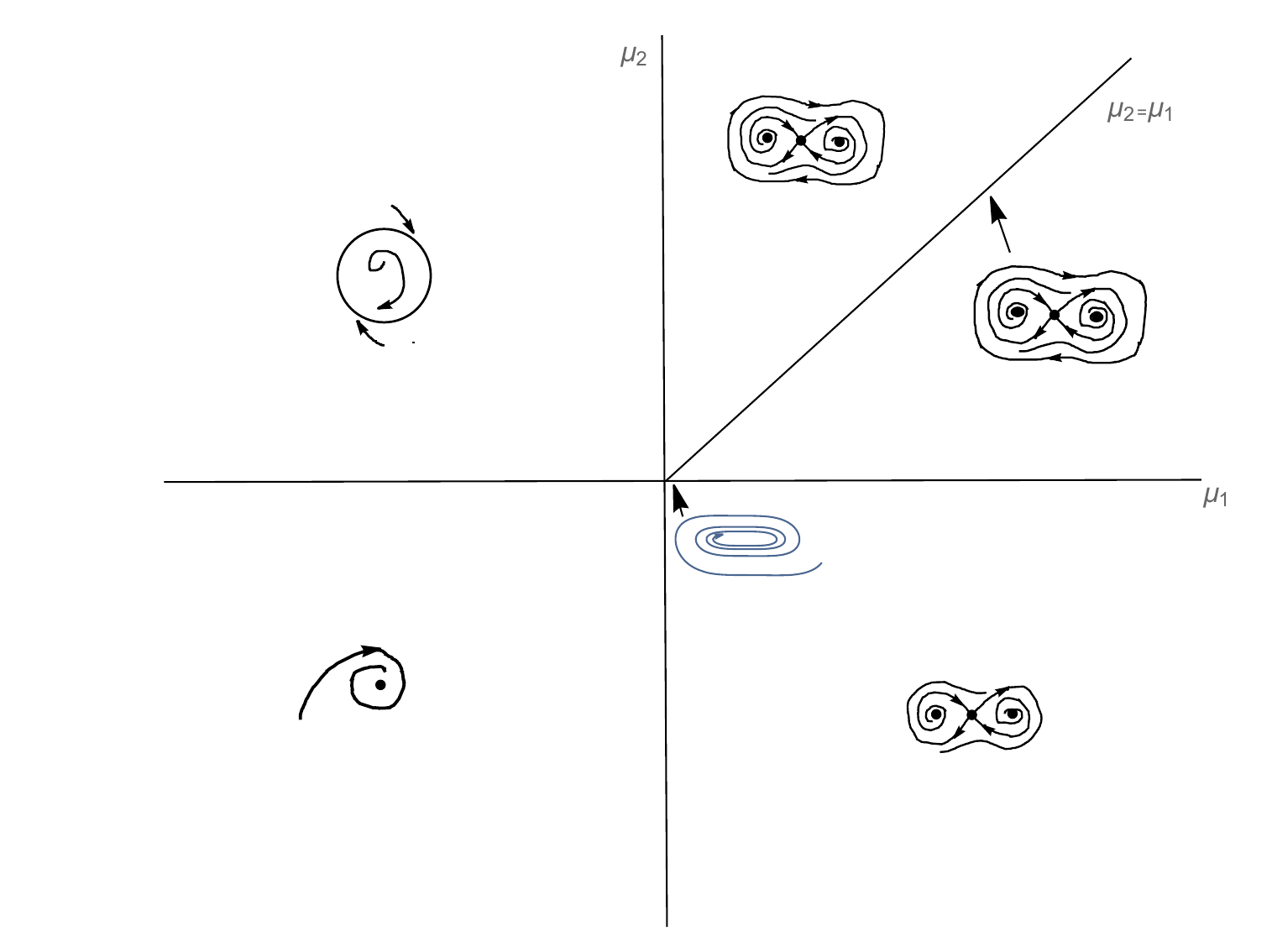}
\caption{The  bifurcation diagram for the Oppenheimer-Snyder-system, negative moduli case. }\label{oppie2}
\end{figure}

The form of the solutions indicates the presence of additional  global bifurcations which in the present case are more substantial than in those of the previous subsection. We refer the reader to the references for this problem (cf. \cite{gh83}, p. 376, Fig. 7.3.9, \cite{ar83}, p. 309, for a summary of the global bifurcations).

This concludes the analysis of the Oppenheimer-Snyder problem.

\section{Metamorphoses of spacetime singularities and black holes}\label{bif4}
In this Section, we present \emph{some} examples of the application of the previous  results to spacetime singularities and black holes. It is not intended to provide an exhaustive treatment of all the possible behaviours depicted in the bifurcation diagrams, only a small selection of examples from them.

\subsection{General comments}
The four bifurcation diagrams, namely, Fig. \ref{bifn} for the NPR-system, Fig. \ref{bifnOmega} for the CV-system, and Figs. \ref{oppie1}, \ref{oppie2} for the Oppenheimer-Snyder system, are the central results found in earlier Sections of this work. Using them, we can make direct contact with the possible metamorphoses of the spacetime near singularities and black holes. We note that no metamorphosis is ever possible without a bifurcation, because in this case any `change of form' is bound to happen exclusively in one phase portrait which can never become anything else.

Before we start, let us make clear that looking at any of the four bifurcation diagrams referred to above, we observe not one but several phase portraits each corresponding to a region (a `stratum') of the parameter diagram lying at the centre of the bifurcation diagram. One may wonder why we need all these diagrams and why we do not have just one phase portrait as that is common in introductory treatments of dynamical systems.
As we explained earlier, the purpose of each one of these phase portraits is to depict the dynamics corresponding to a particular region of the parameter space. This dynamics changes when the parameter point moves from one stratum to another (that is as `the parameter changes'), and with it phase diagrams also change.

However, this is but one way of thinking about the changes (or metamorphoses) of the phase diagrams, and in fact one that it may appear as  a somewhat deceptive one to some readers. Another, perhaps more indicative, approach is to think of all of them as being just one phase portrait which includes all others and as the parameter changes  it smoothly alters its form, each time it precisely becomes one of those different phase portraits we see during the metamorphoses of the dynamics. This approach makes the smooth motion of the phase point more transparent during the phase portraits' metamorphoses.

As the vector parameter $\mu$ changes its values (for instance, as the `parameter point' $\mu$ moves in a small circle around the origin of the parameter plane) and the phase point moves smoothly from one phase portrait to another, the causality relations of spacetime become dependent of $\mu$. Therefore causal set relations  may change as the parameters change their values and pass through from different regions of the parameter plane $(\mu_1,\mu_2)$ separated by the bifurcation boundaries,  that is as the system enters into  or exits from  the various strata.

To wit, the metamorphoses of the phase portraits imply a constant topological transformation of the states of the system as the parameter $\mu$ changes and the system finds itself in different regions of the parameter space. One thus observes how the different phase portraits (seen in any of the four bifurcation diagrams) constantly and smoothly  transform to the next one and take the system represented by a point in any of the phase portraits to move in another such portrait corresponding to another stratum in the parameter space.

We shall first discuss the NPR- and CV-systems, and then in the last Subsection we shall return to the discussion of the Oppenheimer-Snyder bifurcations. We note that in both cases, that of the NPR- and CV-systems, and the two subcases of the Oppenheimer-Snyder collapse, there are in general focusing and \emph{defocusing} solutions, with the focusing ones producing the unstable diagrams \ref{bifn} and \ref{oppie1}, and the defocusing solutions seen in the other pair, namely Figs. \ref{bifnOmega},  \ref{oppie2}. We shall not discuss each one of the phase portraits in the bifurcation diagrams, but only restrict ourselves to discussing some of them as well as of their metamophoses.

\subsection{The NPR and CV systems}

The versal unfoldings (\ref{vaA}), (\ref{vaB}) can be written in the unified form,
\begin{equation}\label{vaD}
\begin{split}
\dot{\rho}&=\mu_1+\rho^2+sz^2,\\
\dot{z}&=\mu_2 z+n\rho z.
\end{split}
\end{equation}
for $z=\sigma$, or, $z=\omega$, respectively,  $s=\pm 1$ is a moduli coefficient, and $n=2,3$.  Eq. (\ref{vaD}) is a standard form of a planar system with a $\mathbb{Z}_2$-symmetry. Since for both values of the modular coefficient, (\ref{vaD}) represents a different subsystem of the Sachs optical equations, it is natural to make the following correspondences of the parameters $\mu_1,\mu_2$ of  the NPR- and the CV-systems:
\be\label{par-corr}
\textrm{NPR-system}:\quad\mu_1\to \mathcal{R},\quad \mu_2\to \mathcal{W},
\ee
with $\mathcal{R},\mathcal{W}$ as in Subsection \ref{r-w}. For the CV-system, since when $\sigma=0$ there is no $W$ term, it is natural in this case to think of $\mu_2$ as a rotational parameter, so that we set,
\be\label{par-corrVO}
\textrm{CV-system}:\quad\mu_1\to \mathcal{R},\quad \mu_2\to \textrm{rotation parameter}.
\ee

\subsubsection{Convergence-shear transfigurations}
We shall start by examining the continuous metamorphoses (or transfigurations) of phase portraits for the NPR-system following the bifurcation diagram in Fig. \ref{bifn}. We have shown that the following bifurcations occur:
\begin{enumerate}
\item \textbf{A pair of saddle-node bifurcations}\emph{ dominated by the convergence $\rho$} on the centre manifold-reduced dynamics, taking the system along the following fragments (cf.  the bifurcation diagram in Fig. \ref{bifn}, and Theorem \ref{s-nNPR}):
    \be
    \chi\to\varepsilon_+\to o, \quad \textrm{or, in opposite direction,}\quad  o\to\varepsilon_+\to\chi,
    \ee
    and similarly for the negative $\varepsilon_-$.
\item \textbf{A pair of pitchfork bifurcations} \emph{dominated by the shear $\sigma$} on the centre manifold-reduced dynamics, taking the system along the following fragments (cf.  the bifurcation diagram in Fig. \ref{bifn}, and Theorem \ref{pitch-thm}):
    \be
    o\to\tau_+\to\beta,\quad \textrm{supercritical, in direction above to below},
     \ee
and,
    \be
   \pi\to\tau_-\to\beta,\quad \textrm{subcritical, in direction below to above}.
    \ee
\end{enumerate}
These bifurcations  describe the possible metamorphoses of the phase portraits of the NPR-problem, and dictate how and where the system will change into something else.
The stratum $\mu_1>0$ in this Fig. \ref{bifn} identifies with positive local energy density as in the energy conditions,  and so it describes a region where gravity is always attractive, while the region $\mu_1<0$ obviously identifies with repulsive gravitational effects.

Therefore, the $\rho$-dominated, saddle-node bifurcations of the NPR-system take the system from a region of attractive gravity to one where gravity is always repulsive and back, whereas both of  the $\sigma$-dominated, pitchfork bifurcations occur in the repulsive gravity region of the parameter space.

In the stratum $\mu_1>0$, where the energy conditions hold, we have only  phase portrait shown in the $\chi$-region, the orbits of which describe the focusing state given by the relations in (\ref{fs}). Hence, this is the region in the parameter plane that corresponds to the Hawking-Penrose singularity theorems, based on the standard focusing effect. According to our results, the system in this region of the parameter space will bifurcate through a saddle-node to the strata $o$ or $\pi$ when crossing the positive or negative $\varepsilon$ axis respectively.

\subsubsection{Spacetime singularities and their metamorphoses}
As an example of a metamorphosis, we now consider the fragment $o\to\varepsilon_+\to\chi$. That is, suppose that  initially the system is in the stratum $o$ of the parameter space ($\mu\in o$) where repulsive gravity rules, so that the possible motions of the system are given by the corresponding phase portrait in that stratum, as in Fig. \ref{bifn}.

Now suppose that the system moves according to the fragment $o\to\varepsilon_+\to\chi$, first to the line $\varepsilon_+$, and then it enters the stratum $\chi$.
The dynamics in this case is described by the saddle-node bifurcation and subsequent metamorphoses as in the Theorem \ref{s-nNPR}. At the $\varepsilon_+$-crossing, $\mu_1=0$, the two equilibria are annihilated and the system is described by the evolution law (\ref{rho-eqn}), that is $\dot{\rho}\geq\rho^2$ in standard focusing effect language,  on the centre manifold as it enters the stratum $\chi$.

In terms of phase space dynamics, consider a phase point moving on a phase curve of the $o$-phase portrait in Fig. \ref{bifn}. We may choose that the phase point in question is on some phase curve corresponding to an initial condition with $\rho_0>0$.
At the $\varepsilon_+$-crossing, and as the phase portrait of the region $o$ transfigures to that of $\varepsilon_+$,  the phase point \emph{continues smoothly} its evolution along that  orbit which passes through the same initial condition in the new $\varepsilon_+$-phase portrait. The same happens to the phase point when it finds itself in the phase portrait that corresponds to the stratum-$\chi$, namely, it continues smoothly its evolution along the corresponding phase orbit of the $\chi$-phase portrait leading to the focusing state eventually.

Therefore during this process, as a focusing state is approached for all phase orbits with suitable initial conditions, a singularity is formed from a previous state (i.e., at $o$) where no such situation existed.
This result gives a first indication of how a spacetime singularity, as predicted by the singularity theorems, may arise in the $\chi$ stratum during the evolution of the system starting from a state in the $o$ stratum where no focusing existed. The main reason for the singularity formation in this case is related to the system's ability to bifurcate in a saddle-node, $\rho$-dominated bifurcation.

On the other hand, and as we have shown in Section \ref{bif}, the centre manifold analysis leading to singularity formation is reversible and applies to the evolution along the opposite fragment, namely,  $\chi\to\varepsilon_+\to o$. In this case, the system starts in the singularity-forming region $\chi$, and then crosses the $\varepsilon_+$ axis to find itself in the $o$-stratum where two new fixed branches have been formed. The transfigurations of the $\chi$-phase portrait to the $\varepsilon_+$ one, to the $o$-region one, are described as before, with the main difference being that at crossing the two new solutions are \emph{created} (instead of colliding) according to the saddle-node prescription, and are then separated in the $o$-stratum.

Another aspect of the  `singularity-forming region'  $\chi$-stratum is related to the so-called \emph{ghost effect}\footnote{This is also called the `bottleneck effect', cf. \cite{stro}, pp. 99, 242, and Refs. therein.}: a slow passage to the eventual singularity at the $\chi$-region of Fig. \ref{bifn}, after the two equilibria of the system present in the $o$-stratum collide and annihilate on $\varepsilon_+$, and as the system enters the singularity $\chi$-region. This is calculated by  integrating the centre manifold evolution law (\ref{rho-eqn}) to obtain,
\be
T_{\textrm{bottleneck}}=\int_{-\infty}^\infty \frac{d\rho}{\mu_1 +\rho^2}=\frac{\pi}{\sqrt{\mu_1}}.
\ee
Therefore the system delays to reach the singularity while in the $\chi$ region, with the delay time scaling as $\mu_1^{-1/2}$. This describes
the parameter-dependence of the time to the singularity as a square-root scaling law. This is a new feature introduced here because of the bifurcating behaviour (saddle-node) and describes possible transitions of the system. It is a parameter-dependent effect, absent in the standard focusing effect (the Raychaudhuri inequality,  $\dot{\rho}\geq \rho^2$ implies that only a dependence on the initial condition can appear).

\subsubsection{Black hole metamorphoses}
As another example of the previous developments, suppose that an event horizon has formed in the $\chi$-region, where according to  standard theorems, $\rho$ will increase and become infinite on null geodesics within a finite affine distance in the future depending on the initial condition $\rho_0$ (as we have also shown above). In other words, in the $\chi$-region, the generators of the boundary of a future set $S$, i.e., $\partial I^+(S)$, will have future end points where they will intersect neighboring generators.

However, here the situation is slightly different. Because of the possibility of transfiguration through the saddle-node bifurcation described above,  all causal-structure pointsets will be equipped with an extra dependence on the parameter $\mu$. For example, in the present case, we can write $\partial_{\mu_1} I^+(S)$ to indicate the parameter dependence. Upon crossing from the $\chi$ to the $o$ region through $\varepsilon_+$, the boundary $\partial I^+(S)$, will gradually transfigure as follows,
\be
\partial_{\mu_1>0} I^+(S)\to\partial_{\mu_1=0} I^+(S)\to\partial_{\mu_1<0} I^+(S),
\ee
and as a result, the generators intersecting for the sets $\partial_{\mu_1>0} I^+(S)$ will not intersect anymore on $\partial_{\mu_1<0} I^+(S)$. In other words, while on $\chi$ once the generators of $\partial_{\mu_1>0} I^+(S)$ started converging they were destined to intersect and have their future endpoints within a finite distance, now because of the transfiguration to the pointset $\partial_{\mu_1<0} I^+(S)$, the same generators now describing $\partial_{\mu_1<0} I^+(S)$ will not intersect each other.

This is  because the phase portrait of the orbits corresponding to $\chi$ has gradually changed to that on $\varepsilon_+$ and finally on $o$, and here there is generally no possibility for a focusing state (that is, at crossing or on $o$), as the corresponding phase diagrams clearly show. Hence, unlike in the proof of Penrose theorem in Section \ref{sings-app} (cf. the `crucial step' mentioned there), the compactness of a set like $\partial_{\mu_1<0} I^+(\mathcal{T})$ does not follow in the present case because no point on this set can be made to belong to the compact set $A$ as constructed in the proof of that result anymore. Further, all such transfigurations will be smooth.

This describes how a blackhole-forming spacetime region corresponding to  $\chi$ is transfigured into one where no such regions exist, while  phase points continue with their orbits smoothly in the phase diagrams of the new strata. Such regions can be now further transfigured using the pitchfork bifurcations which take the system from the regions $o$ and $\pi$ and crossing the $\tau$ parabola into other forms.

\subsubsection{Vorticity-induced transfigurations}
Let us now consider the bifurcation diagram in Fig. \ref{bifnOmega} of the CV-system. Here we have vorticity-induced bifurcations as follows:
\begin{enumerate}
\item \textbf{A pair of saddle-node bifurcations}\emph{ dominated by the convergence $\rho$} on the centre manifold-reduced dynamics, taking the system along the following fragments (cf.  the bifurcation diagram in Fig. \ref{bifnOmega}, and Section \ref{vort-detail}):
    \be
    \chi\to\varepsilon_+\to o, \quad \textrm{or, in opposite direction,}\quad  o\to\varepsilon_+\to\chi,
    \ee
       and analogously for the negative $\varepsilon_-$.
\item \textbf{A pair of pitchfork bifurcations} \emph{dominated by the vorticity $\omega$} on the centre manifold-reduced dynamics, taking the system along the following fragments (cf.  the bifurcation diagram in Fig. \ref{bifnOmega}, and Section \ref{vort-detail}):
    \be
    o\to\tau_+\to\beta,\quad \textrm{supercritical to a node, in direction above to below},
     \ee
and,
    \be
   \pi\to\tau_-\to\beta,\quad \textrm{subcritical to a node, in direction below to above}.
    \ee
\item \textbf{A Hopf  bifurcation} \emph{dominated by the vorticity $\omega$} taking the system along the following fragments (cf.  the bifurcation diagram in Fig. \ref{bifnOmega}, and Section \ref{vort-detail}):
    \be
    \chi\to E\rho\to\alpha\to\nu\to\eta,
    \ee
    or, in the opposite direction.
\end{enumerate}
While the first two are as in the NPR-problem (but with the important difference that \emph{a node} instead of a saddle is involved), the third bifurcation is very important and is wholly due to the effects of convergence-vorticity combined. More precisely, the cycle created in the degenerate Hopf bifurcation on the $E\rho$-axis, stabilizes on the $\alpha$-stratum and makes the  Hopf bifurcation non-degenerate.

This is perhaps the most distinct feature of the vorticity-induced  bifurcations, namely, the creation and annihilation of a \emph{unique stable limit cycle in the stratum $\alpha$} of the parameter space in the attractive-gravity region where the energy condition holds. Its appearance, dominance, and eventual disappearance in the $\eta$ stratum means that \emph{stable configurations with finite convergence and vorticity which attract  nearby orbits} are dominant features in this problem.

We know that the existence of such a closed orbit is an extremely important phenomenon in general: although periodic orbits exist in linear systems, limit cycles only appear in nonlinear studies because they are \emph{isolated} (unlike in the linear case). They describe the ability of a system to oscillate in a \emph{self-sustained} manner (that is without any external forcing). In the present case, the unique stable limit cycle appears in the $\alpha$-stratum and attracts all neighboring orbits. It describes the finite behaviour of the $(\rho,\omega)$ solutions as self-sustained oscillations. We conjecture that the flow on the $\alpha$-stratum, describing configurations with finite $\rho,\omega$ for instance, expanding universes with rotation, becomes quasi-periodic on invariant 2-tori for a dense set of parameter values having positive measure.

Upon parameter variation, the cycle disappears when crossing the $\nu$-line and into the $\eta$-stratum, to become an unstable focus, and then further bifurcate as shown in Fig. \ref{bifnOmega}. It is an intermediate feature that appears in the bifurcation diagram \ref{bifnOmega} in the fragment from the $\chi$ to $\eta$ strata and back.

The analysis of the vorticity-induced bifurcations points to features not present in the NPR bifurcation diagram discussed above.

In both cases, NPR and CV, there are as we have shown  various kinds of bifurcations in the repulsive-gravity regions, corresponding to the left half-spaces in their bifurcating diagrams. These effects  will be discussed in more detail elsewhere.

\subsection{Transfigurations in the Oppenheimer-Snyder example}
We conclude the discussion of the possible metamorphoses of singularities and black holes by giving some general remarks about the perturbations of the Oppenheimer-Snyder example. Unstable as well as isolated closed orbits also appear in the Oppenheimer-Snyder example as in the bifurcation diagrams \ref{oppie1}, \ref{oppie2}.

To interpret our results of this problem, since the versal unfolding (\ref{vaC}) and Figs. \ref{oppie1}, \ref{oppie2} describe all stable perturbations possible for the OS equation (\ref{os1}), we may generally set:
\be\label{par-corr2}
\mu_1\to\textrm{deviations from spherical symmetry},\quad \mu_2\to\textrm{rotation},
\ee
where, when $\mu_1=\mu_2=0$ in (\ref{vaC}) we are back to the original Oppenheimer-Snyder equation (\ref{os1}).

Since the OS equation has two moduli coefficients, we list the possible bifurcations as follows.

\subsubsection{List of bifurcations, positive modular coefficient}
\begin{enumerate}
\item \textbf{A pair of pitchfork bifurcations} \emph{dominated by the $u$-variable} on the centre manifold-reduced dynamics, taking the system along the following fragments (cf.  the bifurcation diagram in Fig. \ref{oppie1}, and Sections \ref{+a}, \ref{+b}):
    \begin{itemize}
    \item 1st to 2nd quadrant, supercritical  in direction right to left ($\mu_2>0$)
    \item 4th to 3rd quadrant, subcritical in direction right to left ($\mu_2<0$).
    \end{itemize}
\item \textbf{A supercritical Hopf  bifurcation} \emph{dominated by the $v$-variable  (that is the $y$)} taking the system in the $\mu_1<0$ half-space from bottom to top. We note that the bifurcating orbit is \emph{stable} on the horizontal axis in the Fig. \ref{oppie1}.
\end{enumerate}
\subsubsection{List of bifurcations, negative modular coefficient}
\begin{enumerate}
\item \textbf{A pair of pitchfork bifurcations} \emph{dominated by the $u$-variable} on the centre manifold-reduced dynamics, taking the system along the following fragments (cf.  the bifurcation diagram in Fig. \ref{oppie2}, and Section \ref{-vea-b}):
    \begin{itemize}
    \item 1st to 2nd quadrant, supercritical  in direction right to left ($\mu_2>0$)
    \item 4th to 3rd quadrant, subcritical in direction right to left ($\mu_2<0$).
    \end{itemize}
\item \textbf{A subcritical Hopf  bifurcation} on the $\mu_1=\mu_2$ line taking the system in the $\mu_1>0$ half-space. We note that the bifurcating orbit is \emph{stable} on the horizontal axis in the Fig. \ref{oppie2}.
\item \textbf{Global bifurcations} (not shown here) because of the presence of the saddle connections.
\end{enumerate}

\subsubsection{Some remarks}
A basic aspect of the bifurcation diagrams is the existence of collapsing solutions for the perturbations. These are described by escaping orbits in the bifurcation diagram of the positive moduli case in Fig. \ref{oppie1}. Since these can be found everywhere in that diagram, we conclude that gravitational collapse is possible for all stable perturbations of the   OS equation as these are described by the versal unfolding (\ref{vaC}). (This is of course also a conclusion of the NPR-system bifurcation analysis of Section 5 as well!)

In this way, the bifurcation diagram in Fig. \ref{oppie1} contains only unstable solutions, and so resembles more to the general situation of the NPR system of Fig. \ref{bifn}. On the other hand, the existence of closed orbits in the Fig. \ref{oppie2} points to stable solutions and for this reason this is closer to the vorticity case in Fig. \ref{bifnOmega}. We see that both properties may be deduced in general terms already for the versal unfolding of the OS equation.

The existence of a focusing state appears clearly everywhere in the positive moduli coefficient case, whereas the negative moduli case is more amenable to defocusing (closed orbit) solutions. In fact, as we have already discussed, in the latter case there is a variety  of global bifurcations (cf. the references).

We shall provide a more detailed description of these bifurcations elsewhere.

\section{Discussion}
In this paper we have provided an analysis of bifurcation theory effects for the problem of the formation of spacetime singularities as these appear in gravitational collapse and cosmological situations. We constructed the complete bifurcation diagrams of the evolution laws associated with these problems, namely, the `Raychaudhuri-related' convergence-shear and convergence-vorticity equations, as well as the differential equation that modelled the Oppenheimer-Snyder problem of `continued gravitational contraction', the first mathematical model of a black hole.

An analysis of these diagrams leads to interesting new features of the overall dynamics of these laws, and we have discussed some of these features in detail in earlier Sections of this paper.

A starting point of our analysis is how to solve the problem of  controlling  the `feedback loop' associated with the fundamental equations of this  problem. In the standard approach, one first employs an energy condition and the positivity of the shear term to directly obtain the solution for the convergence $\rho$ corresponding to some initial condition $\rho_0$, by  integrating the inequality resulting from the Raychaudhuri equation. Then one uses an equation of the form $\dot{x}=a\rho x$, $a$ const., (with $x$ usually being the volume, or area, or the shear of the congruence), to deduce, or `control' the decay behaviour of $x$. In this case, using the behaviour of the `growth factor' $\rho$ found previously, the solution $x$ influences the `forcing term' $\rho (t) x$, which in turn influences \emph{linearly} the solution $x$, thus creating a \emph{linear} feedback loop. This analysis leads directly to the focusing effect and the consequent predictions of spacetime singularities.

A very basic issue associated with the `feedback loop problem' is how to separate the linear from the nonlinear aspects of the  feedback, that is how to distinguish the focusing (or adversarial case) from the possible (or suspected) \emph{defocusing} (or, average case) behaviours (cf. \cite{tao} for more discussion on this fundamental problem). In this paper, guided by the pioneering studies that led to the singularity theorems and related results, we introduced the use of bifurcation theory as an efficient means to separate these two kinds of behaviour. We showed that the nonlinear feedback loop is naturally described by the normal forms of the evolution laws, and this approach provides a novel way to study the problem of spacetime singularities. As an example of this behaviour, we were able to show that a stable perturbation of the OS equation to non-spherical or rotational regimes must include focusing state solutions.

A second, and perhaps even more basic, aspect  of our approach is the problem of structural (in-)stability and genericity of the basic laws that govern the dynamics of feedback loop. It turns out that the basic equations that govern phenomena associated with spacetime singularities in gravitational collapse and cosmology, such as the three systems studied in this paper but also many others, are structurally unstable from the viewpoint of dynamical systems theory. This means that the behaviour of the solutions of \emph{nearby} systems obtained as perturbations of the original ones (a notion that can be made precise) may have very different behaviour than that of the original law. This raises the question of what is the precise meaning of proving global stability of an exact solution of the original system with respect to some perturbations, if the system itself is structurally unstable. In other words, for systems with some kind of degeneracy,  the stability of both the solutions \emph{and of the systems themselves} must be studied in order to obtain a reliable picture, cf. e.g., \cite{ar83}. For the three dynamical systems studied in this work, we have performed a complete analysis of this problem and found the versal unfoldings of each one of them. These extended systems are parametric families which contain all possible stable perturbations of the original equations but, unlike the latter, the versal families are themselves structurally stable.

The dynamical analysis of the versal unfoldings reveals both the  focusing and new \emph{defocusing} aspects of the three main systems. For the NPR equations, the vorticity-induced perturbations lead to defocusing solutions, as does the versal unfolding with negative moduli in the Oppenheimer-Snyder problem. A characteristic feature of the defocusing solutions is the `nucleation' of a unique stable limit cycle by transfer of stability (vorticity case) implying self-sustained oscillations of the perturbations, and  also various closed orbits as well as global bifurcations.

Another aspect of the solutions is that although in the attractive-gravity regions (where $\mu_1>0$) the focusing and defocusing solutions generally correspond to positive curvature solutions, in the repulsive-gravity region (where $\mu_1<0$) solutions correspond to metrics with hyperbolic regions where the curvature is negative. This is evident for example, in the regions $o$ in the NPR diagram where de Sitter solutions form due to the saddle node bifurcation on the $\varepsilon_+$ axis, and then these further bifurcate to give the saddle inside the parabola (recall our interpretation of the parameters $\mu_1,\mu_2$, an the constraint these satisfy on the parabola). In the saddle-type solutions there are focusing (unstable) orbits possibly leading to singularities just like in the standard singularity theorems in $\chi$, but this time in the repulsive region.

It is  interesting to observe the possibility of continuous transfiguration of any of the phase portraits in any of the regions in the four bifurcation diagrams, which is perhaps the most distinctive phenomenon of all solutions employed here. This aspect in turn is probably due to the tendency of the system to maintain its implied global structural stability of the versal families associated with the singularities present in the solutions.

A fuller analysis of these effects will be given elsewhere.

\addcontentsline{toc}{section}{Acknowledgments}
\section*{Acknowledgments}
The author is  especially grateful to Gary Gibbons for many useful  discussions which have had a positive effect on the final  manuscript. 
A Visiting Fellowship to Clare Hall, University of Cambridge, is gratefully  acknowledged. The author further thanks Clare Hall for  its warm hospitality and partial financial support.
This research  was funded by RUDN University,  scientific project number FSSF-2023-0003.


\addcontentsline{toc}{section}{References}
\begin{thebibliography}{99}
\bibitem{pen65}R. Penrose, Gravitational Collapse and Space-Time Singularities, Phys. Rev. Lett. 14 (1965) 57-9.
\bibitem{ha67}S. W. Hawking, The occurence of singularities in cosmology III, Proc. Roy. Soc. Lond. A300 (1967) 187-201.
\bibitem{pe68}R. Penrose, Structure of space-time, In Battelle Rencontres, 1967 Lectures in Mathematics and Physics,  C. M. De Witt and J. A. Wheeler (Benjamin, 1968), pp. 121-235.
\bibitem{hp70}S. W. Hawking and R. Penrose, The Singularities of Gravitational Collapse and Cosmology, Proc.  Roy. Soc. A 314 (1970) 529-548
\bibitem{ha71} S. W. Hawking, Gravitational Radiation from Colliding Black Holes, Phys. Rev. Lett. 26 (1971) 1344
\bibitem{pe72}R. Penrose, Techniques of Differential Topology in Relativity  (SIAM Philadelphia,  1972)
\bibitem{ha73}S. W. Hawking, The Event Horizon, In: Black Holes, B. S.  De Witt and C. M. DeWitt  (Gordon and Breach, 1973)
\bibitem{he}S. W. Hawking and G. F. R. Ellis, The Large Scale Structure of Space-Time (CUP, 1973)
\bibitem{mtw}C. W. Misner, K. P. Thorne, and J. A. Wheeler, Gravitation (Freeman, 1973)
\bibitem{oneill}B. O'Neill, Semi-Riemannian Geometry with Applications to Relativity (Academic Press, 1983)
\bibitem{wald}R. M. Wald, General Relativity (University of Chicago Press, 1984)
\bibitem{stew}J. Stewart, Advanced General Relativity (CUP, 1991)
\bibitem{strau}N. Straumann, General Relativity with Applications to Astrophysics (Springer, 2010)
\bibitem{on}B. O'Neill, The Geometry of Kerr Black Holes (Dover, 2014)
\bibitem{os}J. R. Oppenheimer and H. Snyder, On Continued Gravitational Contraction, Phys. Rev. 56 (1939) 455
\bibitem{ll}L. D. Landau and E. M. Lifshitz, The Classical Theory of Fields (4th Rev. Ed.) (Pergamon Press, 1975)
\bibitem{ray1}A. Raychaudhuri, Phys. Rev. 98 (1955) 1123
\bibitem{ko}A. Komar, Phys. Rev. 104 (1956) 544
\bibitem{ray2}A. Raychaudhuri, Phys. Rev. 106 (1957) 172
\bibitem{tao}T. Tao, Nonlinear Dispersive Equations: Local and Global Analysis (Americal Mathematical Society, 2006)
\bibitem{ar83}V. I. Arnold, Geometrical Methods in the Theory of Ordinary Differential Equations (Springer, 1983)
\bibitem{ar94}V. I. Arnold, Dynamical Systems V: Bifurcation Theory and Catastrophe Theory (Springer, 1994)
\bibitem{ar72}V. I. Arnold, Lectures on bifurcations in versal families,  Russ. Math. Surv. 27 (1972) 54
\bibitem{ar86}V. I. Arnold, Carastrophe Theory (Springer, 1986)
\bibitem{gh83}J. Guckenheimer and P. Holmes, Nonlinear oscillations, dynamical systems, and bifurcations of vector fields (Springer, 1983)
\bibitem{golu1}M. Golubitsky, D. G. Schaeffer, Stable Mappings and their Singularities, (Springer, 1979)
\bibitem{golu2}M. Golubitsky, D. G. Schaeffer, Singularities and Groups in Bifurcation Theory, Volume I (Springer, 1984)
\bibitem{golu3}M. Golubitsky, I. Stewart, D. G. Schaeffer, Singularities and Groups in Bifurcation Theory, Volume II (Springer, 1988)

\bibitem{wig}S. Wiggins, Introduction to applied nonlinear dynamical systems and chaos, 2nd. Ed. (Springer, 2003)
\bibitem{kuz}Yu. A. Kuznetsov, Elements of Applied Bifurcation Theory, Fourth Ed. (Springer, AMS 112, 2023)
\bibitem{thom}R. Thom, Structural Stability and Morphogenesis (CRC Press, 2018)
\bibitem{cot23}S. Cotsakis, Dispersive Friedmann universes and  synchronization, Gen. Rel. Grav. 55 (2023) 61;	arXiv:2208.07892
\bibitem{zol84}K. Zholondek, On the versality of a family of symmetric vector fields in the plane, Math. USSR Sbornik 48 (1984) 463
\bibitem{tak1}F. Takens, Singularities of vector fields, Publ. Math. IHES 43 (1974) 47–100
\bibitem{tak2}F. Takens, Forced oscillations and bifurcations, Comm. Math. Inst. Rijksuniv. Utrecht 3 (1979) 1–59.
\bibitem{stro}S. H.  Strogatz, Nonlinear Dynamics and Chaos (Perseus Books Publishing, 1994)
\end{thebibliography}
\end{document}